\documentclass[sigconf,screen,10pt,nonacm]{acmart}

\usepackage{booktabs} 

\usepackage{caption}
\usepackage[flushleft]{threeparttable}
\usepackage{enumitem}
\usepackage{multirow}
\usepackage{amsmath}
\usepackage{amssymb}
\usepackage{hypcap}
\usepackage{tabularx}
\usepackage[linesnumbered,ruled,vlined,norelsize]{algorithm2e}
\usepackage[aboveskip=2pt,belowskip=-7pt]{caption}
\usepackage[caption=false,labelformat=simple,captionskip=-1pt]{subfig}

\hypersetup{pdfstartview=FitH}
\graphicspath{{figs/}{figs/spread/}{figs/trace/}{figs/distrib/}}
\def\done{\hspace*{\fill} {$\square$}}

\newcommand{\etal}{{et al.}\xspace}
\newcommand{\ie}{{i.e.,}\xspace}
\newcommand{\figurenames}{{Figures}\xspace}
\newcommand{\RR}{{mRR}\xspace}

\newcommand{\OPIMC}{\textsc{OPIM-C}\xspace}
\newcommand{\OPIMF}{\textsc{TRIM}\xspace}
\newcommand{\OPIMFB}{\textsc{TRIM-B}\xspace}
\newcommand{\ASM}{\textsc{ASTI}\xspace}
\newcommand{\ASMB}{\textsc{ASTI-$b$}\xspace}
\newcommand{\ASMT}{\textsc{ASTI-2}\xspace}
\newcommand{\ASMF}{\textsc{ASTI-4}\xspace}
\newcommand{\ASME}{\textsc{ASTI-8}\xspace}
\newcommand{\AdaptSM}{\textsc{AdaptIM}\xspace}
\newcommand{\TEUC}{\textsc{ATEUC}\xspace}

\newcommand{\MINTSS}{\textsc{Greedy-Mintss}\xspace}
\newcommand{\AdaptIM}{\textsc{AdaptIM-1}\xspace}

\newcommand{\E}{\mathbb{E}}
\newcommand{\R}{\mathcal{R}}

\newcommand{\e}{{\ensuremath{\mathrm{e}}}}
\newcommand{\ratio}{\ensuremath{1-1/\e}}

\newcommand{\indeg}{\ensuremath{\mathrm{indeg}}}
\newcommand{\NP}{\ensuremath{\mathrm{NP}}}
\newcommand{\DTIME}{\ensuremath{\mathrm{DTIME}}}
\newcommand{\poly}{\operatorname{poly}}
\newcommand{\polylog}{\operatorname{polylog}}

\newcommand{\OPTT}{{\operatorname{OPT}}}

\newcommand{\EW}{\E_{\Phi\sim\Omega}}
\newcommand{\EWi}[1]{\E_{\Phi\sim\Omega_{#1}}}

\newcommand{\spara}[1]{\vspace{2mm}\noindent\textbf{#1.}}
\newcommand{\abs}[1]{\lvert#1\rvert}

\newcommand{\eat}[1]{}
\newcommand{\report}[1]{}
\renewcommand{\report}[1]{#1}

\begin{document}
\title{Efficient Approximation Algorithms for Adaptive Seed Minimization}
\titlenote{A short version of the paper appeared in 2019 International Conference on Management of Data (SIGMOD '19), June 30--July 5, 2019, Amsterdam, Netherlands. ACM, New York, NY, USA, 18 pages. \url{https://doi.org/10.1145/3299869.3319881}}

\author{Jing Tang}
\orcid{0000-0002-0785-707X}
\authornote{These authors have contributed equally to this work.}
\affiliation{%
	\department{Dept. of Ind. Syst. Engg. and Mgmt.}
	\institution{National University of Singapore}
}
\email{isejtang@nus.edu.sg}
\author{Keke Huang}
\authornotemark[2]
\affiliation{%
	\department{School of Comp. Sci. and Engg.}
	\institution{Nanyang Technological University}
}
\email{khuang005@ntu.edu.sg}
\author{Xiaokui Xiao}
\affiliation{%
	\department{School of Computing}
	\institution{National University of Singapore}
}
\email{xkxiao@nus.edu.sg}
\author{Laks V.S. Lakshmanan}
\affiliation{%
	\department{Department of Computer Science}
	\institution{University of British Columbia}
}
\email{laks@cs.ubc.ca}

\author{Xueyan Tang}
\affiliation{%
	\department{School of Computer Science and Engineering}
	\institution{Nanyang Technological University}
}
\email{asxytang@ntu.edu.sg}

\author{Aixin Sun}
\affiliation{%
	\department{School of Computer Science and Engineering}
	\institution{Nanyang Technological University}
}
\email{axsun@ntu.edu.sg}

\author{Andrew Lim}
\affiliation{%
	\department{Dept. of Ind. Syst. Engg. and Mgmt.}
	\institution{National University of Singapore}
}
\email{isealim@nus.edu.sg}

\begin{abstract}
    As a dual problem of influence maximization, the seed minimization problem asks for the minimum number of seed nodes to influence a required number $\eta$ of users in a given social network $G$. Existing algorithms for seed minimization mostly consider the {\it non-adaptive} setting, where all seed nodes are selected in one batch without observing how they may influence other users. 
	
	In this paper, we study seed minimization in the {\it adaptive} setting, where the seed nodes are selected in several batches, such that the choice of a batch may exploit information about the actual influence of the previous batches. We propose a novel algorithm, {\it \ASM}, which addresses the adaptive seed minimization problem in $O\Big(\frac{\eta \cdot (m+n)}{\varepsilon^2}\ln n \Big)$ expected time and offers an approximation guarantee of $\frac{(\ln \eta+1)^2}{(1 - (1-1/b)^b)  (\ratio)(1-\varepsilon)}$ in expectation, where $\eta$ is the targeted number of influenced nodes, $b$ is size of each seed node batch, and $\varepsilon \in (0, 1)$ is a user-specified parameter. To the best of our knowledge, \ASM is the first algorithm that provides such an approximation guarantee without incurring prohibitive computation overhead. With extensive experiments on a variety of datasets, we demonstrate the effectiveness and efficiency of \ASM over competing methods.
\end{abstract}

%
%



\maketitle

\begin{sloppy}
\vspace{-0.02in}
\section{Introduction}\label{sec:introduction}
Social networks are becoming increasingly popular for people to discuss and share their thoughts and comments towards public topics. Based on the established relations among individuals, ideas and opinions can be spread over social networks via a word-of-mouth effect. To exploit this effect for advertising, advertisers often provide free samples of their products to selected social network users, in exchange for those users to promote those products and create a cascade of influence to other users. In such a setting, advertisers might want to know the minimum number of free samples required to be given away, so as to draw sufficient attention. Goyal~\etal~\cite{Goyal_MINTSS_2013} are the first to formulate this problem as a \textit{seed minimization} problem, which asks for the minimum number of {\it seed nodes} (i.e., users who receive free samples) needed to influence at least a required number $\eta$ of users, taking into account the randomness in the influence propagation process. 

\eat{If a topic attracts sufficiently high attention (\ie more than a certain \textit{threshold}), it is often labeled as a \textit{hot topic}~\cite{Gladwell_spread_2006}. Based on this observation, advertisers who need to promote their products might want to know the minimum number of free samples required to be given away in order to draw sufficient attention, so as to trigger a hot topic or to  create a sizable cascade of information about the product.}

Existing work on seed minimization mostly focuses on the \textit{non-adaptive} setting \cite{Goyal_MINTSS_2013,Zhang_SMPCG_2014,Han_SM_2017}, which requires that all seed nodes should be selected in one batch without observing the actual influence of any node, i.e., no randomness in the influence propagation process can be removed until all seed nodes are fixed. As a consequence of the non-adaptiveness, these solutions may return a seed set that fails to influence at least $\eta$ nodes in the actual propagation process, or may select an excessive number of seed nodes that generate an actual influence spread much larger than required.

To address the above issues, Vaswani and Lakshmanan~\cite{Vaswani_adapIM_2016} propose to consider seed minimization under the {\it adaptive} setting, where (i) the seed nodes are selected one by one, and (ii) before selecting the $i$-th seed node, the actual influence of the first $i - 1$ seed nodes can be observed, i.e., we may optimize the choice of the $i$-th seed node to influence those users that have not been influenced by the previous $i - 1$ seed nodes. Such an adaptive strategy ensures that (i) the seed set returned always achieves the required number of influenced users (since the actual influence of each seed node is known after selection), and (ii) the number of seed nodes would not be excessive (because we can stop selecting seed nodes as soon as the targeted influence is achieved). We note that similar adaptive approaches have also been adopted by other practical problems, such as influence maximization \cite{Yadav_DIM_2016}, sensor placement \cite{Asadpour_sensor_2008}, active learning \cite{Chen_activeLearning_2013}, and object detection \cite{Chen_detection_2014}.

To our knowledge, the only existing solution for adaptive seed minimization is by Vaswani and Lakshmanan~\cite{Vaswani_adapIM_2016}. As we discuss in Section~\ref{sec:existing-solu}, however, the solution in \cite{Vaswani_adapIM_2016} requires that the expected influence of any seed set should be estimated with extremely high accuracy, which results in prohibitive computation overhead. Furthermore, the solution does not provide any non-trivial approximation guarantee, due to an ineffective approach used to select each seed node under the adaptive setting. Therefore, it remains an open problem to devise efficient approximation algorithms for adaptive seed minimization.

In this paper, we address the above open problem with \ASM, a novel framework tailored for \textit{adaptive seed minimization}. The key idea of \ASM is to adaptively choose the seed node with the maximum expected {\it truncated influence spread} in each round of seed selection. Specifically, given a diffusion model $M$ that captures the uncertainty of influence propagation in $G$, we consider the set $\Omega$ of all possible \textit{realizations}, each of which represents a possible scenario of influence propagation among the nodes in $G$. For each possible realization $\phi\in\Omega$, the influence spread of a seed set $S$, denoted as $I_{\phi}(S)$ is the number of nodes influenced by $S$, while the truncated influence spread of $S$ is defined as $\Gamma_{\phi}(S) = \min\{\eta, I_{\phi}(S)\}$. We consider $\Gamma_{\phi}(S)$ instead of $I_{\phi}(S)$ because, intuitively, the extra influence spread beyond $\eta$ is useless for fulfilling the requirement on influence. (In fact, as we show in Section~\ref{sec:existing-solu}, the extra influence spread may even lead to incorrect choice of seed nodes, and hence, it has to be ignored.)
 
When developing algorithms under the \ASM framework, the key challenge that we face is the design of methods to accurately estimate a seed set $S$'s expected truncated influence spread over a given set of possible realizations. We show that existing methods \cite{Huang_SSA_2017,Nguyen_DSSA_2016,Tang_IMM_2015,Tang_OPIM_2018,Tang_reverse_2014,Borgs_RIS_2014} for estimating un-truncated influence spread cannot be applied in our truncated setting, since they are unable to take into account the effect of truncation by $\eta$. Motivated by this, we propose a novel sampling method based on the concept of \textit{multi-root reverse reachable (\RR)} sets, and prove that our method provides non-trivial guarantees in terms of the efficiency and accuracy of truncated influence estimation. Building upon this sampling method, we develop {\it \OPIMF}, an algorithm for maximizing truncated influence spread with a provable approximation guarantee of $(\ratio)(1-\varepsilon)$. We show that instantiating \ASM using \OPIMF leads to strong theoretical guarantees for adaptive seed minimization, and \OPIMF can be extended into a batched version \OPIMFB that selects a batch of $b$ nodes in each round, so as to accelerate seed selection. 

In summary, we make the following contributions: 
\begin{itemize}[topsep=2mm, partopsep=0pt, itemsep=1mm, leftmargin=18pt]
	\item {\bf \ASM, a general framework.} We analyze the characteristics of adaptive seed minimization, based on which we propose a general framework \ASM tailored for the problem.
	\item{\bf \RR-set, a novel sampling method.} \ASM requires accurate estimation of truncated influence spreads, for which the existing sampling methods are either inefficient or ineffective. To address this challenge, we propose a novel sampling method, \RR, which is able to estimate the truncated influence spread in a cost-effective manner.
	\item {\bf \OPIMF, an efficient algorithm for truncated influence maximization.} A key step of \ASM is to identify a set of nodes with the maximum {\it expected} truncated influence spread, for which we propose the \OPIMF algorithm based on \RR-sets. With a rigorous theoretical analysis, we show that \ASM instantiated by \OPIMF returns a $\frac{(\ln \eta+1)^2}{(\ratio)(1-\varepsilon)}$-approximate solution for adaptive seed minimization with expected time complexity of $O\big(\frac{\eta \cdot(m+n)}{\varepsilon^2}\ln n\big)$.
	\item {\bf \OPIMFB, the batched version of \OPIMF.} For further performance gain, we extend \OPIMF into a batched version \OPIMFB that selects seed nodes in a predefined batch size $b$ in each round. \ASM instantiated by \OPIMFB provides an approximation guarantee of $\frac{(\ln \eta+1)^2}{(1-(1-1/b)^b)(\ratio)(1-\varepsilon)}$ with the same time complexity as \OPIMF.
	\item {\bf An extensive set of experiments.} We experimentally evaluate \ASM instantiated by \OPIMF and \OPIMFB against the state-of-the-art non-adaptive algorithm \TEUC \cite{Han_SM_2017}, and show that (i) our solutions are much more effective in minimizing the number of seed nodes needed and ensuring that the required influence spread is achieved, and (ii) our solutions are able to efficiently handle social networks with millions of nodes and edges.
\end{itemize}
\section{Preliminaries}\label{sec:preliminaries}

This section formally defines the problem of adaptive seed minimization, and reviews the existing solutions. Table~\ref{tbl:notation} summarizes the notations that are frequently used. For ease of exposition, our discussions focus on the {\it independent cascade (IC)} model \cite{Kempe_maxInfluence_2003}, which is one of the most widely adopted propagation models in the literature. But we note that our algorithms can be easily extended to other propagation models, such as the linear threshold model \cite{Kempe_maxInfluence_2003} and the topic-aware models \cite{Barbieri_TAI_2012}.

\begin{table}[!t]
	\centering
	\captionsetup{aboveskip=6pt,belowskip=0pt}
	\caption{Frequently used notations.}
	\label{tbl:notation}
    \begin{small}
	\begin{tabular}{|c|m{2.38in}|}\hline
		{\bf Notation}	   & \multicolumn{1}{c|}{\bf Description} 										\\  \hline
		$G=(V, E)$							   & a graph $G$ with node set $V$ and edge set $E$                				\\  \hline
		$n,m$                          & the number of nodes and edges in $G$         \\  \hline
		$\eta$								   & the threshold for the targeted number of nodes to be activated				\\	\hline
		$I(S),\E[I(S)]$                        & the spread of a seed set $S$ and its expectation               			\\  \hline
		$\Gamma(S),\E[\Gamma(S)]$              & the truncated spread of $S$ and its expectation               				\\  \hline
		$G_i=(V_i,E_i)$						   & the $i$-th residual graph, where $G_1=G$               						\\  \hline
		$n_i,m_i$                          & the number of nodes and edges in $G_i$         \\  \hline
		$\eta_i$                          	   & the shortfall in activating $\eta$ nodes in the $i$-th round, \ie~$\eta_i=\eta-(n-n_i)$	\\  \hline
		$I(S\mid S_{i-1})$					   	   & the marginal spread of $S$ on top of $S_{i-1}$, \ie~the spread of $S$ in $G_i$		\\  \hline
		$\Gamma(S\mid S_{i-1})$			   		   & the marginal truncated spread of $S$ on top of $S_{i-1}$, \ie~$\Gamma(S\mid S_{i-1})=\min\{I(S\mid S_{i-1}),\eta_i\}$	\\  \hline
		$\tilde{\Gamma}(S\mid S_{i-1})$			& a binary estimator with value $\eta_i$ if $S\cap R\neq\emptyset$ and $0$ otherwise \\ \hline
		$R,\R$											& a random \RR-set and a set of \RR-sets \\ \hline
		$\Lambda_{\R}(v)$					   & the number of \RR-sets in $\R$ covered by $v$                     \\  \hline
		$v^\ast,v^\diamond,v^\circ$                        & the optimal node maximizing $\Lambda_{\R}(v)$, $\E[\tilde{\Gamma}(v\mid S_{i-1})]$, and $\E[\Gamma(v\mid S_{i-1})]$, respectively 									\\  \hline
		$\OPTT_i$ 								&  the optimum of $\E[\tilde{\Gamma}(v\mid S_{i-1})]$, i.e., $\OPTT_i=\max_v\E[\tilde{\Gamma}(v\mid S_{i-1})]$    \\ \hline
		$\phi, \Phi, \Omega$	         	   & a specific realization, a random realization, and the realization space	\\  \hline
		$\pi, \pi^\ast$				   & a random policy, and an optimal policy					\\  \hline
	\end{tabular}
	\end{small}
 	\vspace{2mm}
\end{table}

\subsection{Influence Propagation and Realization}\label{sec:IP-realization}

Let $G$ be a social network with a node set $V$ and a directed edge set $E$, where $|V|=n$ and $|E|=m$. For any edge $\langle u,v\rangle \in E$, we refer to $u$ as an \textit{incoming neighbor} of $v$, and $v$ as an \textit{outgoing neighbor} of $u$. Each edge $e=\langle u,v\rangle$ is associated with a \textit{propagation probability} $p(e)\in(0,1]$. We refer to such a social network as a \textit{probabilistic social network}. 

Given a node set $S \subseteq V$, the influence propagation initiated by $S$ under the \textit{independent cascade (IC)} model \cite{Kempe_maxInfluence_2003} is modeled as a discrete-time stochastic process as follows. At time slot $t_0$ (the subscript indicates the index of the time slot), all nodes in $S$ are activated while all other nodes are inactive. Suppose that node $u$ is first activated at slot $t_i$, then $u$ has \textit{one} chance to activate each outgoing neighbor $v$ with the probability $p(u,v)$ at slot $t_{i+1}$, after which $u$ remains active. This influence propagation process continues until no more inactive nodes can be activated. As to the \textit{linear threshold (LT)} model, it demands that for each node $v \in G$, the propagation probabilities of all edges ending at $v$ sum up to no more than $1$. With a given node set $S$, LT model works in a similar discrete-time stochastic procedure as follows. At time slot $t_0$, each node $v \in G$ is assigned with a threshold $\lambda_v$ sampled uniformly from $[0,1]$, and only nodes in $S$ are activated. At time slot $t_i$, we check all inactive node $u$ of its incoming edges from activated neighbors that if the sum of their propagation probabilities is no smaller than $\lambda_u$. If it is, then $u$ is activated; otherwise $u$ remains inactive. This influence propagation process terminates once there is no further node activated. Let $I(S)$ be the total number of active nodes in $G$ when the influence propagation terminates. We refer to $S$ as the \textit{seed set}, and $I(S)$ as the \textit{spread} of $S$.

Alternatively, the influence propagation process can also be described by the \textit{live edge} procedure \cite{Kempe_maxInfluence_2003}. Specifically, for each edge $e\in E$, we independently flip a coin of head probability $p(e)$ to decide whether the edge $e$ is \textit{live} or \textit{blocked} to generate a sample of influence propagation. All the blocked edges are removed and the remaining graph is referred to as a \textit{realization} of the probabilistic social network $G$, denoted as $\phi$. Note that there are $2^{m}$ distinct possible realizations. Let $\Omega$ be the set of all possible realizations (\ie the sample space) such that $\abs{\Omega}=2^{m}$, and $\Phi\sim\Omega$ denote that $\Phi$ is a realization randomly sampled from $\Omega$. Given a realization $\phi \in \Omega$, the spread of any seed set $S \subseteq V$ under $\phi$ is the total number of nodes that are reachable from $S$, denoted as $I_\phi(S)$. Thus, for any seed set $S$, its \textit{expected spread} $\E[I(S)]$ is defined as
\begin{equation}
\E[I(S)]:=\EW[I_\Phi(S)]=\sum_{\phi \in \Omega}I_\phi(S) \cdot p(\phi),
\end{equation}
where $p(\phi)$ is the probability for realization $\phi$ to occur. In other words, the expected spread of $S$ is the (weighted) average spread over all the realizations in $\Omega$.

\begin{figure*}[!t]
	\centering
	\captionsetup[subfigure]{captionskip=3pt}
	\subfloat[A social graph $G$]{\includegraphics[width=0.19\linewidth]{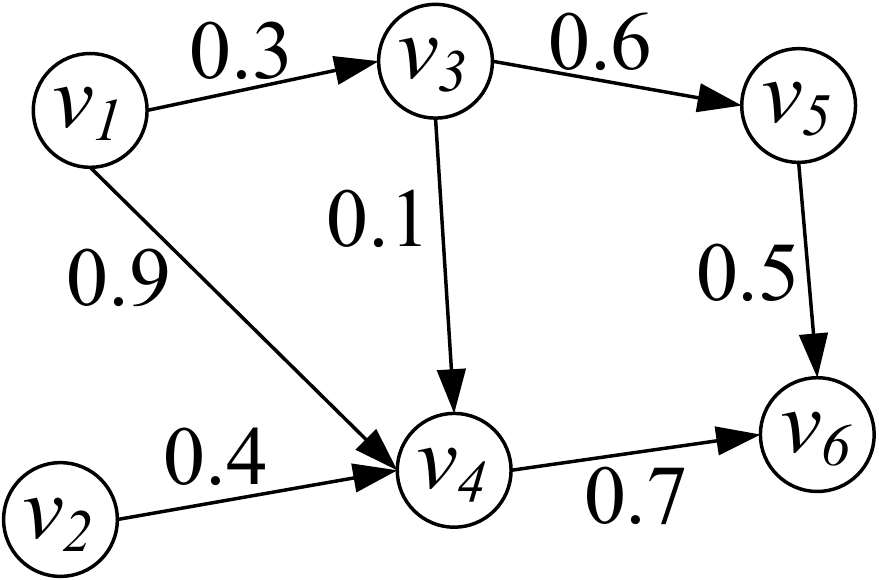}\label{subfig:a}}\hfill
	\subfloat[A possible realization $\phi$]{\includegraphics[width=0.19\linewidth]{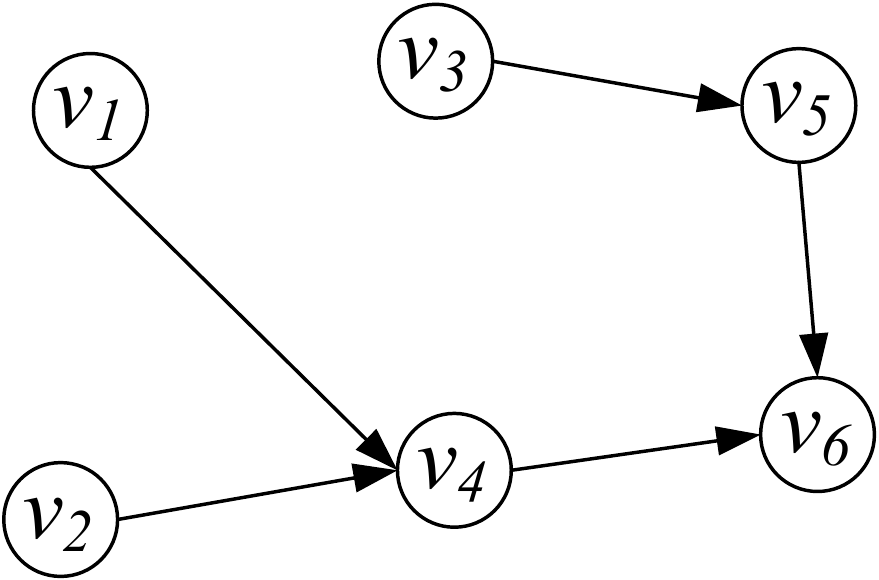}\label{subfig:b}}\hfill
	\subfloat[$v_1$ as the first seed]{\includegraphics[width=0.19\linewidth]{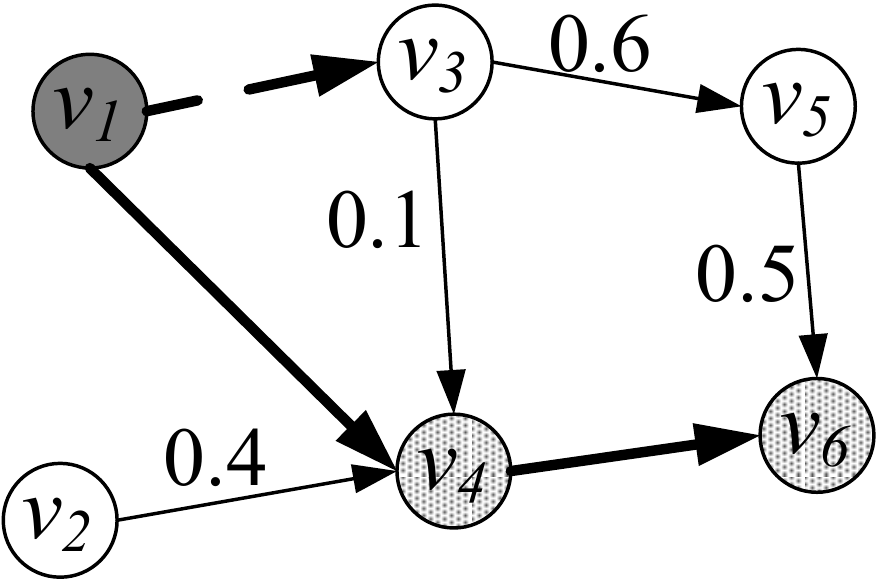}\label{subfig:c}}\hfill
	\subfloat[$v_3$ as the second seed]{\includegraphics[width=0.19\linewidth]{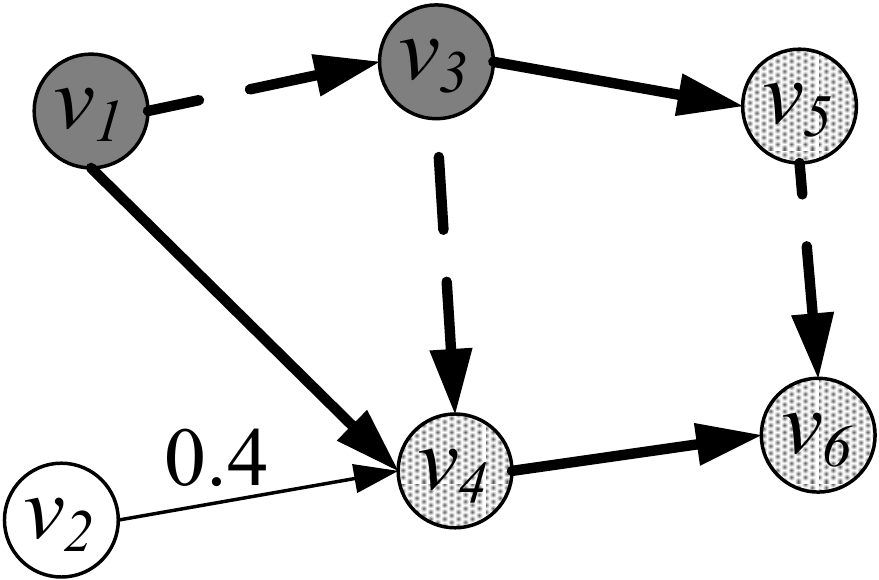}\label{subfig:d}}
	\caption{An adaptive seed minimization process.}\label{fig:ASM-process}
\end{figure*}

\subsection{Adaptive Seed Minimization}\label{sec:asm-definition}

Given a probabilistic social network $G=(V,E)$ and a threshold $\eta \in [1,n]$, the seed minimization problem aims to select a minimum number of seed nodes to influence at least $\eta$ nodes. In the conventional ``non-adaptive'' setting, seed minimization requires selecting a node set $S$ such that $\E[I(S)] \ge \eta$, without any knowledge of realization that would occur in the actual influence propagation process.
As a consequence, the selected $S$ may influence fewer than $\eta$ nodes for some realizations or much more than $\eta$ nodes for some other realizations, both of which are undesirable scenarios. 

Meanwhile, the \textit{adaptive} strategy (\ie~a recursive \textit{select-observe-select} procedure) has been shown to be more effective than the \textit{non-adaptive} (\ie just select based on model) strategy in many real-world applications \cite{Asadpour_sensor_2008,Chen_activeLearning_2013,Chen_detection_2014}. Specifically, an adaptive strategy first \textit{selects} a node $u$ from graph $G$, and then \textit{observes} the set of nodes activated by choosing node $u$ as a seed node. Based on this observation, the strategy would choose the next node as one that could influence as many {\it currently inactive} nodes as possible. This procedure is carried out in an recursive manner, until at least $\eta$ active nodes are observed.  

\figurename~\ref{fig:ASM-process} illustrates the adaptive strategy. \figurenames~\ref{subfig:a} and \ref{subfig:b} show a social graph $G$ and one possible realization $\phi$ of $G$, respectively. Let $\eta=4$ and $\phi$ be the actual realization of influence propagation (which is unknown  apriori). \figurename~\ref{subfig:c} indicates that we first select node $v_1$ (in dark gray) as a seed node. Note that node $v_1$ influences nodes $v_4$ and $v_6$ (in light gray), with each bold (resp.\ dashed) arrow denoting a successful (resp.\ failed) step of influence. In addition, the thin arrows in \figurenames~\ref{subfig:c}--\ref{subfig:d} correspond to influence attempts which are not yet revealed. Since the number of nodes influenced by $v_1$ is less than $\eta$, we continue to select the second seed node. \figurename~\ref{subfig:d} shows that we select $v_3$, which results in a total of $5$ active nodes, reaching the threshold $\eta$. Then, the adaptive seed selection process terminates.

In this paper, we aim to study seed selection strategies (referred to as {\it policies}) for \textit{adaptive seed minimization (ASM)}, which is formally defined as follows:
\begin{definition}[Adaptive Seed Minimization]\label{def:ASM}
	Given a probabilistic social graph $G=(V,E)$ and a threshold $\eta \in [1,n]$, the adaptive seed minimization problem aims to identify a policy $\pi$ that minimizes the expected number of seed nodes required to achieve an influence spread of at least $\eta$ on possible realizations $\phi \in \Omega$, \ie~\[\min_\pi \E[\abs{S(\pi,\phi)}]~\text{subject to}~I_\phi(S(\pi,\phi))\geq\eta~\text{for all}~\phi,\] where $S(\pi,\phi)$ is the seed set selected by $\pi$ under realization $\phi$ and $\E[\abs{S(\pi,\phi)}]=\sum_{\phi \in \Omega} \abs{S(\pi,\phi)}\cdot p(\phi)$.
\end{definition}

Note that when the propagation probability of every edge in $G$ is $1$, ASM reduces to the deterministic version of seed minimization, which is shown to be NP-hard \cite{Goyal_MINTSS_2013}. Therefore, finding an optimal policy for ASM is also NP-hard.

\eat{\spara{Remark} Note that in the general case that $p(e)\in [0,1]$, under some realization, the optimal policy $\pi^\ast$ might not derive the optimal solution. This is because the realization is unknown to $\pi^\ast$ in advance, which means the cover set of each node remains indeterminate. Therefore $\pi^\ast$ selects the node that maximizes the total spread in expectation by taking all possible outcomes into consideration. For example, consider a simple graph with two nodes $u$ and $v$, where $p(\langle u, v \rangle)=0.5$ and $p(\langle v, u \rangle)=0.4$. Suppose $\eta=2$. It is easy to verify that the optimal policy $\pi^\ast$ will select $u$ first. If $u$ activates $v$, the algorithm stops. Otherwise, it will continue selecting $v$. However, under the realization that $\langle u, v \rangle$ is blocked and $\langle v, u \rangle$ is alive, which emerges with probability $0.2$, $\pi^\ast$ selects $\{u,v\}$ but the optimal solution should be $\{v\}$.}

\subsection{Truncated Influence Spread}
Note that, in ASM, the influence spread in excess of the threshold $\eta$ has no value. Accordingly, we introduce the notion of {\it truncated influence spread} as follows.

\begin{definition}[Truncated Influence Spread] 
Given a seed set $S$ and a threshold $\eta$, the \textit{truncated influence spread} $\Gamma_\phi(S)$ of $S$ under a realization $\phi$ is the smaller one between $I_\phi(S)$ and $\eta$, \ie
\begin{equation}
\Gamma_\phi(S):=\min\{I_\phi(S),\eta\}.
\end{equation} 
\end{definition}

Recall that ASM requires considering the influence spreads of nodes when the actual influence of some other nodes has been observed. Therefore, we also introduce the notion of {\it marginal} truncated influence spread as follows. Let $V_1 = V$ and $G_1 = G$. Let $V_i$ be the subset of nodes that remain inactive after round $(i-1)$, $G_i$ be the subgraph of $G$ induced by $V_i$. We refer to $G_i$ as the $i$-th {\it residual graph}. For example, in \figurename~\ref{fig:ASM-process}, after round $1$, only nodes $v_2, v_3,v_5$ remain inactive, so $V_2 = \{v_2, v_3,v_5\}$ and $G_2 = (V_2, E_2)$ denotes the induced subgraph containing the thin edge $\langle v_3, v_5 \rangle$. 

Let $S_{i}$ be the set of nodes selected as seeds by a policy in the first $i$ rounds. Similar to the definition of $\Omega$, we denote $\Omega_i$ as the set of all possible realizations in the $i$-th round. Then, for a node set $S\subseteq V_i$, we define the marginal spread $I_\phi(S\mid S_{i-1})$ as the additional spread that $S$ provides on top of $S_{i-1}$ under realization $\phi\in\Omega_i$, and define truncated marginal spread $\Gamma_\phi(S\mid S_{i-1})$ accordingly, \ie 
\begin{align}
&I_\phi(S\mid S_{i-1}):=I_\phi(S\cup S_{i-1})-I_\phi(S_{i-1}),\\
\text{and}\quad&\Gamma_\phi(S\mid S_{i-1}):=\Gamma_\phi(S\cup S_{i-1})-\Gamma_\phi(S_{i-1}).
\end{align}
Note that $I_\phi(S\mid S_{i-1})$ is exactly the influence spread of $S$ in the residual graph $G_i$ under realization $\phi$.

Let $n_i=\abs{V_i}$ be the number of nodes in $G_i$, \ie~ $I_\phi(S_{i-1})=n-n_i$  nodes have been activated by the end of round $i-1$, based on the partial realization revealed so far. Define $\eta_i=\eta-(n-n_i)$. This is the amount by which the policy falls short of the target $\eta$ in the beginning of round $i$. Before reaching the threshold $\eta$, \ie~$\Gamma_\phi(S_{i-1})=I_\phi(S_{i-1})<\eta$, we can rewrite $\Gamma_\phi(S\mid S_{i-1})$ as
\begin{align}
	\Gamma_\phi(S\mid S_{i-1})
	&=\min\{I_\phi(S\cup S_{i-1}),\eta\}-I_\phi(S_{i-1})\nonumber\\
	&=\min\{I_\phi(S\mid S_{i-1}),\eta_i\}.
\end{align}
Then, $\Gamma_\phi(S\mid S_{i-1})$ can be easily computed in the residual graph $G_i$. For brevity, we define $\Gamma_\phi(v\mid S_{i-1}):=\Gamma_\phi(\{v\}\mid S_{i-1})$ for a singleton node set $\{v\}$.

Finally, we define the \textit{expected marginal truncated spread} $\Delta (v\mid S_{i-1})$ as
\begin{equation}\label{eqn:marginal-spread}
\Delta (v\mid S_{i-1}):=\EWi{i}[\Gamma_\Phi(v\mid S_{i-1})].
\end{equation}
In other words, the expected marginal truncated spread of a node $v$ is defined based on the ``lift'' in the expected number of active nodes that $v$ brings on top of previously selected seeds, over all realizations consistent with what has been observed in previous rounds.

\begin{figure*}[!t]
	\centering
	\captionsetup[subfigure]{captionskip=3pt}
	\subfloat[A social graph $G$]{\includegraphics[width=0.18\linewidth]{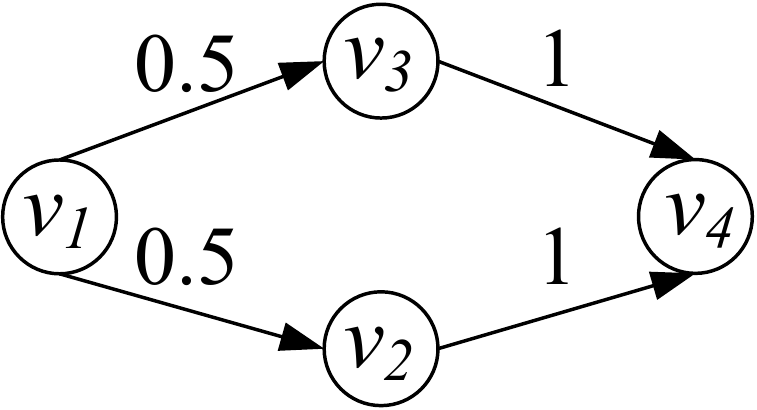}\label{subfig:ta}}\hfill
	\subfloat[Realization $\phi_1$]{\includegraphics[width=0.18\linewidth]{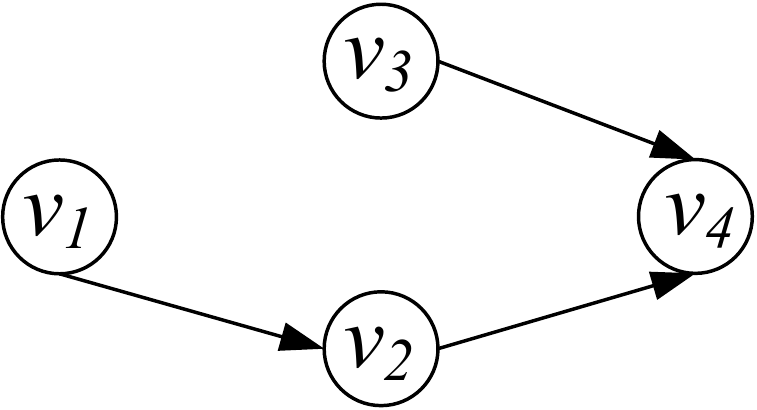}\label{subfig:tb}}\hfill
	\subfloat[Realization $\phi_2$]{\includegraphics[width=0.18\linewidth]{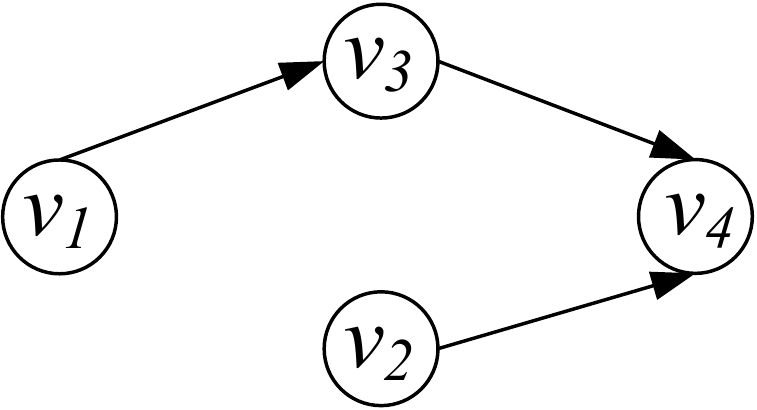}\label{subfig:tc}}\hfill
	\subfloat[Realization $\phi_3$]{\includegraphics[width=0.18\linewidth]{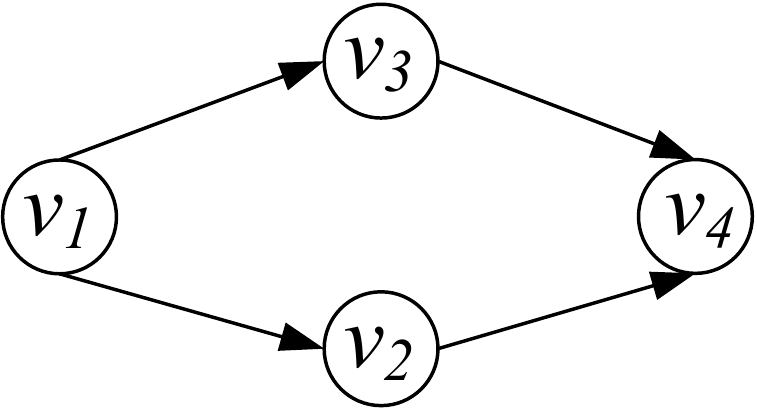}\label{subfig:td}}\hfill
	\subfloat[Realization $\phi_4$]{\includegraphics[width=0.18\linewidth]{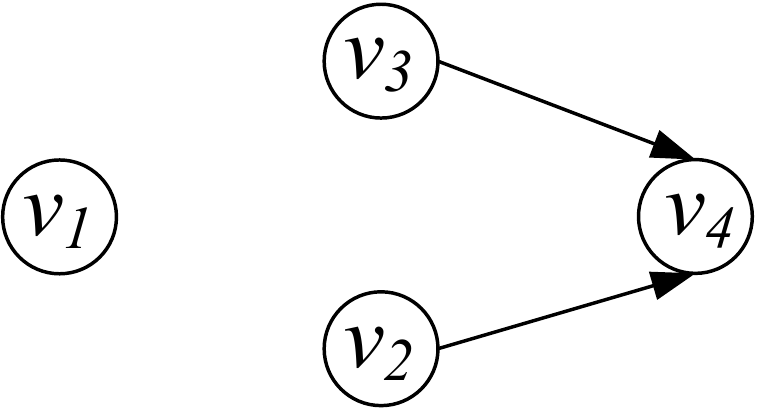}\label{subfig:te}}
	\caption{A social graph and all of its possible realizations.}\label{fig:possible-relization}
\end{figure*}

\subsection{Existing Solutions}\label{sec:existing-solu}

Golovin and Krause~\cite{Golovin_adaptive_2017} study the \textit{adaptive stochastic minimum cost coverage} problem, which can be regarded as a variant of ASM in the case where there exists an {\it oracle} that accurately reports the expected marginal truncated spread for any given seed set. They propose to adopt a greed policy as follows. First, select the node $s_1$ with the largest expected truncated spread, \ie~$\Delta(s_1\mid S_0)\geq\Delta(v\mid S_0)$ for all $v\in V$. Then, observe the actual nodes that are activated by $s_1$ during the stochastic process, and remove them from $G$ to induce the residual graph $G_2$. After that, identify the node $s_2$ with the maximum expected marginal truncated spread $\Delta(s_2\mid S_1)$ in the residual graph $G_2$. This process continues, such that each round selects the node with the largest expected marginal truncated spread, until we observe that no less than $\eta$ nodes have been influenced.

Golovin and Krause~\cite{Golovin_adaptive_2017} show that the above greedy policy returns a $(\ln \eta +1)^2$-approximate solution to the optimum.\footnote{Golovin and Krause claim that the approximation guarantee is $(\ln \eta +1)$ in an earlier version of their work \cite{Golovin_adaptive_2011}, but point out that the proof has gaps in a revised version \cite{Golovin_adaptive_2017}. Whether the logarithmic bound holds is an interesting open problem.} This approximation guarantee, however, does not lead to a practical algorithm for the ASM problem, because (i) it requires the help from an oracle to exactly identify the node with the maximum expected marginal truncated spread in each round, but (ii) computing the exact expected spread of any node set is \#P-hard~\cite{Chen_MIA_2010}.

Motivated by this observation, Vaswani and Lakshmanan~\cite{Vaswani_adapIM_2016} attempt to extend Golovin and Krause's method by replacing the oracle with an spread estimator with bounded errors. In particular, they assume that for any node set $S$, the estimation ${\E}[\tilde{I}(v\mid S)]$ of the marginal gain $\E[I(v\mid S)]:=\E[I(S\cup\{v\})]-\E[I(S)]$ should satisfy
\begin{equation}\label{eqn:relatve-error}
\alpha^\bot \E[I(v\mid S)]\leq {\E}[\tilde{I}(v\mid S)]\leq \alpha^\top\E[I(v\mid S)],
\end{equation}
where ${\alpha^\top}/{\alpha^\bot}$ denotes the multiplicative error in calculating the marginal gains. Unfortunately, this requirement on the spread estimation is so stringent that no existing methods for influence estimation could fulfill the requirement without incurring prohibitive estimation overhead. To explain, suppose that the expected marginal spread $\E[I(v\mid S)]$ of a node $v$ on top of $S$ is small. In that case,  Equation~\eqref{eqn:relatve-error} would only allow a trivial amount of estimation error, which is rather difficult to achieve by existing methods for spread estimation. 

In addition, the algorithm in \cite{Vaswani_adapIM_2016} attempts to select the node with the largest marginal spread in each round, instead of the node with the maximum marginal {\it truncated} spread. As a consequence, even when there exists an efficient estimator that provides highly accurate spread estimation, the algorithm in \cite{Vaswani_adapIM_2016} would still fail to achieve the type of approximation guarantee in \cite{Golovin_adaptive_2017}, which the theoretical analysis in \cite{Golovin_adaptive_2017} is based on the notion of truncated spreads. We illustrate this issue with an example.

\begin{example}
Consider \figurename~\ref{fig:possible-relization}(a), which shows a social graph $G$ with four nodes and four directed edges. The number on each edge indicates the propagation probability of the edge. $G$ has four possible realizations $\phi_1$, $\phi_2$, $\phi_3$, and $\phi_4$ in total, as shown in \figurenames~\ref{subfig:tb}--\ref{subfig:te}. Each realization has an equal probability of $0.25$ to happen. Assume that $\eta=2$. Then, the expected spread of node $v_1$ is $\E[I(v_1)]=0.25\times(3+3+4+1)=2.75$, which is larger than that of the other three nodes. Thus, when the \textit{vanilla} expected spread is adopted as the measure, node $v_1$ will be selected as the first seed node. On realizations $\phi_1$, $\phi_2$, and $\phi_3$, $v_1$ is qualified to influence at least $\eta = 2$ users. However, there is a probability of $0.25$ that $\phi_4$ happens, in which case $v_1$ can only influence itself, and hence, one additional seed node is required. Overall, $2\times0.25+1\times(1-0.25)=1.25$ seed nodes are selected in expectation. 

Now observe that the expected {\it truncated} spread of nodes $v_1$, $v_2$, $v_3$, and $v_4$ are $1.75$, $2$, $2$, and $1$, respectively. Therefore, when the expected truncated spread is adopted as the measure, either $v_2$ or $v_3$ is selected as the first seed node, which can influence $2$ users under all four realizations. This demonstrates that, for ASM, choosing nodes based on expected truncated spreads is more effective than that based on vanilla expected spreads. \done
\end{example}


In recent work~\cite{Han_AIM_2018}, Han \etal study the problem of {\it adaptive influence maximization}, which also considers the adaptive setting, but aims to identify a predefined number of $k$ seed nodes that could influence the maximum number of users in $G$ in expectation. At the first glance, it may seem that we can modify the adaptive influence maximization algorithms to solve the adaptive seed minimization problem, in the same way that existing work \cite{Goyal_MINTSS_2013} transforms non-adaptive influence maximizing algorithms to address non-adaptive seed minimization. This approach, however, does not work because the algorithm in \cite{Han_AIM_2018} is designed based on {\it vanilla} expected marginal spreads. Instead, ASM requires considering {\it truncated} expected marginal spreads, as we previously discussed. As a consequence, the algorithm in \cite{Han_AIM_2018} cannot be adopted in our setting.
\section{Our Solution}\label{sec:solution}

\subsection{Algorithmic Framework}\label{subsec:framework}

\begin{algorithm}[!t]
\begin{small}
	\caption{{\ASM}$(G,\eta)$}
	\label{alg:framework}
	\KwIn{an input graph $G$ and a theshlod $\eta$.}
	\KwOut{a seed set $S$ such that $\Gamma(S)=\eta$.}
	initialize $S\leftarrow\emptyset$, $\Gamma(S)\leftarrow0$, $i\leftarrow 1$\;
	\Repeat {$\Gamma(S)\ge \eta$\label{algline:ASM-stop}}
	{			
		select $s_i$ from $V_i$ such that $\Delta(s_i\mid S_{i-1})\geq \alpha\Delta(v\mid S_{i-1})$ for all $v\in V_i$\;\label{algline:ASM-select}
		observe the influence of $s_i$ in $G_i$\;\label{algline:ASM-observe}
		insert $s_i$ into $S$ and increase $\Gamma(S)$ acoordingly\;\label{algline:ASM-update-S}
		remove all nodes in $G_i$ that are influenced by $s_i$, and denote the resulting graph as $G_{i+1}$\;\label{algline:ASM-update-G}
		$i\leftarrow i+1$\;\label{algline:ASM-update-i}
	}
	\Return $S$\;
\end{small}
\end{algorithm}

We propose a general framework, referred to as \ASM, to address the ASM problem. Algorithm~\ref{alg:framework} shows the details. Given a probabilistic social graph $G$ and a threshold $\eta$, \ASM aims to return a seed set $S$ such that $\Gamma(S)\ge\eta$, where $\Gamma(S)$ is the truncated influence spread of $S$ (\ie~the smaller one of the threshold $\eta$ and the number of active nodes influenced by $S$). In a nutshell, \ASM iteratively (i) selects the node to maximize the \textit{expected marginal truncated spread} (Line~\ref{algline:ASM-select}), (ii) observes the newly influenced nodes (Line~\ref{algline:ASM-observe}), and then (iii) updates the corresponding information (Lines~\ref{algline:ASM-update-S}--\ref{algline:ASM-update-i}). The process stops when at least $\eta$ nodes are activated (Line~\ref{algline:ASM-stop}).

The key step of \ASM is \textit{truncated influence maximization} that targets at identifying a node to maximize the expected marginal truncated spread (Line~\ref{algline:ASM-select}). If an $\alpha$-approximate solution for truncated influence maximization is obtained in each round (Line~\ref{algline:ASM-select}), \ASM provides a non-trivial approximation guarantee, as shown in the following theorem.
\begin{theorem}\label{thm:approx}
	Suppose $\pi$ is an $\alpha$-approximate greedy policy, for some $\alpha \in (0,1]$, i.e., for any $G_i$ and $v\in V_i$, it selects a node $s_i$ satisfying 
	\begin{equation}
	\Delta (s_i\mid S_{i-1})\geq \alpha\Delta (v\mid S_{i-1}).
	\end{equation}
	Then $\pi$ achieves an approximation ratio of $\frac{(\ln \eta+1)^2}{\alpha}$ to the optimal adaptive seed minimization policy.
\end{theorem}

The proof\footnote{The formal proofs of all theoretical results are given in Appendix~\ref{appendix:proof}.} of Theorem~\ref{thm:approx} is based on adaptive submodular optimization~\cite{Golovin_adaptive_2017}. Theorem~\ref{thm:approx} requires that the policy should be an $\alpha$-approximate greedy one with respect to the expected marginal truncated spread $\Delta(v\mid S_{i-1})$. The challenge for designing such an $\alpha$-approximate greedy policy lies in how to develop a proper sampling method for estimating the truncated influence spread.

\eat{\spara{Comparison with Non-adaptive Algorithms} Goyal \etal \cite{Goyal_MINTSS_2013} propose a similar greedy algorithm \MINTSS for non-adaptive seed minimization. The differences in algorithm design are three-fold: (i) we only consider the realizations that are consistent with the observations made in previous rounds while they always use all realizations for each round, (ii) \ASM looks at marginal truncated spread instead of marginal spread by \MINTSS, and (iii) \ASM stops when at least $\eta$ nodes are activated whereas \MINTSS only requires to influence $(\eta-\xi)$ users that falls short of the required threshold by a predefined shortfall parameter $\xi > 0$. In addition, we note that the approximation guarantees of both algorithms are positively correlated to $\ln(\eta/\eta_l)$, where $l$ is the final number of nodes selected and $\eta_l$ is the shortfall in achieving $\eta$ in the $l$-th round. To provide theoretical guarantees, $\eta_l$ must be lower bounded by some constant factors. In \ASM, $\eta_l=\eta-I(S_{l-1})$ is an integer such that $\eta_l\geq 1$. In \MINTSS, however, $\eta_l=\eta-\E[I(S_{l-1})]$ is a real value which can be arbitrarily close to $0$. To ensure a lower bound on $\eta_l$, \MINTSS immediately stops once more than $(\eta-\xi)$ nodes are activated such that $\eta_l\geq \xi$. As a consequence, \ASM manages to achieve a constant approximation, whereas \MINTSS can only yield a bicriteria approximation (\ie~the solution achieves an expected spread of $(\eta-\xi)$, which is smaller than the required threshold $\eta$).}

\subsection{Truncated Influence Maximization}\label{sec:naive-rrset}

According to Theorem~\ref{thm:approx}, in order to provide the theoretical guarantee, the algorithm is supposed to identify a node whose truncated marginal spread is an $\alpha$-approximation  to the maximum truncated marginal spread in each round. At a first glance, it seems that we can utilize Borgs~{\etal}'s reverse influence sampling method \cite{Borgs_RIS_2014}. Unfortunately, in what follows, we show that Borgs~{\etal}'s sampling method \cite{Borgs_RIS_2014} fails to estimate the truncated influence spread accurately.
 
Specifically, Borgs~\etal~\cite{Borgs_RIS_2014} propose to generate random \textit{reverse reachable (RR)} sets for influence maximization. Compared with the Monte-Carlo simulation \cite{Kempe_maxInfluence_2003}, RR-sets can dramatically accelerate the seed selection process while  retaining the same approximation guarantees for influence maximization \cite{Borgs_RIS_2014}. In particular, a random RR-set of $G$ is generated by first selecting a node $v\in V$ uniformly at random, and then taking the nodes that can reach $v$ in a random realization. Evidently, a random RR-set is a subgraph of the corresponding random realization $\Phi$, which is generated by performing a reverse breadth first search (BFS) on $\Phi$ starting from the random node $v$. A random RR-set $R$ is an unbiased spread estimator, \ie~for any seed set $S$, \[\E[I(S)]=n\cdot\Pr[R\cap S\neq\emptyset].\] Unfortunately, RR-sets fail to estimate truncated influence spread accurately. Intuitively, the expectation of this estimator for truncated influence spread of $S$ is \[\eta\cdot\Pr[R\cap S\neq\emptyset]=\frac{\eta}{n}\cdot \E[I(S)]=\frac{\eta}{n}\sum_{\phi \in \Omega}I_\phi(S) \cdot p(\phi).\] Recall that the true expected truncated influence spread is \[\E[\Gamma(S)]=\sum_{\phi \in \Omega}\Gamma_\phi(S) \cdot p(\phi).\] Obviously, for any $\phi\in\Omega$, unless $I_\phi(S)=n$, \[\frac{\eta}{n}\cdot I_\phi(S)<\min\{I_\phi(S),\eta\}=\Gamma_\phi(S).\] Specifically, consider the case that $I_\phi(S)\leq \eta$ for all $\phi$. Then, this estimator is biased with a discount ${\eta}/{n}$, which is extremely inaccurate when $\eta\ll n$. In practice, $\eta$ is likely to be a fraction of $n$, since even a set of ten thousand seed nodes has been found to influence less than half population on many datasets \cite{Nguyen_DSSA_correct_2016}. These facts indicate that RR-sets are highly biased for estimating truncated influence spread. As a consequence, the state-of-the-art algorithms \cite{Huang_SSA_2017,Nguyen_DSSA_2016,Tang_IMM_2015,Tang_OPIM_2018,Tang_reverse_2014} for influence maximization that utilize RR-sets \cite{Borgs_RIS_2014} cannot provide theoretical guarantees for truncated influence maximization. In turn, this means that these algorithms cannot be fashioned to solve ASM with approximation guarantees. To address this issue, we propose a novel sampling approach that generates \textit{multi-root reverse reachable (\RR)} sets which can estimate the truncated influence spread efficiently and effectively. The algorithm utilizing \RR-sets is referred to as \OPIMF\footnote{\underline{TR}uncated \underline{I}nfluence \underline{M}aximization.}. We rigorously show that \OPIMF can provide strong theoretical guarantees for truncated influence maximization and thus \ASM instantiated with \OPIMF is guaranteed to approximate ASM within a constant ratio.

\subsection{Multi-Root Reverse Reachable Set}\label{sec:muti-root-rrset}
If we generate $n$ correlated RR-sets such that (i) they start from $n$ distinct nodes, and (ii) the materialization of each edge is consistent in all the RR-sets, then merging these RR-sets (with duplicates removed) as well as the edge statuses forms a realization sample. Based on this observation, if we generate $k$ $(k<n)$ correlated RR-sets using the same rule, then merging them as a $k$-root RR-set is likely to estimate the truncated influence spread more accurately compared against a vanilla RR-set. To explain how \textit{multi-root reverse reachable (\RR)} set works, we first introduce its definition.
\begin{definition}[Random \RR-set]
	Let $\Phi$ be a random realization of $G$ sampled from the realization space and $K$ be a size-$k$ node set selected uniformly at random from $V$. A random \RR-set is the set of nodes in $\Phi$ that can reach $K$. (That is, for each node $v$ in the \RR-set, there is a directed path in $\Phi$ from $v$ to some node in $K$.)
\end{definition}

By definition, the key difference between an \RR-set and an RR-set is that the former has multiple roots whereas the latter has one single root only. Similar to the generation of RR-sets, a random \RR-set can be generated by:
\begin{enumerate}
	\item Choose a set of $k$ nodes $K\subseteq V$ uniformly at random;
	\item Perform a stochastic reverse breadth first search (BFS) that starts from $K$ and follows the incoming edges of each node. Insert into $R$ all nodes that are traversed during the stochastic BFS. 
\end{enumerate}

A natural question is how to decide the size of $k$ for truncated spread estimation? The setting of $k$ yields a tradeoff between efficiency and accuracy in that a larger $k$ provides more accurate estimation but takes more computational resources. Through the aforementioned analysis of RR-set, we find that the high-efficiency of RR-set comes from its ``binary'' property. In particular, a random RR-set $R$ estimates the influence spread of any node set $S$ as $n$ if $R\cap S\neq 0$, and as $0$ otherwise. To avoid maintaining the edge statuses, our \RR-set estimator shall retain this binary property. That is, it estimates the truncated influence spread of $S$ as $\eta$ if and only if $S$ intersects this \RR-set, and as $0$ otherwise. For a given $k$-RR-set $R$, if a node $v\in R$, then $v$ can reach at least one of the $k$ starting nodes. Then, its influence spread is estimated to be at least $n/k$ and thus its estimated truncated influence spread is at least $\min\{n/k,\eta\}$. By setting $n/k\geq \eta$, the estimated truncated influence spread is $\eta$.

On the other hand, to improve the accuracy, $k$ should be set as large as possible. So we choose $k=n/\eta$. However, ${n}/{\eta}$ is not an integer in general. To address this issue, we adopt a randomized rounding approach. To generate a \RR-set, we randomly choose a set $K$ of nodes such that its size $k$ equals $\lfloor\frac{n}{\eta}\rfloor+1$ with probability $\frac{n}{\eta}-\lfloor\frac{n}{\eta}\rfloor$, and equals $\lfloor\frac{n}{\eta}\rfloor$ otherwise. Then, the expectation of $k$ is ${n}/{\eta}$. However, we note that when $k=\lfloor\frac{n}{\eta}\rfloor+1$, the possible value of the estimated truncated influence spread is no longer binary (i.e., $0$ or $\eta$). To address such a new challenge, we define an estimator $\tilde{\Gamma}(S)$ as $\tilde{\Gamma}(S)=\eta$ if and only if $S\cap R\neq\emptyset$, and $\tilde{\Gamma}(S)=0$ otherwise. At the first glance, it seems that the relationship between $\E[{\Gamma}(S)]$ and $\E[\tilde{\Gamma}(S)]$ is unclear. Fortunately, the following theorem shows that under the above setting of $k$ such that $\E[k]={n}/{\eta}$, the ratio of ${\E[\tilde{\Gamma}(S)]}$ and ${\E[{\Gamma}(S)]}$ is in the range of $[\ratio,1]$.
\begin{theorem}\label{thm:mrr-relative-random}
	Let $k^\bot=\lfloor\frac{n}{\eta}\rfloor$ and $r={n}/{\eta}-k^\bot$ be the integer and fractional part of ${n}/{\eta}$, respectively. For any \RR-set, if we sample $k$ nodes such that $k=k^\bot+1$ with probability $r$ and $k=k^\bot$ otherwise, then
	\begin{equation}\label{eqn:mrr-relative-random}
	(\ratio)\E[{\Gamma}(S)]\leq \E[\tilde{\Gamma}(S)] \leq \E[{\Gamma}(S)].
	\end{equation}
\end{theorem}

Theorem~\ref{thm:mrr-relative-random} states that $\tilde{\Gamma}$ is a biased but sufficiently accurate estimator of the expected truncated influence spread $\E[\Gamma(S)]$. In fact, this estimator also works for any residual graph $G_i$. Specifically, let $\tilde{\Gamma}(S\mid S_{i-1})$ be the estimated truncated spread of $S$ in $G_i$ with respect to $\eta_i$, the lowered target corresponding to graph $G_i$. Recall that $\eta_i = \eta - (n-n_i)$. We have the following corollary.
\begin{corollary}\label{corollary:mrr-relative-random-general}
	In the residual graph $G_i$, let $k^\bot=\lfloor\frac{n_i}{\eta_i}\rfloor$ and $r={n_i}/{\eta_i}-k^\bot$ be the integer and fractional part of ${n_i}/{\eta_i}$, respectively. For each \RR-set, if we sample $k$ nodes such that $k=k^\bot+1$ with probability $r$ and $k=k^\bot$ otherwise, then
	\begin{equation}\label{eqn:mrr-relative-random-general}
	\!\!\!(\ratio)\E[{\Gamma}(S\mid S_{i-1})]\leq \E[\tilde{\Gamma}(S\mid S_{i-1})] \leq \E[{\Gamma}(S\mid S_{i-1})].
	\end{equation}
	Furthermore, for any two sets $S,S^\prime\subseteq V_i$, it holds that
	\begin{equation}\label{eqn:mrr-relative-ratio}
		\frac{\E[{\Gamma}(S\mid S_{i-1})]}{\E[{\Gamma}(S^\prime\mid S_{i-1})]}\geq (\ratio)\frac{\E[\tilde{\Gamma}(S\mid S_{i-1})]}{\E[\tilde{\Gamma}(S^\prime\mid S_{i-1})]}.
	\end{equation}
\end{corollary}

Now, we can construct a $(\ratio)(1-\varepsilon)$-approximate greedy policy using the estimator $\tilde{\Gamma}$ built upon \RR-sets.

\spara{Remark} It is worth pointing out that our randomized rounding approach for choosing $k$ is critical for achieving the above approximation bound. Specifically, if we fix $k$ to be $\lfloor\frac{n}{\eta}\rfloor$, following the proof methodology of Theorem~\ref{thm:mrr-relative-random}, we may derive that the ratio of ${\E[\tilde{\Gamma}(S)]}$ to ${\E[{\Gamma}(S)]}$ will be in the range of $[1-1/\sqrt{\e},1]$. On the other hand, if we fix $k$ to be $\lfloor\frac{n}{\eta}\rfloor+1$, the ratio of ${\E[\tilde{\Gamma}(S)]}$ to ${\E[{\Gamma}(S)]}$ will be in the range of $[\ratio,2]$. Both settings yield much coarser bounds than our setting that uses a smart randomized rounding approach.

\subsection{The Design of \OPIMF}\label{sec:trim-alg}
\setlength{\textfloatsep}{0.5em}
\begin{algorithm}[!t]
\begin{small}
	\caption{{\OPIMF}$(G_i,\varepsilon)$}
	\label{alg:OPIM-E}
	\KwIn{Graph $G_i$ and error threshold $\varepsilon$.}
	\KwOut{A $(\ratio)(1-\varepsilon)$-approximate solution $v^\ast$ for truncated influence maximization.} 
	$\delta\leftarrow {\varepsilon}/{\big(100(\ratio)(1-\varepsilon)\eta_i\big)}$, and $\hat{\varepsilon}\leftarrow{{99\varepsilon}/{(100-\varepsilon)}}$\;\label{algline:OPIM-E-correct-error}
	$\theta_{\max}\leftarrow{2n_i\big(\sqrt{\ln({6}/{\delta})}+\sqrt{\ln n_i+\ln({6}/{\delta})}\big)^2}\cdot \hat{\varepsilon}^{-2}$\;
	$\theta_{\circ}\leftarrow{\theta_{\max}\cdot\hat{\varepsilon}^2}/{n_i}$\;
	$T\leftarrow\lceil \log_2\frac{\theta_{\max}}{\theta_{\circ}}\rceil+1$\;
	$a_1\leftarrow\ln(3T/\delta)+\ln n_i$, and $a_2\leftarrow\ln(3T/\delta)$\;
	generate a set $\R$ of $\theta_{\circ}$ random \RR-sets\;\label{algline:OPIM-E-RRsets2}	
	\For{$t\leftarrow1$ \KwTo $T$\label{algline:OPIM-E-loop}}{
		find $v^\ast\leftarrow\arg\max_{v\in V_i}\Lambda_\R(v)$\; \label{algline:OPIM-E-greedy}
		$\Lambda^l(v^\ast)\leftarrow(\sqrt{\Lambda_\R(v^\ast)+{2a_1}/{9}}-\sqrt{{a_1}/{2}} )^2- {a_1}/{18}$\;\label{algline:OPIM-E-lower}
		$\Lambda^u(v^\circ)\leftarrow(\sqrt{\Lambda_\R(v^\ast)+{a_2}/{2}}+\sqrt{{a_2}/{2}} )^2$\;\label{algline:OPIM-E-upper}
		\lIf{$\frac{\Lambda^l(v^\ast)}{\Lambda^u(v^\circ)}\geq 1-\hat{\varepsilon}$ {\bf or} $t=T$\label{algline:OPIM-E-stop}}{\Return{$v^\ast$}\label{algline:OPIM-E-return}}
		double the size of $\R$\;\label{algline:OPIM-E-doubleRRsets}
	}
\end{small}
\end{algorithm}
Algorithm~\ref{alg:OPIM-E} presents the details of \OPIMF that can return a $(\ratio)(1-\varepsilon)$-approximate solution for truncated influence maximization for any input graph $G_i$ and error threshold $\varepsilon$. \OPIMF is similar in spirit to \OPIMC which is the state-of-the-art algorithm for influence maximization \cite{Tang_OPIM_2018}. Specifically, \OPIMC uses two disjoint groups of random RR-sets, among which one group is used to derive the solution and the other is used to verify its quality. We customize \OPIMF by utilizing one group of \RR-sets, which would be more efficient for selecting a singleton seed set as pointed out in \cite{Huang_SSA_2017}. In a nutshell, \OPIMF starts from a small number of \RR-sets and iteratively increases the \RR-set number until a satisfactory solution is identified. Next, we discuss the details of \OPIMF.

In the \RR-set sampling stage (Lines~\ref{algline:OPIM-E-RRsets2} and \ref{algline:OPIM-E-doubleRRsets}), each \RR-set is started from a random set $K$ of nodes whose size $k$ is an independent random number. Recall that $k$ is $\lfloor\frac{n_i}{\eta_i}\rfloor+1$ with probability $\frac{n_i}{\eta_i}-\lfloor\frac{n_i}{\eta_i}\rfloor$ and $\lfloor\frac{n_i}{\eta_i}\rfloor$ otherwise. Given a set $\R$ of random \RR-sets, we say that a node $v$ \textit{covers} a \RR-set $R\in\R$ if $v\in R$, and we define the \textit{coverage} of $v$ in $\R$, denoted as $\Lambda_\R(v)$, as the number of \RR-sets in $\R$ that are covered by $v$. Based on the \RR-sets generated, \OPIMF identifies the node $v^{\ast}\in V_i$ that covers the largest number of \RR-sets in $\R$ (Line~\ref{algline:OPIM-E-greedy}).  Let $v^{\circ}$ be the optimal node such that $\Delta(v^{\circ}\mid S_{i-1}) = \max_{v \in V_i} \Delta(v\mid S_{i-1})$. Then, $\Lambda_\R(v^{\circ})$ is bounded by $\Lambda_\R(v^\ast)$. According to Lemma~\ref{lemma:concentration-ept} in Appendix \ref{appenix:inequality}, with high probability, $\Lambda^l(v^\ast)$ (Line~\ref{algline:OPIM-E-lower}) is a lower bound on the expected coverage of $v^\ast$ in $\R$, which indicates that
\begin{equation}\label{eqn:tis-lower}
\frac{\eta_i\Lambda^l(v^\ast)}{\abs{\R}}\leq \E[\tilde{\Gamma}(v^\ast\mid S_{i-1})].
\end{equation}
Similarly, with high probability, $\Lambda^u(v^\circ)$ (Line~\ref{algline:OPIM-E-upper}) is an upper bound on the expected coverage of $v^\circ$ in $\R$. Thus,
\begin{equation}\label{eqn:tis-upper}
\frac{\eta_i\Lambda^u(v^\circ)}{\abs{\R}}\geq \E[\tilde{\Gamma}(v^\circ\mid S_{i-1})].
\end{equation} 
In addition, by Equation~\eqref{eqn:mrr-relative-ratio} in Corollary~\ref{corollary:mrr-relative-random-general}, we know that
\begin{equation}\label{eqn:relation-sol-opt}
\frac{\Delta(v^\ast\mid S_{i-1})}{\Delta(v^\circ\mid S_{i-1})}\geq(\ratio)\frac{\E[\tilde{\Gamma}(v^\ast\mid S_{i-1})]}{\E[\tilde{\Gamma}(v^\circ\mid S_{i-1})]}.
\end{equation}
Combining Equations~\eqref{eqn:tis-lower}--\eqref{eqn:relation-sol-opt}, we can derive a quantitative relationship between $\Delta(v^\ast\mid S_{i-1})$ and $\Delta(v^\circ\mid S_{i-1})$ such that with high probability
\begin{equation}
	\Delta(v^\ast\mid S_{i-1})\geq \frac{\Lambda^l(v^\ast)}{\Lambda^u(v^\circ)}\cdot(\ratio) \cdot \Delta(v^\circ\mid S_{i-1}).
\end{equation}
Therefore, the final guarantee is $(\ratio){\Lambda^l(v^\ast)}/{\Lambda^u(v^\circ)}$. Note that in our stopping condition of ${\Lambda^l(v^\ast)}/{\Lambda^u(v^\circ)}\geq 1-\hat{\varepsilon}$ (Line~\ref{algline:OPIM-E-stop}), we use $\hat{\varepsilon}$ (defined in Line~\ref{algline:OPIM-E-correct-error}) to correct the error on Equations~\eqref{eqn:tis-lower} and \eqref{eqn:tis-upper} (with low failure probability). This proves the $(\ratio)(1-\varepsilon)$ approximation ratio of $\Delta(v^\ast\mid S_{i-1})$.

\subsection{Theoretical Analysis}\label{sec:trim-analysis}

Before we proceed to the theoretical analysis, we first present the hardness of ASM.
\begin{lemma}\label{lem:NP-appro}
	Given a probabilistic social network $G=(V,E)$ with $|V|=n$ and a threshold $\eta \in [1,n]$, for any $\xi >0$, adaptive seed minimization cannot be approximated within a ratio of $(1-\xi)\ln \eta$ in polynomial time unless $\NP\subseteq \DTIME(n^{O(\log \log n)})$.
\end{lemma}\vspace{-2mm}
\spara{Approximation Guarantee} Theorem~\ref{thm:approx} indicates that any $\alpha$-approximation greedy policy $\pi$ could achieve an approximation ratio of $\frac{(\ln \eta+1)^2}{\alpha}$. We examine the potential of \OPIMF to serve the role of such a policy. To cope with the randomness of seed selection algorithms (due to sampling), we use the notion of \textit{expected approximation guarantee}, which considers the \textit{average case}. We first obtain the approximation ratio of \OPIMF for each round of seed selection.
\begin{lemma}\label{lem:thrim-alpha}
For the $i$-th round of seed selection in $G_i$, \OPIMF returns a $(\ratio)(1-\varepsilon)$-approximate solution to the optimum.\footnote{Here, $\alpha$-approximation indicates that $\E[\tfrac{1}{\Delta (v^\ast\mid S_{i-1})}]\leq \tfrac{1}{\alpha}\cdot\tfrac{1}{\Delta (v^\circ\mid S_{i-1})}$, which is required by Theorem~\ref{thm:approx} for a randomized algorithm through a detailed check of the proof of Theorem 40 in \cite{Golovin_adaptive_2017}.} 
\end{lemma}

Combining Theorem~\ref{thm:approx} and Lemma~\ref{lem:thrim-alpha}, we obtain the approximation guarantee of \ASM.
\begin{theorem}\label{thm:trim-ratio}
	\ASM with the instantiation of \OPIMF achieves an expected approximation ratio of $\frac{(\ln \eta+1)^2}{(\ratio)(1-\varepsilon)}$.
\end{theorem}\vspace{-2mm}

\spara{Time Complexity} The time complexity of \OPIMF is dominated by the procedure for generating \RR-sets. Intuitively, this is based on (i) how much time is used for generating a random \RR-set, and (ii) how many \RR-sets are generated. In what follows, we show their relationship. In particular, for the $i$-th round of seed selection in $G_i$, let $\OPTT_i$ (resp. $v^\diamond$) be the optimum (resp. optimal node) of $\E[\tilde{\Gamma}(v\mid S_{i-1})]$, i.e., $\OPTT_i=\E[\tilde{\Gamma}(v^\diamond\mid S_{i-1})]=\max_v{\E[\tilde{\Gamma}(v\mid S_{i-1})]}$. (Note that $v^\ast$ maximizes $\Lambda_{\R}(v)$, $v^\diamond$ maximizes $\E[\tilde{\Gamma}(v\mid S_{i-1})]$, and $v^\circ$ maximizes $\Delta(v\mid S_{i-1})$.) We first show the expected time used for generating a random \RR-set in the following lemma. 
\begin{lemma}\label{lemma:time-ept}
	For the $i$-th round of seed selection in $G_i$, the expected time complexity for generating a random \RR-set is $O\big(\frac{\OPTT_i}{\eta_i}m_i\big)$.
\end{lemma}

Now, we present the following lemma that gives the expected number of \RR-sets generated by \OPIMF. The proof is similar to that of \OPIMC \cite{Tang_OPIM_2018}.
\begin{lemma}\label{lemma:sample-ept}
	For the $i$-th round of seed selection in $G_i$, the expected number of \RR-sets \OPIMF generated is $O\big(\frac{\eta_i\ln{n_i}}{\varepsilon^2\OPTT_i}\big)$.\footnote{In general, it is $O\big(\frac{\eta_i\ln{({n_i}/{\varepsilon})}}{\varepsilon^2\OPTT_i}\big)$. Here, we assume that $\varepsilon\in \Omega\big(\frac{1}{\poly(n_i)}\big)$.}
\end{lemma}

Finally, we provide the expected time complexity of \OPIMF in the following lemma.
\begin{lemma}\label{lemma:main-lem}
	For the $i$-th round of seed selection in $G_i$, \OPIMF achieves an expected time complexity of $O\big(\frac{m_i+n_i}{\varepsilon^2}\ln{n_i}\big)$.
\end{lemma}

At the first glance, the expected time complexity of \OPIMF is counterintuitive. In particular, the expected root size of $n_i/\eta_i$ in the $i$-th round is increasing with $i$. It seems that the time complexity of \OPIMF is more likely to increase with $i$. However, Lemma~\ref{lemma:main-lem} just tells us the opposite. This is due to either the residual graph $G_i$ being reduced significantly (Lemma~\ref{lemma:time-ept}) or the \RR-set size being reduced considerably (Lemma~\ref{lemma:sample-ept}). Overall, the time complexity of \OPIMF in each round can be independent of the number of initially selected nodes. There are at most $\eta$ rounds in total, we can derive the expected time complexity of \ASM instantiated with \OPIMF.
\begin{theorem}\label{thm:trim-expected-time}
	\ASM with the instantiation of \OPIMF has an expected time complexity of $O\big(\frac{\eta \cdot(m+n)}{\varepsilon^2}\ln{n}\big)$.
\end{theorem}

\vspace{-2mm}
\section{Extensions}\label{sec:extension}

\OPIMF selects one node in each round until at least $\eta$ users are influenced. Therefore, the seed selection phase in \ASM instantiated by \OPIMF can be quite time consuming due to that the marginal (truncated) spread of a singleton node set is potentially small which may (i) involve in many rounds to achieve the target $\eta$, and (ii) generate a large number of \RR-sets for constructing an $\alpha$-approximate solution in each round.
To mitigate the enormous overhead, we propose a batched version of \OPIMF, referred to as \OPIMFB\footnote{\underline{TR}uncated \underline{I}nfluence \underline{M}aximization in the \underline{B}atched model.} algorithm, to accelerate the node selection process of \ASM.

\subsection{Batched Version of \OPIMF} \label{sec:trimb-alg}

\begin{algorithm}[!t]
\begin{small}
	\caption{{\OPIMFB}$(G_i,\varepsilon, b)$}
	\label{alg:OPIM-Batch}
	\KwIn{Graph $G_i$, error threshold $\varepsilon$, and batch size $b$.}
	\KwOut{A $\rho_b(\ratio)(1-\varepsilon)$-approximate solution $S_b$ with size-$b$ for truncated influence maximization, where $\rho_b=1-(1-1/b)^b$.} 
	$\delta\leftarrow {\varepsilon}/{\big(100(\ratio)(1-\varepsilon)\eta_i\big)}$, and $\hat{\varepsilon}\leftarrow{{99\varepsilon}/{(100-\varepsilon)}}$\;\label{algline:OPIM-B-variable}
	$\theta_{\max} \leftarrow {2n_i\Big(\sqrt{\ln\frac{6}{\delta}}+\sqrt{ \big(\ln\binom{n_i}{b}+\ln\frac{6}{\delta}\big)/\rho_b}\Big)^2}/{(b\hat{\varepsilon}^2)}$\;\label{algline:OPIM-B-Max}
	$\theta_{\circ}\leftarrow{\theta_{\max}\cdot b\hat{\varepsilon}^2}/{n_i}$\;\label{algline:OPIM-B-Min}
	$T\leftarrow\lceil \log_2\frac{\theta_{\max}}{\theta_{\circ}}\rceil+1$\;
	$a_1\leftarrow\ln(3T/\delta)+\ln \binom{n_i}{b}$, and $a_2\leftarrow\ln(3T/\delta)$\;
	generate a set $\R$ of $\theta_{\circ}$ random \RR-sets\;
	\For{$t\leftarrow1$ \KwTo $T$}
	{
		find $S_b\leftarrow$ Greedy$(\R)$\; \label{algline:OPIM-B-greedy}
		$\Lambda^l(S_b)\leftarrow(\sqrt{\Lambda_\R(S_b)+{2a_1}/{9}}-\sqrt{{a_1}/{2}} )^2- {a_1}/{18}$\;\label{algline:OPIM-B-lower}
		$\Lambda^u(S_b^\circ)\leftarrow(\sqrt{\Lambda_\R(S_b)/\rho_b+{a_2}/{2}}+\sqrt{{a_2}/{2}} )^2$\;\label{algline:OPIM-B-upper}
		\lIf{$\frac{\Lambda^l(S_b)}{\Lambda^u(S_b^\circ)}\geq \rho_b(1-\hat{\varepsilon})$ {\bf or} $t=T$\label{algline:OPIM-B-stop}}
		{\Return{$S_b$}}
		double the size of $\R$\;
	}
\end{small}
\end{algorithm}

Algorithm~\ref{alg:OPIM-Batch} shows the details of the \OPIMFB algorithm. \OPIMFB generalizes \OPIMF by selecting a fixed number of $b$ seeds in each round, where $b$ is an input parameter to determine the batch size. Specifically, \OPIMFB first generates a small number of random \RR-sets and then uses a greedy algorithm for maximum coverage  \cite{Vazirani_approxAlg_2003} to identify a size-$b$ seed set $S_b$ to cover \RR-sets with an approximation guarantee of $\rho_b = 1-(1-1/b)^b$ (Line~\ref{algline:OPIM-B-greedy}). If $S_b$ meets the condition (Line~\ref{algline:OPIM-B-stop}), \OPIMFB terminates; otherwise, the number of \RR-sets is doubled until a qualified $S_b$ is derived. Consequently, the approximation ratio of \OPIMFB is $\rho_b(\ratio)(1-\varepsilon)$. Note that when the batch size $b$ is $1$, \OPIMFB degenerates to \OPIMF. 

The major differences in the design between \OPIMFB and \OPIMF are as follows. First, in \OPIMFB, the definitions of variables $\theta_{\max}$ and $\theta_{\circ}$ are involved with $\rho_b$ and $b$ for generalization, as shown in Line~\ref{algline:OPIM-B-Max} and Line~\ref{algline:OPIM-B-Min}, respectively. Second, to obtain the upper bound on the coverage of the optimal solution $S_b^\circ$ in $\R$, the coverage of $\Lambda_\R(S_b)$ is divided by $\rho_b$ (Line~\ref{algline:OPIM-B-upper}). Third, the ratio in the stop condition is updated to be $\rho_b(1-\hat{\varepsilon})$ (Line~\ref{algline:OPIM-B-stop}). 

\subsection{Theoretical Analysis}\label{sec:trimb-analysis}

The theoretical analysis of \OPIMFB can be obtained by generalizing the  properties of \OPIMF.

\spara{Approximation Guarantee} To establish the overall approximation guarantee, we first analyze the approximation ratio of \OPIMFB in each round of seed selection.
\begin{lemma}\label{lem:trimb-alpha}
For the $i$-th round of seed selection in $G_i$, \OPIMFB returns a $\rho_b(\ratio)(1-\varepsilon)$-approximate solution, where $\rho_b=1-(1-1/b)^b$. 
\end{lemma}

Combining Theorem~\ref{thm:approx} and Lemma~\ref{lem:trimb-alpha}, we obtain the approximation guarantee of \OPIMFB. 
\begin{theorem}\label{thm:trimb-ratio}
	\ASM with the instantiation of \OPIMFB achieves an expected approximation ratio of $\frac{(\ln \eta+1)^2}{\rho_b(\ratio)(1-\varepsilon)}$.
\end{theorem}\vspace{-2mm}
\spara{Remark} Note that there exists a gap between the optimal policy in the sequential model and the optimal policy in the batched model, which is known as the {\em adaptivity gap}~\cite{Golovin_adaptive_2017}. Adaptivity gap quantifies the performance difference between the optimal adaptive policy and the optimal non-adaptive policy. To explain, a size-$b$ seed set is selected as a batch ($b \ge 1$) in \OPIMFB without observing the realization of any seed therein. This selection is an non-adaptive process compared to that of $b = 1$ in \OPIMF. As a consequence, there exists an adaptivity gap between the two algorithms if the batch size $b>1$. However, to the best of our knowledge, this adaptivity gap remains unknown in viral marketing applications, which makes it hard to quantify the difference between the optimal policy in the sequential model and that in the batched model. Meanwhile, the existing bound of adaptivity gap of $(\ratio)$ in \cite{Chen_activeLearning_2013} is not applicable to adaptive seed minimization. It holds only if the nodes in social graph $G$ are {\em independent}, which, however, is not true.

\spara{Time Complexity} The time complexity of \OPIMFB depends on three factors: (i) the time for generating a random \RR-set, (ii) the number of \RR-sets generated, and (iii) the time to derive a size-$b$ seed set. The expected time used for generating a random \RR-set is given in Lemma~\ref{lemma:time-ept}. We now show the number of \RR-sets generated.
\begin{lemma}\label{lem:trimb-sample}
	For the $i$-th round of seed selection in $G_i$, the expected number of \RR-sets \OPIMFB generates is $O\Big(\frac{\eta_i\ln{\binom{n_i}{b}}}{\varepsilon^2\OPTT_{b,i}}\Big)$, where $\OPTT_{b,i}$ denotes the maximum expected truncated spread among all the size-$b$ seed sets in $G_i$.
\end{lemma}
On the other hand, the greedy algorithm for identifying the size-$b$ seed set runs in time linear to the total size of its input \cite{Vazirani_approxAlg_2003}, \ie~$\sum_{R \in \R}\abs{R}$. Meanwhile, the total number of \RR-sets examined in all the iterations is within twice of that in the last iteration. According to Wald's equation \cite{Wald_equation_1947}, the expected time complexity of the greedy procedure is $O(\E[\abs{\R}]\cdot\E[\abs{R}])$, which is dominated by that for generating \RR-sets. Consequently, by Lemma~\ref{lemma:time-ept} and Lemma~\ref{lem:trimb-sample}, the expected time used in the $i$-th round of \OPIMFB is $O\big(\tfrac{b(m_i+n_i)\ln n_i}{\varepsilon^2}\big)$.
There are at most $O(\eta/b)$ rounds in total. Based on the analysis above, the expected time complexity of \OPIMFB is given in the following theorem.
\begin{theorem}\label{thm:trimb-expected-time}
	\ASM with the instantiation of \OPIMFB achieves an expected time complexity of $O\big(\frac{\eta \cdot (m+n)}{\varepsilon^2}\ln{n}\big)$.
\end{theorem}
\section{Additional Related Work}\label{sec:related-work}

In Section~\ref{sec:existing-solu}, we have discussed the work \cite{Vaswani_adapIM_2016} most related to ours. In what follows, we survey other relevant work in the literature. 

Influence maximization, as the dual problem of seed minimization, seeks to identify a set of $k$ seed nodes with the maximum expected spread. Domingos and Richardson \cite{Domingos_maxInfluence_2001,Richardson_viralMarketing_2002} are the first to study viral marketing from an algorithmic perspective. After that, Kempe \etal \cite{Kempe_maxInfluence_2003} formulate the influence maximization problem and propose a greedy algorithm that returns $(\ratio-\epsilon)$-approximation for several influence diffusion models, by utilizing Monte Carlo simulations. Subsequently, there has been a large body of research on improved algorithms for influence maximization \cite{Kim_IM_2013,Chen_LDAG_2010,Chen_MIA_2010,Chen_degreeDiscount_2009,Goyal_infMax_2011,Jung_IRIE_2012,Wang_community_2010, Leskovec_CELF_2007,Kempe_maxInfluence_2003,Kempe_influence_2005,Borgs_RIS_2014,Tang_IMM_2015,Tang_OPIM_2018,Tang_reverse_2014,Nguyen_DSSA_2016,Huang_SSA_2017,Tang_infMax_2017,Tang_IMhop_2018,Galhotra_EaSyIM_2016,Arora_benchmark_2017,Cheng_IMRank_2014,Cohen_SKIM_2014,Goyal_CELF_2011,Zhou_UBLF_2013}. Among them, some recent work \cite{Borgs_RIS_2014,Huang_SSA_2017,Nguyen_DSSA_2016,Tang_IMM_2015,Tang_OPIM_2018,Tang_reverse_2014} focuses on algorithms that ensure $(\ratio-\varepsilon)$-approximations by utilizing the reverse influence sampling technique \cite{Borgs_RIS_2014}. 

Seed minimization, which has mainly been studied from the non-adaptive perspective, aims at finding a minimum-size set of seed nodes to achieve a given threshold of expected spread. Chen~\cite{Chen_MINSeed_2009} investigates seed minimization under a variant of the linear threshold model, where each node is assigned with a fixed threshold. Chen shows that the problem cannot be approximated within a ratio of $O(2^{\log^{1-\epsilon}n})$ unless $\NP\subseteq\DTIME(n^{\polylog(n)})$ as the expected spread function under the fixed threshold model is not submodular. After that, Long and Wong \cite{Long_MINSeed_2011} study seed minimization under the widely used independent cascade and linear threshold models. Goyal~\etal~\cite{Goyal_MINTSS_2013} provide a bi-criteria approximation algorithms for seed minimization. Zhang~\etal~\cite{Zhang_SMPCG_2014} then improve the theoretical results by removing the bi-criteria restriction. However, the requirements of these algorithms are either impractical or extremely stringent, which makes these algorithms vastly ineffective in practice. Han~\etal~\cite{Han_SM_2017} propose the \TEUC algorithm for non-adaptive seed minimization by utilizing reverse influence sampling for estimating the spreads of nodes. However, the expected time complexity of the algorithm is unknown, and its worst-case time complexity is prohibitively large. As we show in the experiments, our adaptive algorithm is more effective than these non-adaptive algorithms in terms of the number of seed nodes required. 

Finally, there is a series of recent work \cite{Vaswani_adapIM_2016,Horel_adapSeeding_2015,Badanidiyuru_adapSeeding_2016,Seeman_adapSeeding_2013,Han_AIM_2018} that focuses on adaptive influence maximization. Recall that, as analyzed in Section~\ref{subsec:framework}, to construct approximate solutions for adaptive seed minimization, some approximation algorithms for truncated influence maximization are required. However, the algorithms for adaptive influence maximization generally target at maximizing the influence spread in each round, which cannot provide theoretical guarantees for truncated influence maximization, as we point out in Section~\ref{sec:naive-rrset}. As a consequence, techniques developed for adaptive influence maximization are inapplicable to the adaptive seed minimization problem. In addition, in the case of influence maximization, going adaptive does not really boost the spread significantly, as confirmed by the experiments in \cite{Han_AIM_2018}. However, it shall be observed in our experiments that going adaptive provides a substantial advantage for seed minimization.   

\begin{figure}[!t]
	\centering
	\vspace{-0.15in}
	\subfloat[NetHEPT $\&$ Youtube]{\includegraphics[width=0.45\linewidth]{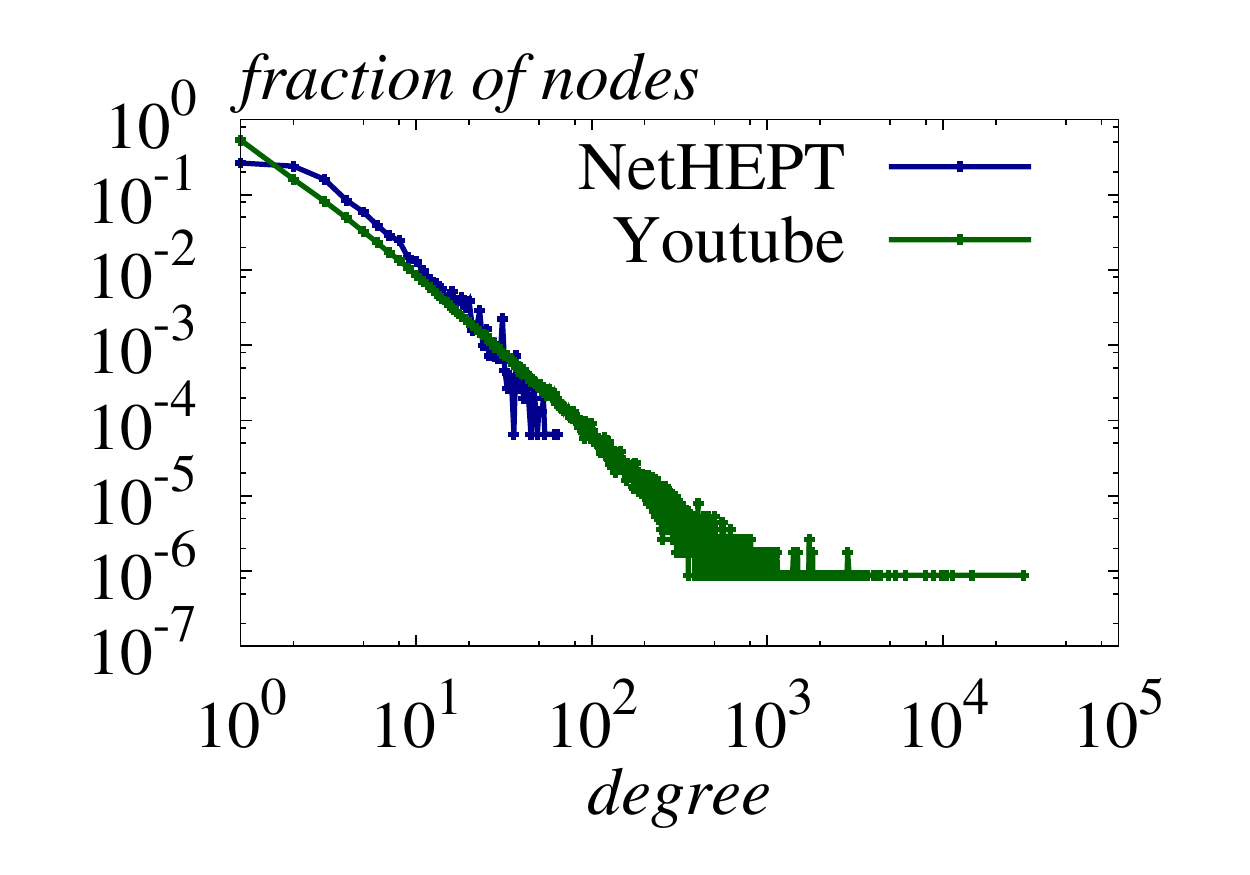}\label{subfig:hep-you}}\hfill
	\subfloat[Epinions $\&$ LiveJournal]{\includegraphics[width=0.45\linewidth]{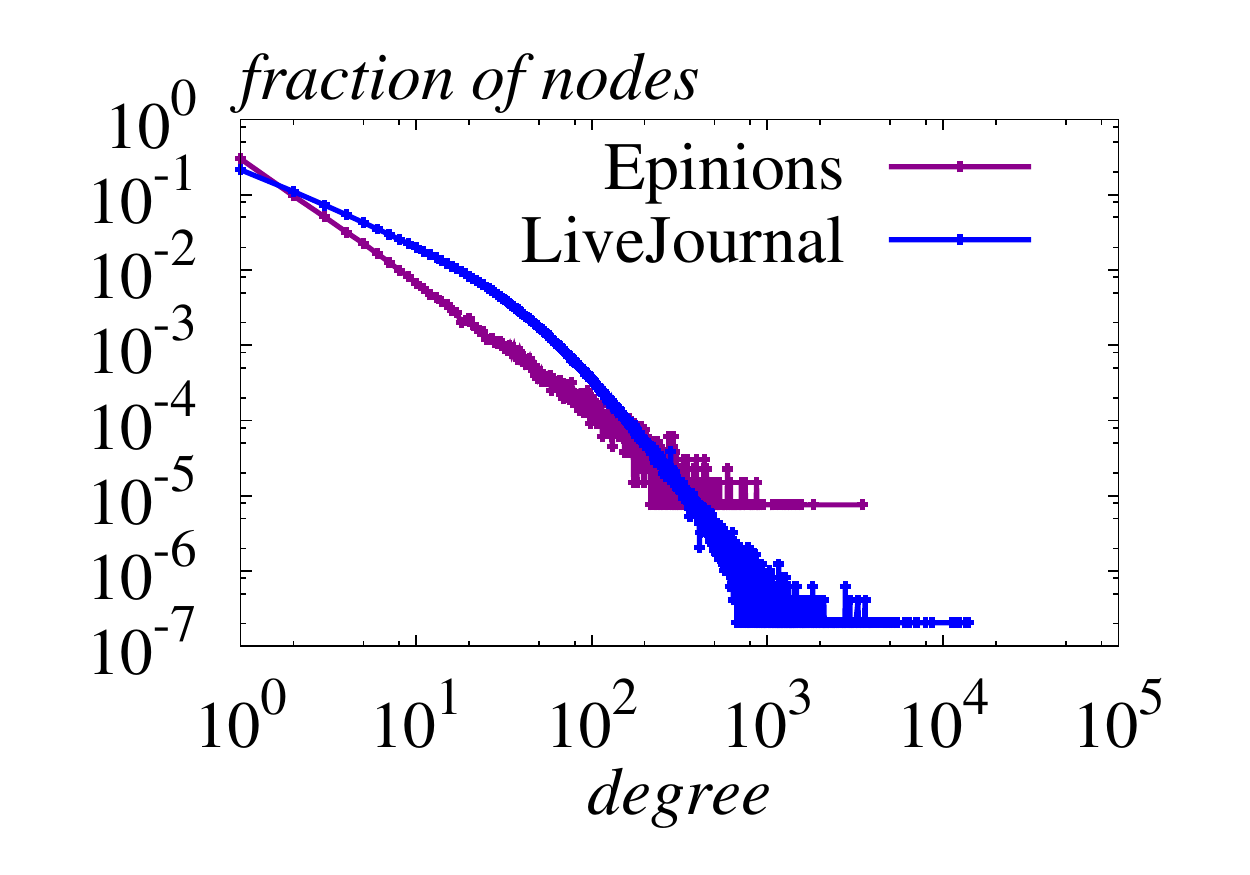}\label{subfig:epi-live}}
	\caption{Degree distribution of tested datasets.} \label{fig:Deg-dis}
	\vspace{2mm}
\end{figure}
\section{Experiments}\label{sec:experiment}

This section evaluates the performance of the proposed algorithms against the state of the art. All the experiments are conducted on a Linux machine with an Intel Xeon 2.6GHz CPU and 64GB RAM. For fair comparison, we first randomly generate $20$ possible realizations for each dataset, and then measure the performance of each algorithm on those $20$ realizations and report the average performance.

\subsection{Experimental Setting}\label{sec:exp-setting}

\begin{table}[!t]
	\centering
	\captionsetup{aboveskip=6pt,belowskip=0pt}
	\caption{Dataset details. ($\boldsymbol{\textrm{K}=10^3, \textrm{M}=10^6}$)} \label{tbl:dataset}
	\vspace{-1mm}
	\setlength{\tabcolsep}{0.3em} 
	\begin{small}
		\begin{tabular} {c|rrrcc}
			\hline
			{\bf Dataset} & \multicolumn{1}{c}{$\boldsymbol{n}$} & \multicolumn{1}{c}{$\boldsymbol{m}$} & \multicolumn{1}{c}{\bf Type}  &
			\multicolumn{1}{c}{\bf Avg.\ deg.} & \multicolumn{1}{c}{\bf  LWCC size}  \\ \hline
			{NetHEPT}       & 15.2K			&  31.4K		& 	undirected		 &	4.18        & 6.80K \\ 		
			{Epinions}		& 132K			&  841K			&  	directed		 &	13.4        & 119K  \\ 
			{Youtube}		& 1.13M			&  2.99M		&  	undirected		 &	5.29        & 1.13M \\ 
			{LiveJournal}   & 4.85M			&  69.0M		&  	directed		 &	28.5        & 4.84M \\ \hline
		\end{tabular}
	\end{small}
\vspace{2mm}
\end{table}

\vspace{-2mm}
\spara{Datasets} The experiments are conducted on four datasets, i.e., {NetHEPT}, {Epinions}, {Youtube}, and {LiveJournal}. {NetHEPT \cite{Chen_degreeDiscount_2009}} represents the academic collaboration networks of "High Energy Physics - Theory" area. The rest of the three are real-life social networks from~\cite{Leskovec_SNAP_2014}. Table~\ref{tbl:dataset} summarizes the details of the four datasets. Note that an undirected edge is transformed into two directed edges. There does exist any isolated node in the four tested datasets. Furthermore, the number of nodes in the largest weakly connected component (LWCC) indicates that nodes are highly interconnected, especially for the three social networks. As shown in \figurename~\ref{fig:Deg-dis}, all the four datasets have a power law degree distribution. The largest dataset that has been used for adaptive seed minimization in the literature contains $75k$ nodes and $500k$ edges~\cite{Vaswani_adapIM_2016}, which is far smaller than {LiveJournal}. To the best of our knowledge, {LiveJournal} with millions of nodes and edges is the largest dataset ever tested in adaptive seed minimization experiments.

\begin{figure*}[!t]
	\centering
	\includegraphics[height=9.5pt]{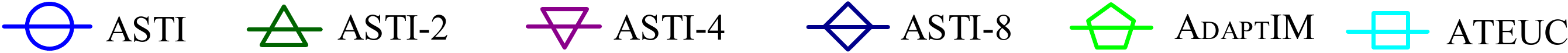}\vspace{-0.2in}\\
	\subfloat[NetHEPT]{\includegraphics[width=0.23\linewidth]{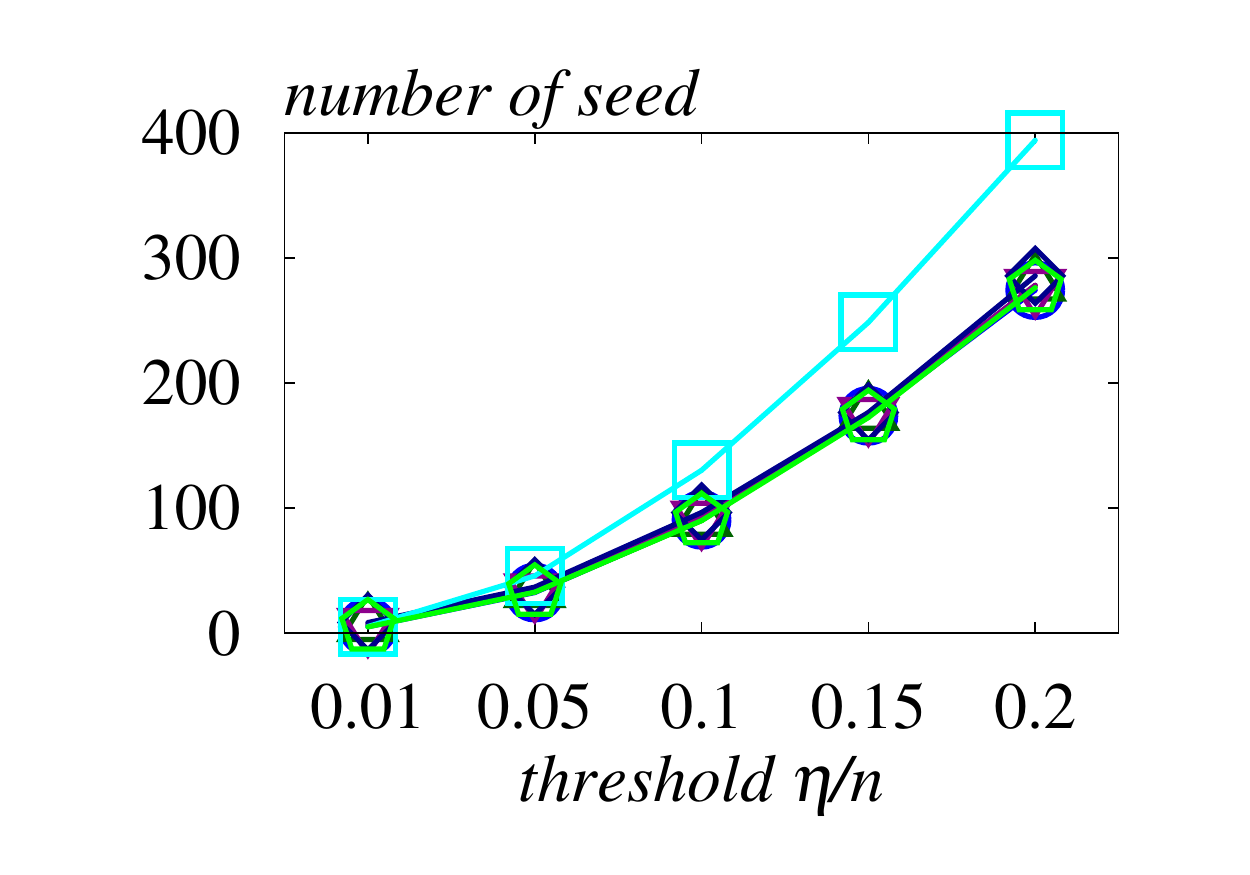}\label{subfig:NetHEPT-seed-ic}}\hfill
	\subfloat[Epinions]{\includegraphics[width=0.23\linewidth]{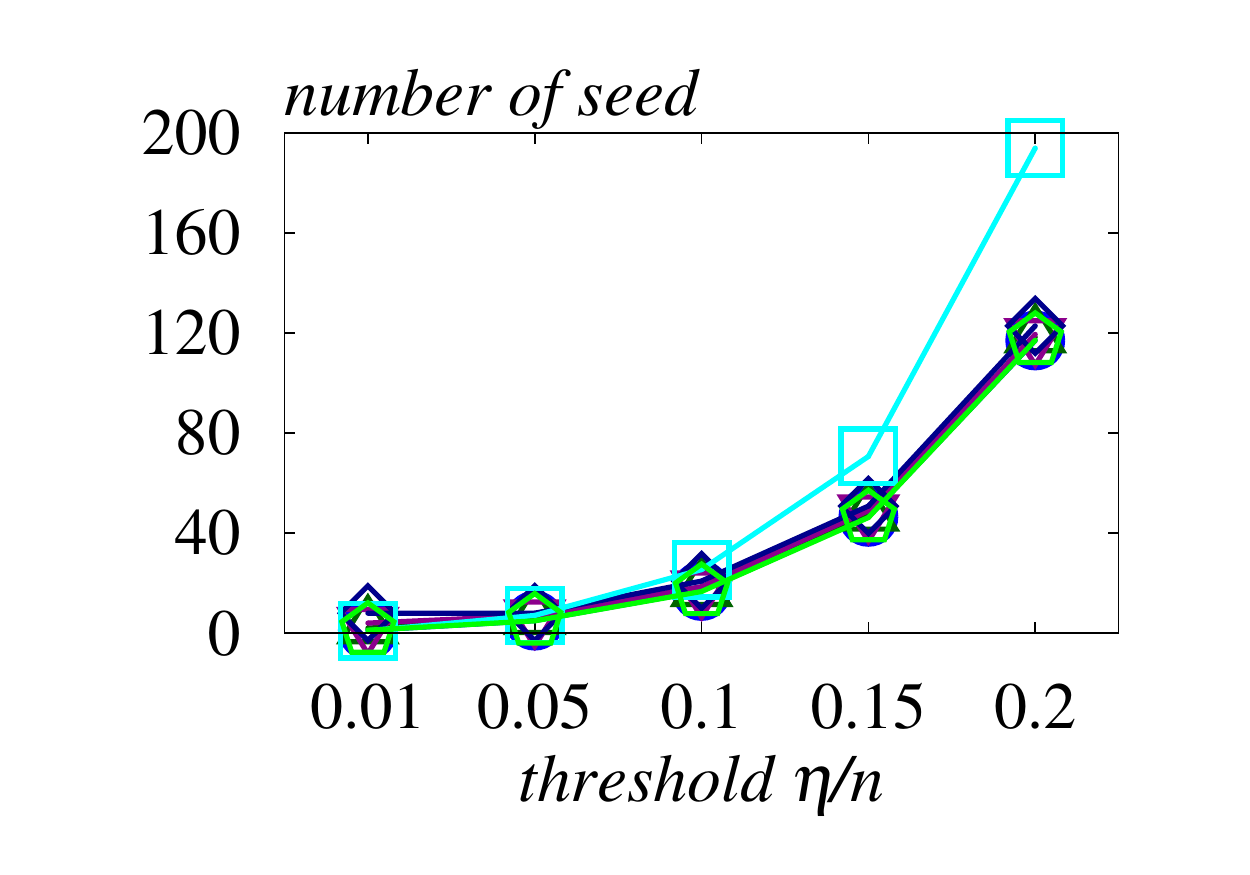}\label{subfig:Epinions-seed-ic}}\hfill
	\subfloat[Youtube]{\includegraphics[width=0.23\linewidth]{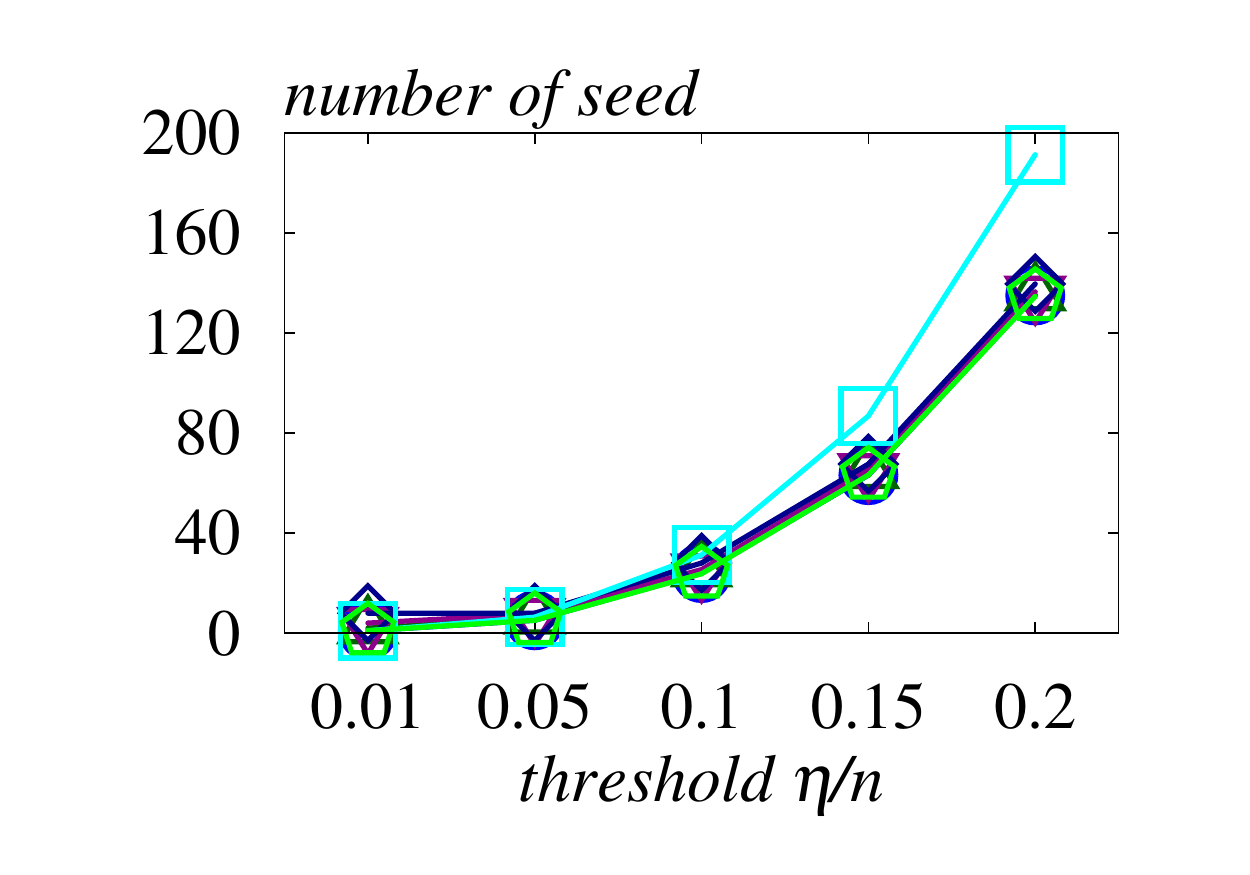}\label{subfig:Youtube-seed-ic}}\hfill
	\subfloat[LiveJournal]{\includegraphics[width=0.23\linewidth]{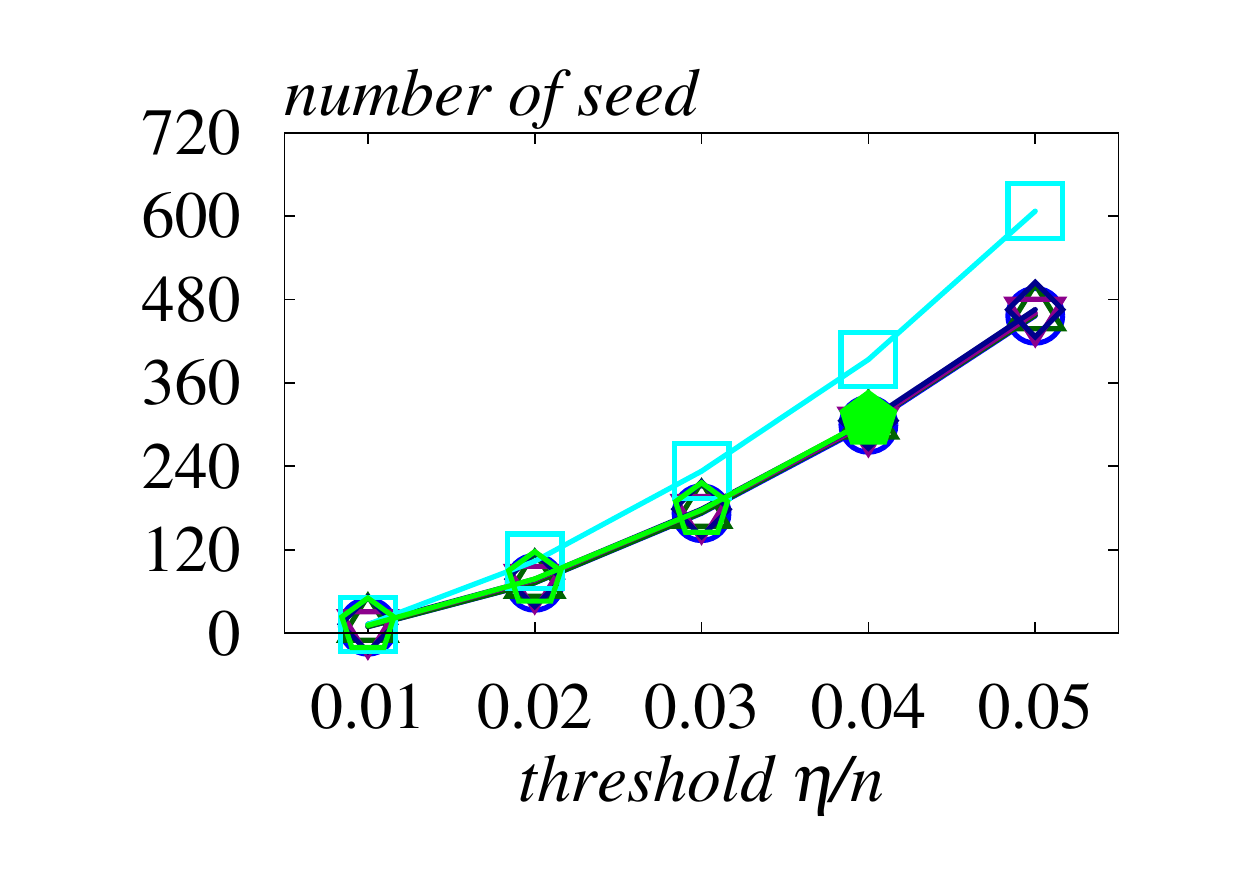}\label{subfig:LiveJournal-seed-ic}}
	\caption{Number of seed nodes vs. threshold under the IC model.}\label{fig:seed-ic}
\end{figure*}

{\spara{Algorithms} We evaluate six algorithms: \ASM, \ASMT, \ASMF, \ASME, \AdaptSM and \TEUC~\cite{Han_SM_2017}. \ASMB is \ASM instantiated by \OPIMFB with the batch sizes of $b$. (Note that \ASM is the version with a batch size of $1$.) \AdaptSM is modified from the \AdaptIM method proposed in \cite{Han_AIM_2018} for the adaptive influence maximization problem. It iteratively runs a non-adaptive influence maximization algorithm (\ie~EPIC~\cite{Han_AIM_2018}) to select the node that maximizes the expected marginal influence spread on the residual graphs, until the desired threshold is reached. \AdaptSM differs from our \ASM algorithm in that it greedily selects the node to maximize the influence spread instead of the truncated influence spread. The batch size of \AdaptSM is set to $1$ by default. As introduced in Section~\ref{sec:related-work}, \TEUC is the state of the art for the non-adaptive seed minimization problem. By comparing \ASM with \TEUC, we aim to prove the advantage of adaptivity over non-adaptivity in terms of the effectiveness. Meanwhile, three batched algorithms, \ie \ASMT, \ASMF, \ASME, are compared with both \ASM and \TEUC to study how the batch size would affect the efficiency and effectiveness. For \AdaptSM, we obtain the source code of \AdaptIM from the authors~\cite{Han_AIM_2018} with some necessary modifications (e.g., stop condition). For the other five algorithms, we implement them in C++ strictly following the algorithm description and compile them with the same optimization options.}

\spara{Parameter Settings} In our experiments, all the algorithms are tested under both the \textit{Independent Cascade (IC)} model and the \textit{Linear Threshold (LT)} model. Following the common setting in the literature \cite{Tang_reverse_2014,Arora_benchmark_2017}, we set the approximation parameter $\varepsilon=0.5$ for the five adaptive algorithms. For those parameters in \TEUC, we use the values recommended in~\cite{Han_SM_2017}. For each dataset, we set the edge probability $p(\langle u,v\rangle)=\frac{1}{\indeg_v}$ where $\indeg_v$ is the in-degree of node $v$.  

The performance metrics measured include the number of seeds selected and the corresponding running time. To better understand the performance of the algorithms, we design the {\em large $\eta$ setting} of the threshold for {NetHEPT}, {Epinions}, and {Youtube}, \ie $\frac{\eta}{n}=\{0.01, 0.05, 0.1, 0.15, 0.2\}$, where $n$ is the number of nodes in the social network. Observing that around $2K$ nodes are required on {LiveJournal} under the {\em large $\eta$ setting} which is not convenient for exhibition, we thus use a tailored {\em small $\eta$ setting}, \ie $\frac{\eta}{n}=\{0.01, 0.02, 0.03, 0.04, 0.05\}$ for {LiveJournal}.

\vspace{-2mm}
\subsection{Results under the IC model}\label{sec:exp-result-ic}
\vspace{-2mm}

\spara{Seed Size vs. Threshold} \figurename~\ref{fig:seed-ic} reports the number of seeds selected by the six algorithms for different thresholds $\eta$ under the IC model. As can be seen, \ASM selects far fewer seed nodes than \TEUC does, especially when the threshold $\eta$ becomes larger. In general, \TEUC selects around $30\%$--$40\%$ more nodes than \ASM does on all the four datasets. In particular, with a threshold $\eta/n=0.2$ on dataset {Epinions}, \ASM selects $116.95$ seed nodes on average while \TEUC needs $193.8$ seed nodes (\ie~$65.7\%$ more nodes). For the sake of clarity, Table~\ref{tbl:impro-ratio} shows the exact improvement ratio of \ASM over \TEUC on the number of seed nodes for the corresponding five thresholds under both the IC and LT model. Note that there exist many points (indicated by {\bf N/A}) where the actual number of nodes activated by the seed set returned by \TEUC does not reach the required threshold under some realizations. This is because \TEUC selects a node set $S$ such that $\E[I(S)]\geq \eta$ but may influence fewer than $\eta$ nodes under some realizations, whereas our adaptive algorithms always ensure that at least $\eta$ nodes are influenced by the returned node set under every realization. We shall explore this in more detail in Section~\ref{sec:spread-distribution}. These facts support the superiority of adaptive algorithms over non-adaptive algorithms. We also observe that the number of nodes selected by \AdaptSM is close to that of \ASM, which indicates that \AdaptSM is \textit{empirically} effective in seed minimization. However, it does not provide any approximation guarantees in terms of the number of nodes selected. Another interesting observation is that \ASMT, \ASMF, and \ASME slightly increase the number of seed nodes selected compared with \ASM and still select nodes far less than \TEUC does for most of the cases. This confirms that adaptive algorithms by utilizing the information of partial realizations are more effective than non-adaptive algorithms.
\begin{table}[!t]
	\centering
	\captionsetup{aboveskip=6pt,belowskip=0pt}
	\begin{small}
		\caption{Improvement ratio of \ASM over \TEUC}\label{tbl:impro-ratio}
		\begin{tabular}{l|l|ccccc}\hline
			\multicolumn{1}{c|}{}       & $\eta/n$    & 0.01                       & 0.05   & 0.1    & 0.15   & 0.2    \\ \hline
			\multirow{4}{*}{\rotatebox{90}{\bf IC Model}}  & NetHEPT     & {\bf N/A}                        & 40.8\% & 43.8\% & 43.0\% & 43.7\% \\
			& Epinions    & {\bf N/A}                        & {\bf N/A} & 50.7\% & {\bf N/A} & {\bf 65.7\%} \\
			& Youtube     & 0.0\%                      & 24.3\% & {\bf N/A} & 37.5\% & 41.7\% \\
			& LiveJournal & {\bf N/A}                     & 43.0\% & 34.9\% & {\bf N/A} & 33.0\% \\ \hline
			\multirow{4}{*}{\rotatebox{90}{\bf LT Model}} & NetHEPT     & {\bf N/A} & {\bf N/A} & {\bf N/A} & 44.3\% & 47.5\% \\
			& Epinions    & {\bf N/A} & {\bf N/A} & {\bf N/A} & {\bf N/A} & {\bf N/A} \\
			& Youtube     & 0.0\%                  & 39.5\% & 54.1\% & {\bf N/A} & 47.9\% \\
			& LiveJournal & {\bf N/A}                     & {\bf N/A} & {\bf N/A}    & {\bf N/A}    & {\bf N/A}  \\ \hline
		\end{tabular}
		\begin{tablenotes}
			\small
			\item {\bf N/A}: \TEUC does not meet the threshold for some realizations.
		\end{tablenotes}
	\end{small}
	    \vspace{1.7mm}
\end{table}

\begin{figure*}[!t]
	\centering
		\includegraphics[height=9.5pt]{lgd}\vspace{-0.2in}\\
	\subfloat[NetHEPT]{\includegraphics[width=0.23\linewidth]{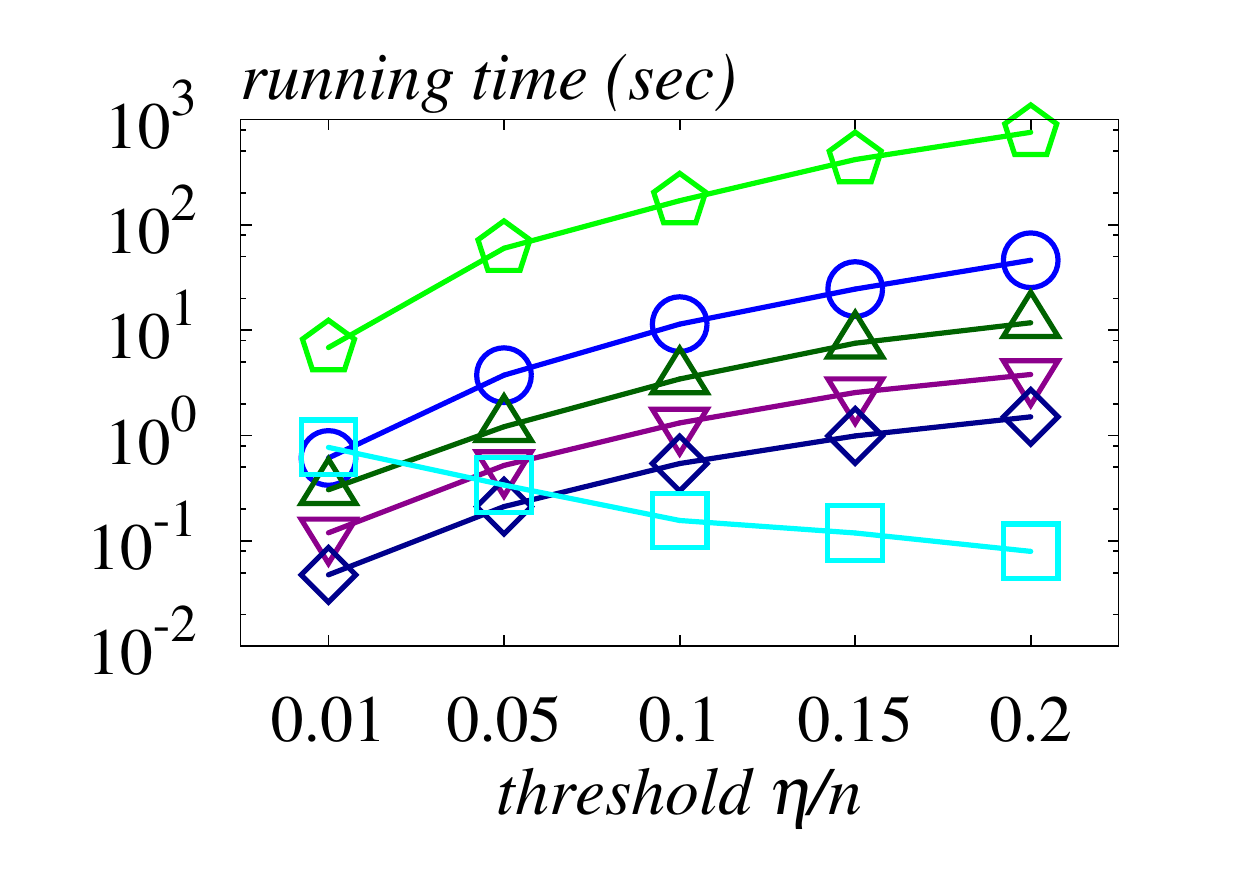}\label{subfig:NetHEPT-runtime-ic}}\hfill
	\subfloat[Epinions]{\includegraphics[width=0.23\linewidth]{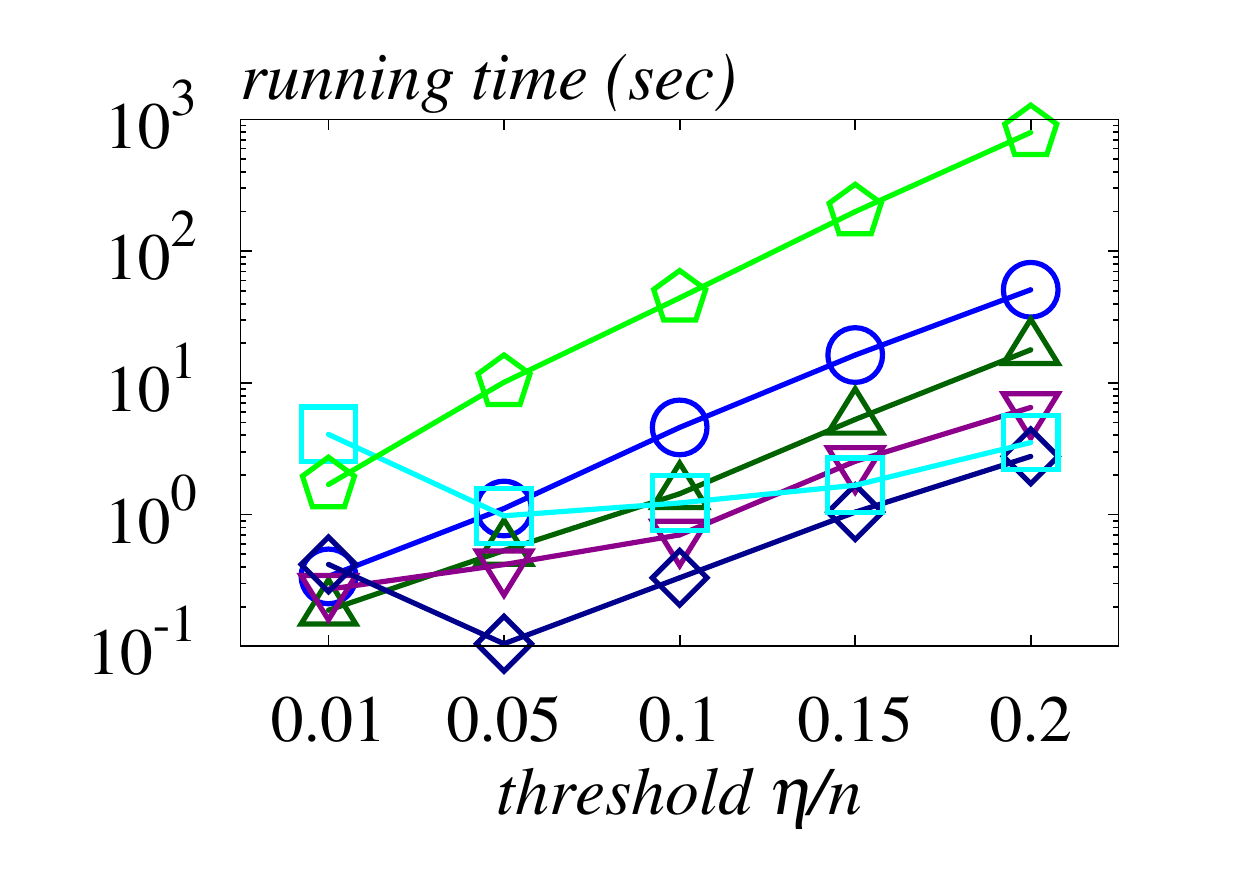}\label{subfig:Epinions-runtime-ic}}\hfill
	\subfloat[Youtube]{\includegraphics[width=0.23\linewidth]{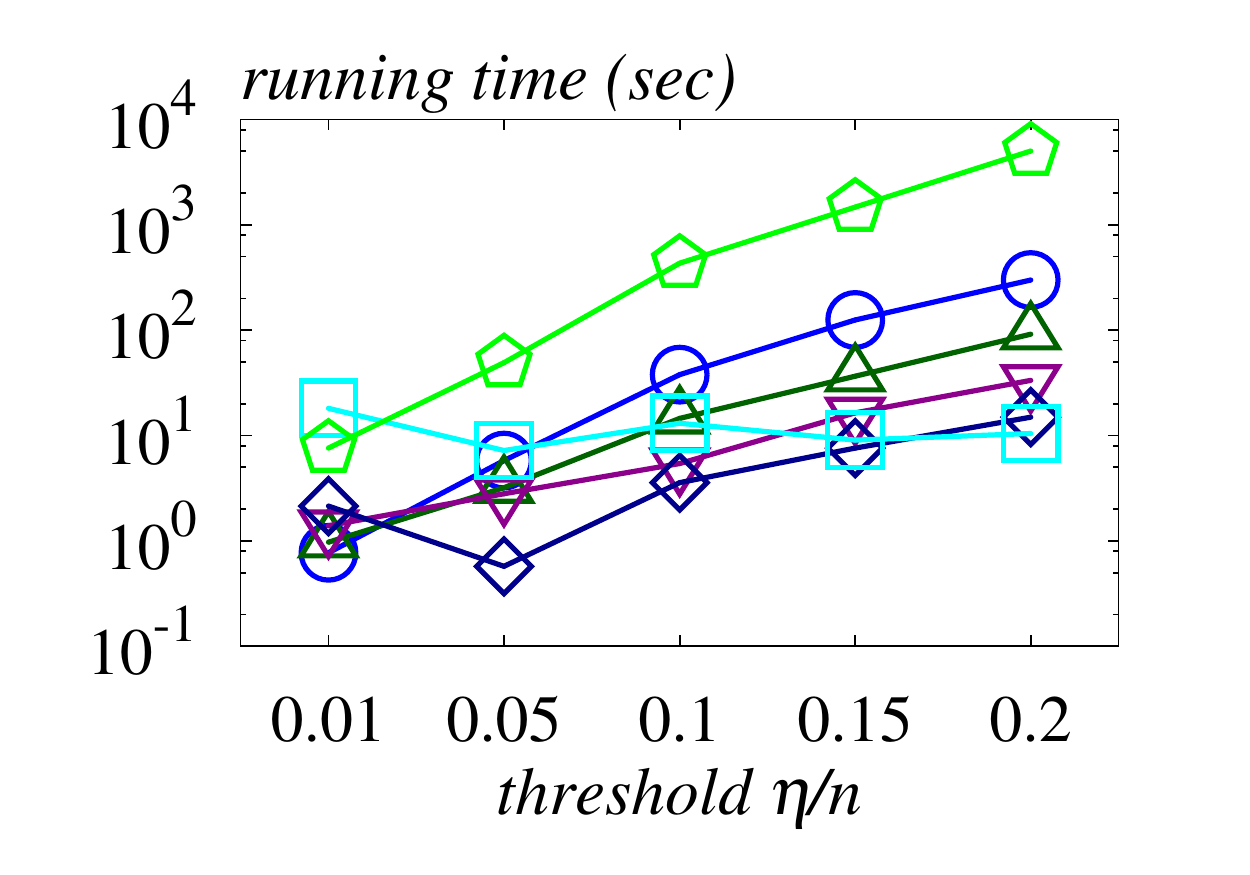}\label{subfig:Youtube-runtime-ic}}\hfill
	\subfloat[LiveJournal]{\includegraphics[width=0.23\linewidth]{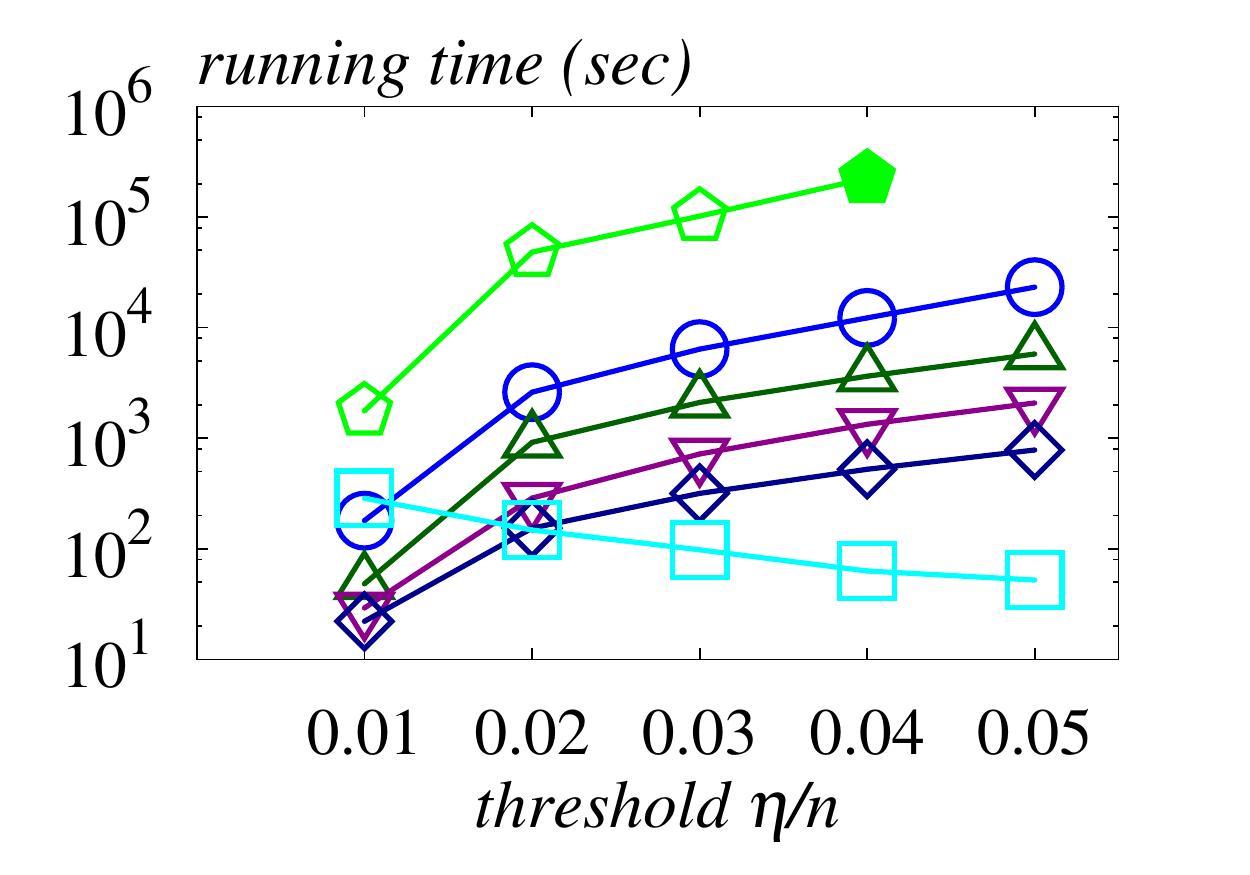}\label{subfig:LiveJournal-runtime-ic}}
	\caption{Running time vs. threshold under the IC model.}\label{fig:runtime-ic}
\end{figure*}
\begin{figure*}[!t]
	\vspace{1mm}
	\centering
	\includegraphics[height=9.5pt]{lgd}\vspace{-0.2in}\\
	\subfloat[NetHEPT]{\includegraphics[width=0.23\linewidth]{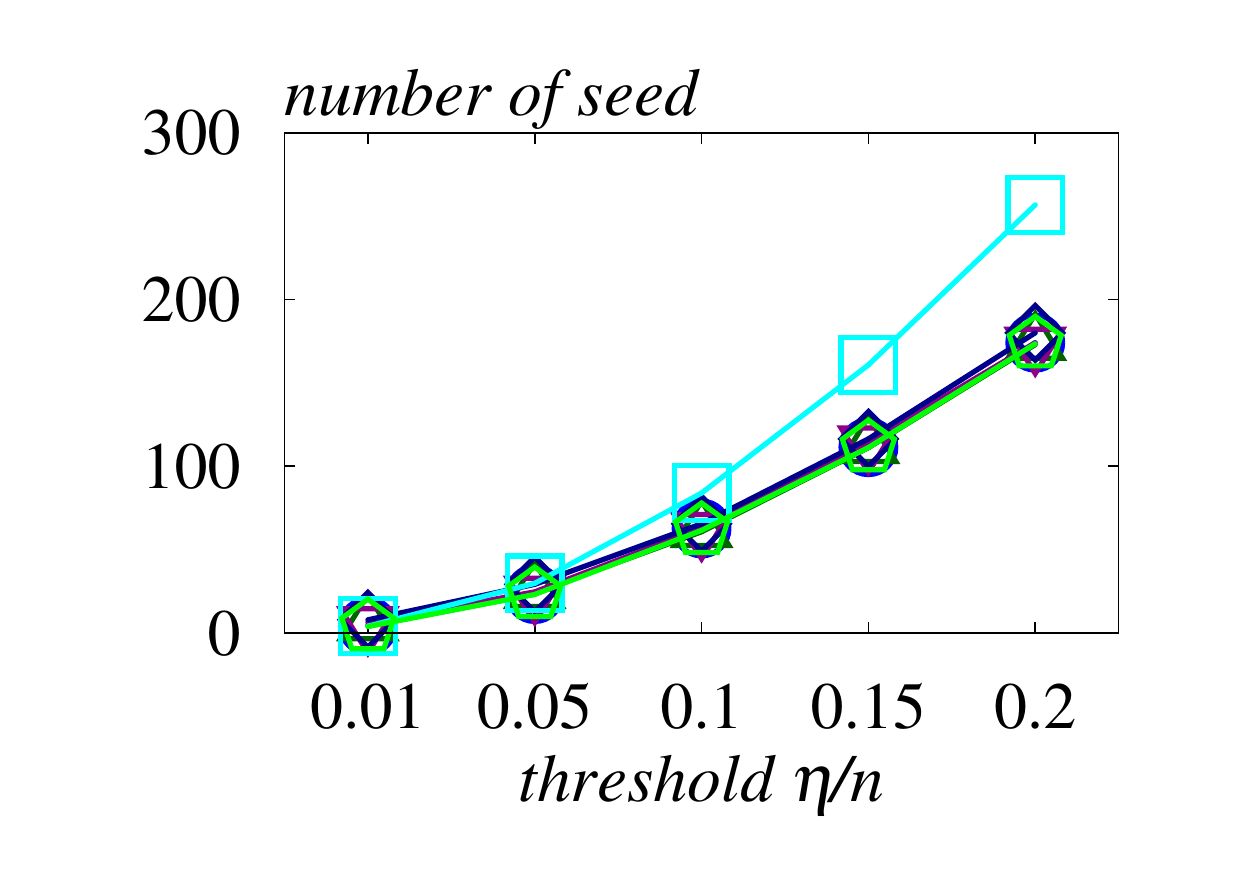}\label{subfig:NetHEPT-seed-lt}}\hfill
	\subfloat[Epinions]{\includegraphics[width=0.23\linewidth]{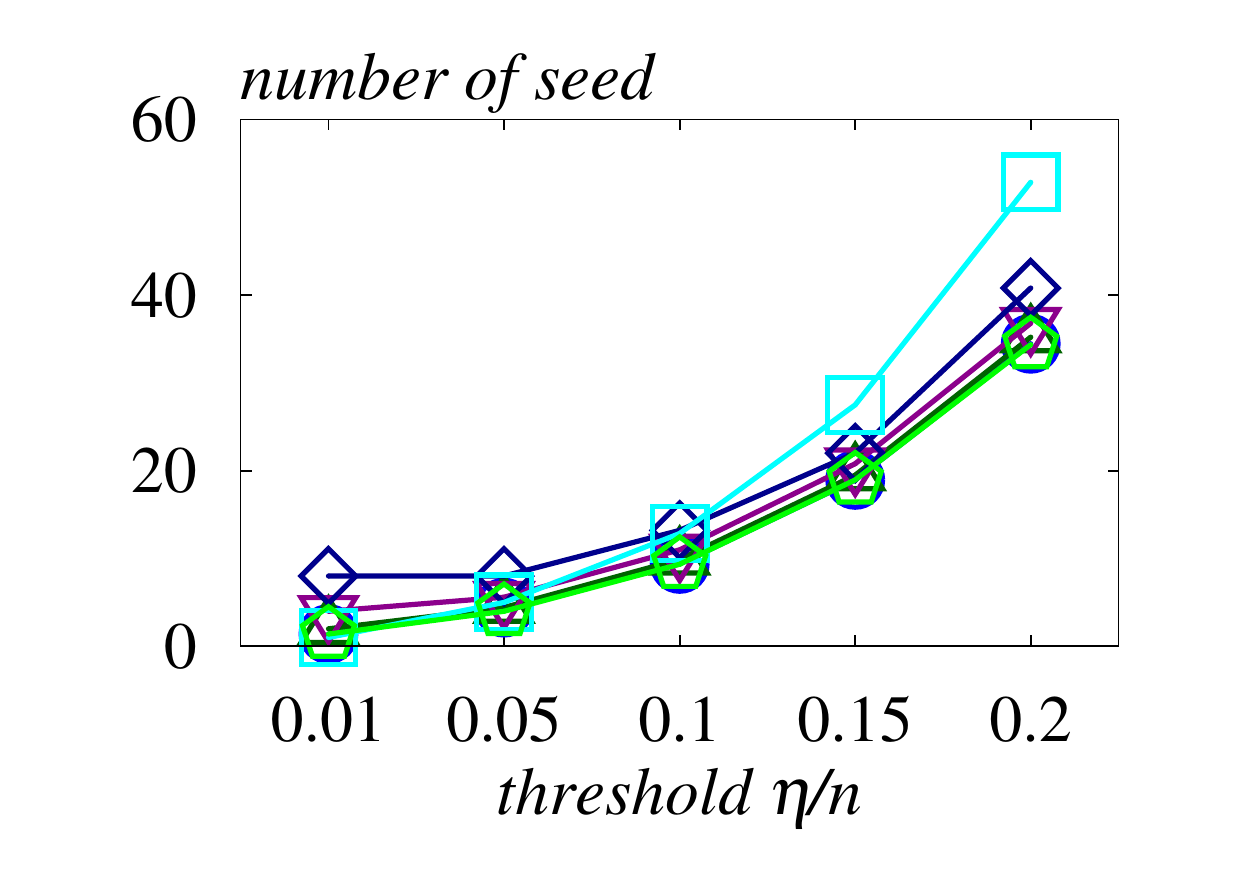}\label{subfig:Epinions-seed-lt}}\hfill
	\subfloat[Youtube]{\includegraphics[width=0.23\linewidth]{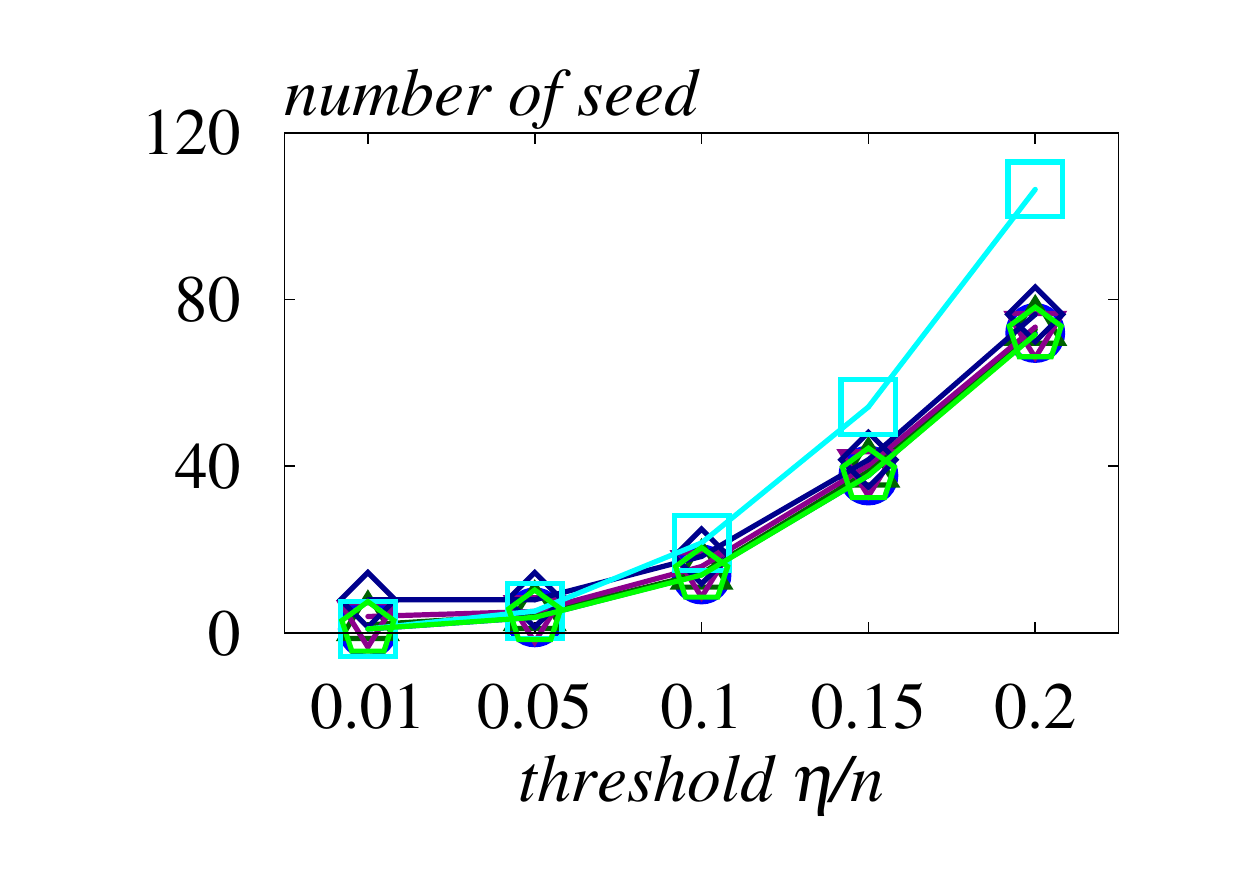}\label{subfig:Youtube-seed-lt}}\hfill
	\subfloat[LiveJournal]{\includegraphics[width=0.23\linewidth]{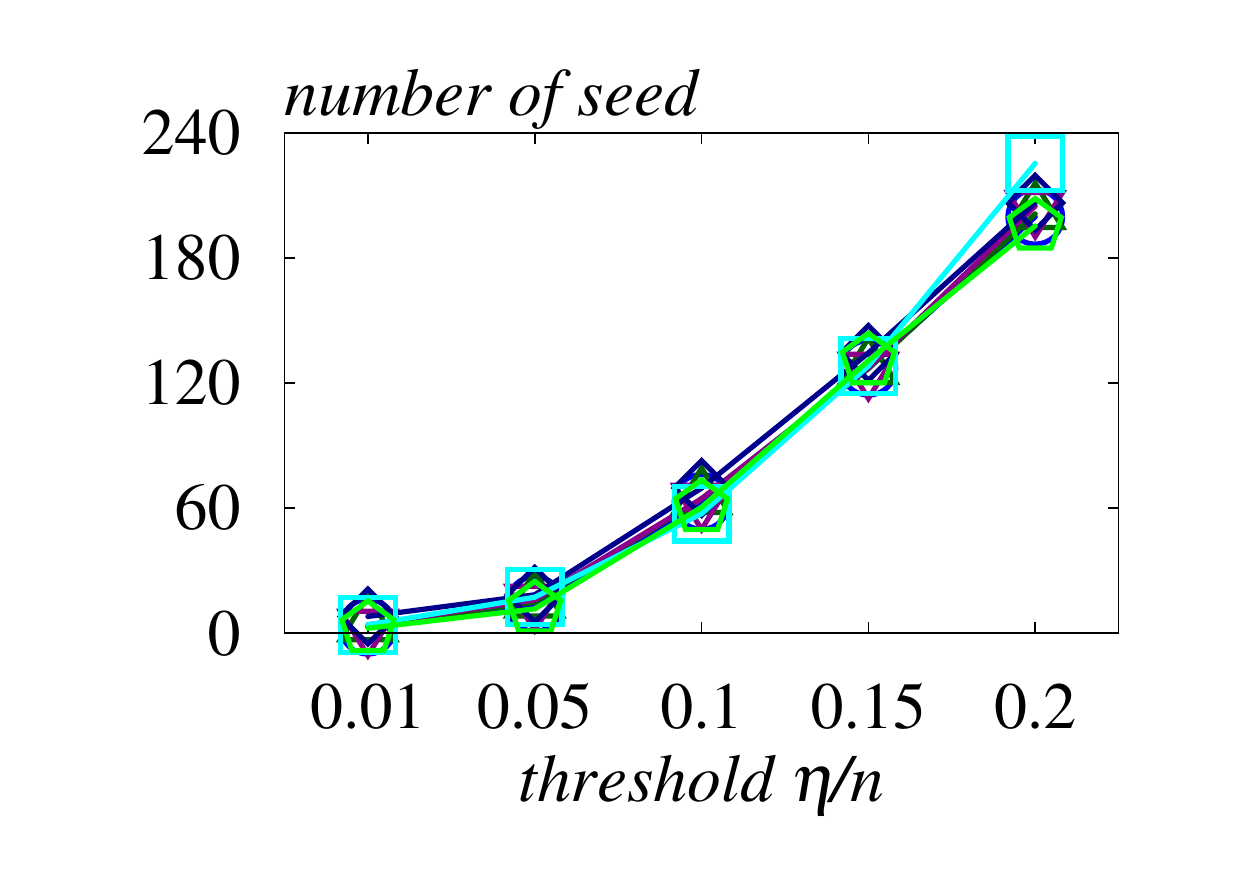}\label{subfig:LiveJournal-seed-lt}}
	\caption{Number of seed nodes vs. threshold under the LT model.}\label{fig:seed-lt}
	\vspace{1mm}
\end{figure*}
\spara{Running Time vs. Threshold} \figurename~\ref{fig:runtime-ic} presents the results of running time against the threshold under the IC model. As the results show, \TEUC runs faster than the other five adaptive algorithms on the four datasets when the threshold $\eta$ is large. The main reason is that adaptive algorithms involve multiple rounds of seed selection whereas only one round is required for non-adaptive algorithms. Observe that the running time of \TEUC generally decreases with the increase of the threshold $\eta$,  unlike the results of the five adaptive algorithms. The reason lies in the design of \TEUC. Specifically, \TEUC selects two seed set candidates $S_u$ and $S_l$, which are taken as the upper bound and lower bound on the number of seed nodes in the optimal solution. Only when the condition $|S_u| \le 2|S_l|$ is satisfied, the candidate set $S_u$ is returned as the solution; otherwise \TEUC will continue to refine $S_u$ and $S_l$~\cite{Han_SM_2017}. The larger the threshold, the more seed nodes are required, and the more easily this stop condition is met, which explains the unique running time pattern of \TEUC. We also observe that \AdaptSM runs around $10$--$20$ times slower than \ASM for all cases. Particularly, \AdaptSM cannot finish within $72$ hours when $\eta/n=0.05$ under the IC model on the LiveJournal dataset (see \figurename~\ref{subfig:LiveJournal-runtime-ic}). This demonstrates that \AdaptSM is significantly inferior to \ASM in terms of computational overheads. The reason behind this is that \ASM selects the node to maximize the expected marginal \textit{truncated} spread, while \AdaptSM attempts to maximize the expected marginal influence spread. Specifically, recall that the expected number of \RR-sets generated by \ASM is proportional to $\eta_i/\OPTT_i$. Meanwhile, the expected number of RR-sets generated by \AdaptSM is proportional to $n_i/\OPTT_i^\prime$, where $\OPTT_i^\prime$ is the maximum expected marginal influence spread in the $i$-th round of seed selection in $G_i$. For the last few rounds of seed selection, we have $\OPTT_i^{\prime}\approx\OPTT_i\approx\eta_i\ll n_i$, which indicates that the number of \RR-sets generated by \ASM is much smaller than the number of RR-sets generated by \AdaptSM. Consequently, \ASM runs remarkably faster than \AdaptSM. As such, \ASM is more preferable than \AdaptSM, as the former provides significantly better efficiency and approximation guarantees than the latter, while offering similar empirical effectiveness. Note that the batched algorithms, \ie \ASMT, \ASMF, and \ASME, reduce the running time significantly, to around $30\%$, $10\%$, and $5\%$ of \ASM, which makes them quite competitive with \TEUC in terms of the efficiency, not to mention \AdaptSM. In addition, as explained earlier, the terminal condition $|S_u| \le 2|S_l|$ in \TEUC is easier satisfied when the threshold $\eta$ is larger, and hence, \TEUC runs faster along with the increase of $\eta$. On the other hand, the running times of the adaptive algorithms increase with $\eta$. Therefore, \ASMF and \ASME outperform \TEUC on datasets Epinions and Youtube when $\eta$ is relatively small, but when the threshold $\eta/n=0.2$, the running times of all three algorithms become similar, as shown in Figures~\ref{subfig:Epinions-runtime-ic} and ~\ref{subfig:Youtube-runtime-ic}. Recall that \ASME selects far fewer seed nodes than \TEUC does. Therefore, \ASME strikes a good balance between efficiency and effectiveness in the current setting. We also observe that the running time of $\ASME$ fluctuates from $\eta/n=0.01$ to $\eta/n=0.05$ on datasets {Epinions} and {Youtube}. This is due to the combined effects of the threshold and the batch size. In these cases, it needs no more than $8$ nodes to reach the thresholds. Consequently, \ASME finishes selecting seed nodes within just one round. However, when $\eta/n$ increases from $0.01$ to $0.05$, the root size of \RR-sets decreases. As a consequence, it takes relatively less time to generate a random \RR-set in practice, which leads to the decrease in running time.

\vspace{-1mm}
\subsection{Results under the LT model}\label{sec:exp-result-lt}
\vspace{-2mm}

\begin{figure*}[!t]
	\centering
	\includegraphics[height=9.5pt]{lgd}\vspace{-0.2in}\\
	\subfloat[NetHEPT]{\includegraphics[width=0.23\linewidth]{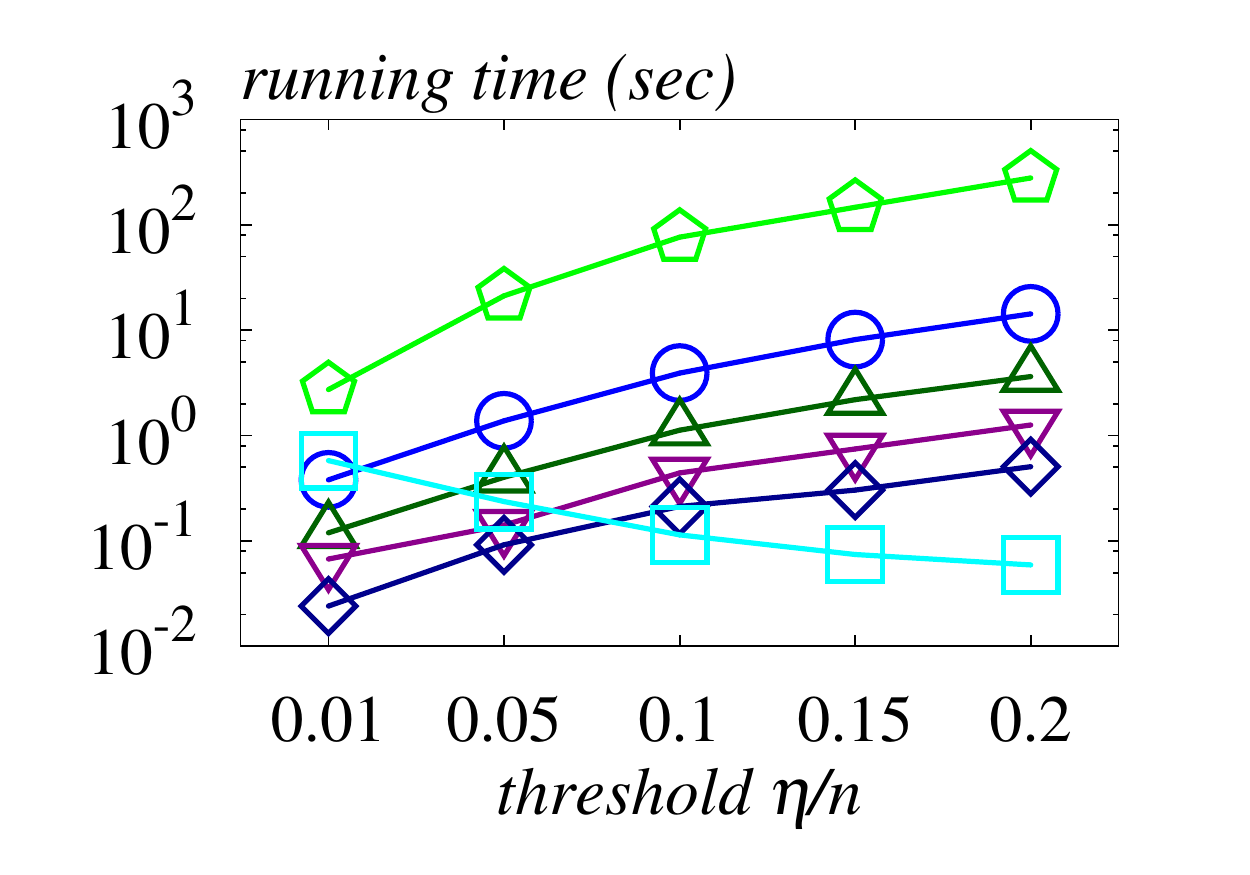}\label{subfig:NetHEPT-runtime-lt}}\hfill
	\subfloat[Epinions]{\includegraphics[width=0.23\linewidth]{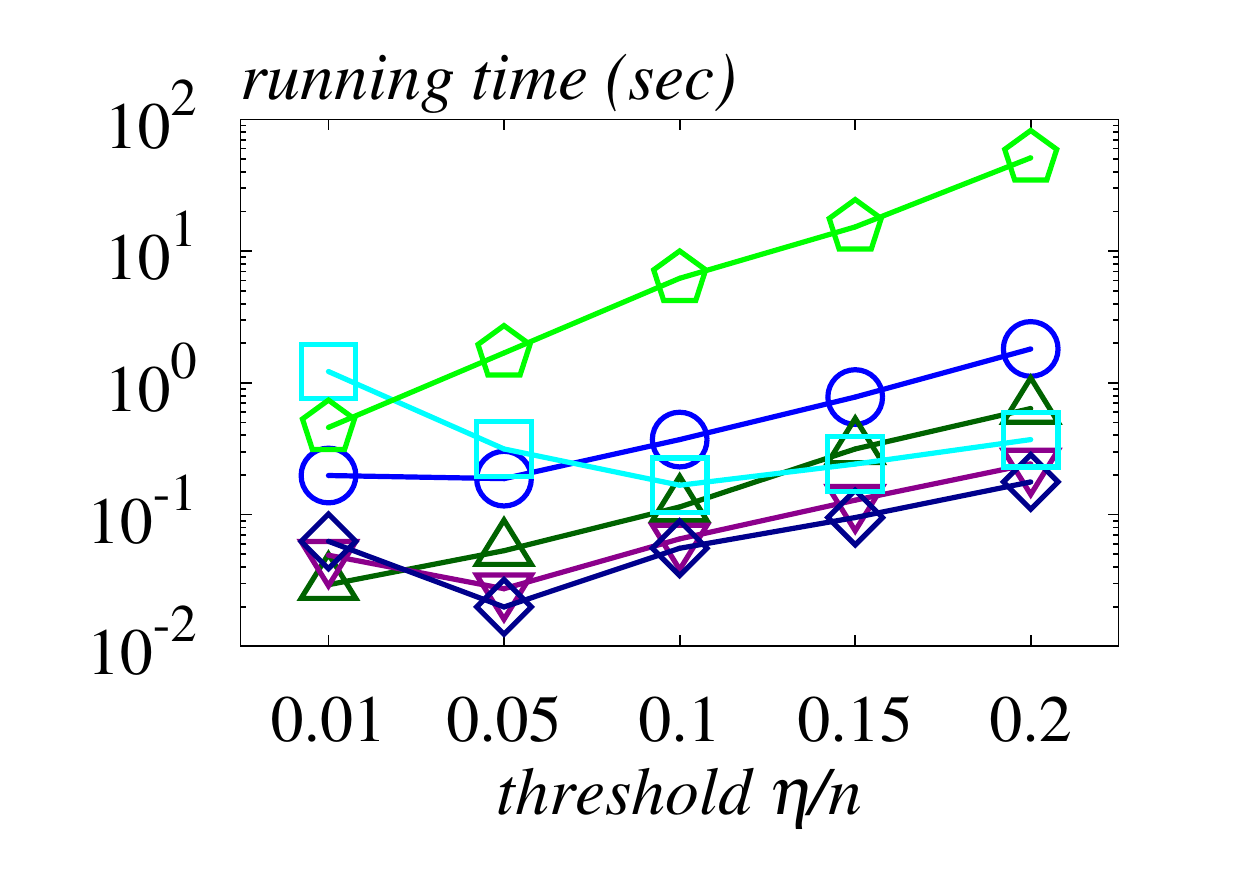}\label{subfig:Epinions-runtime-lt}}\hfill
	\subfloat[Youtube]{\includegraphics[width=0.23\linewidth]{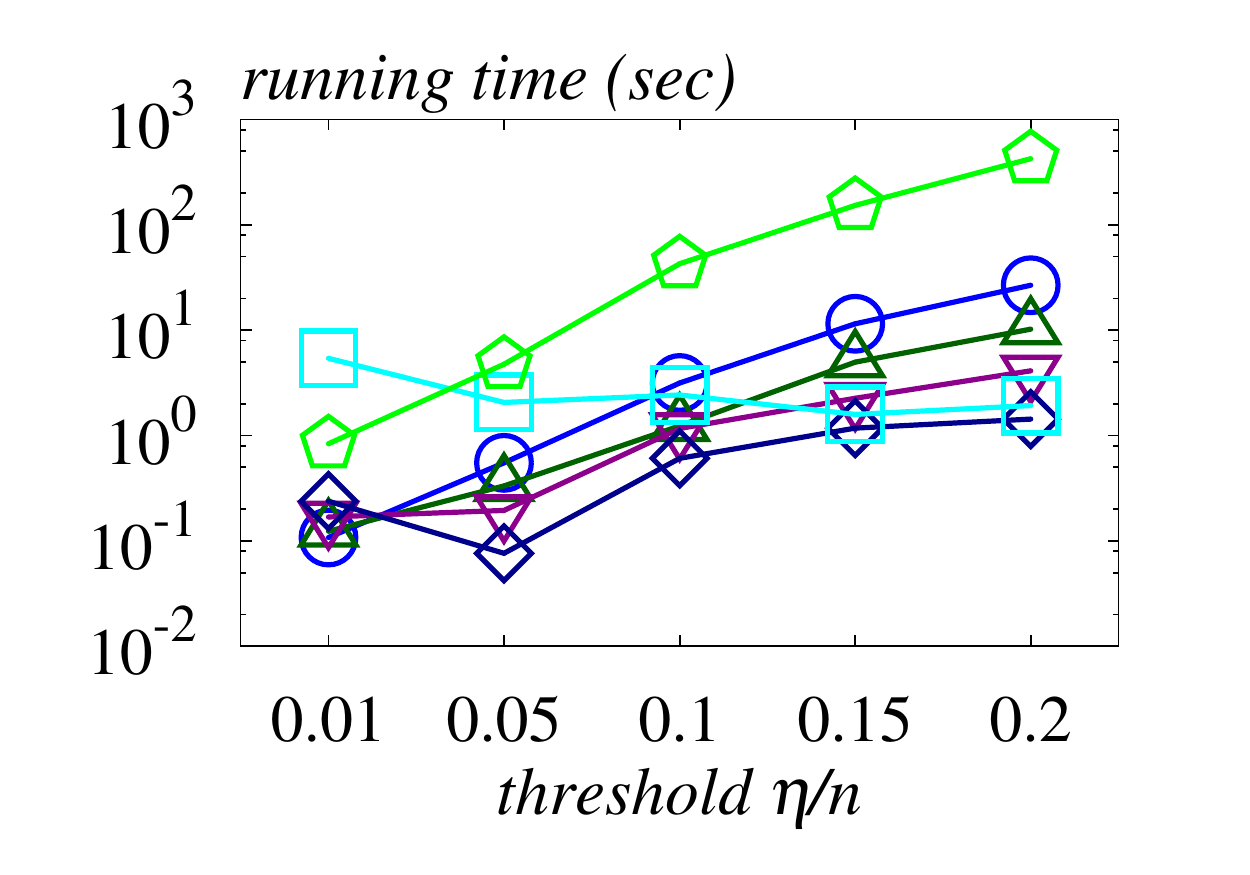}\label{subfig:Youtube-runtime-lt}}\hfill
	\subfloat[LiveJournal]{\includegraphics[width=0.23\linewidth]{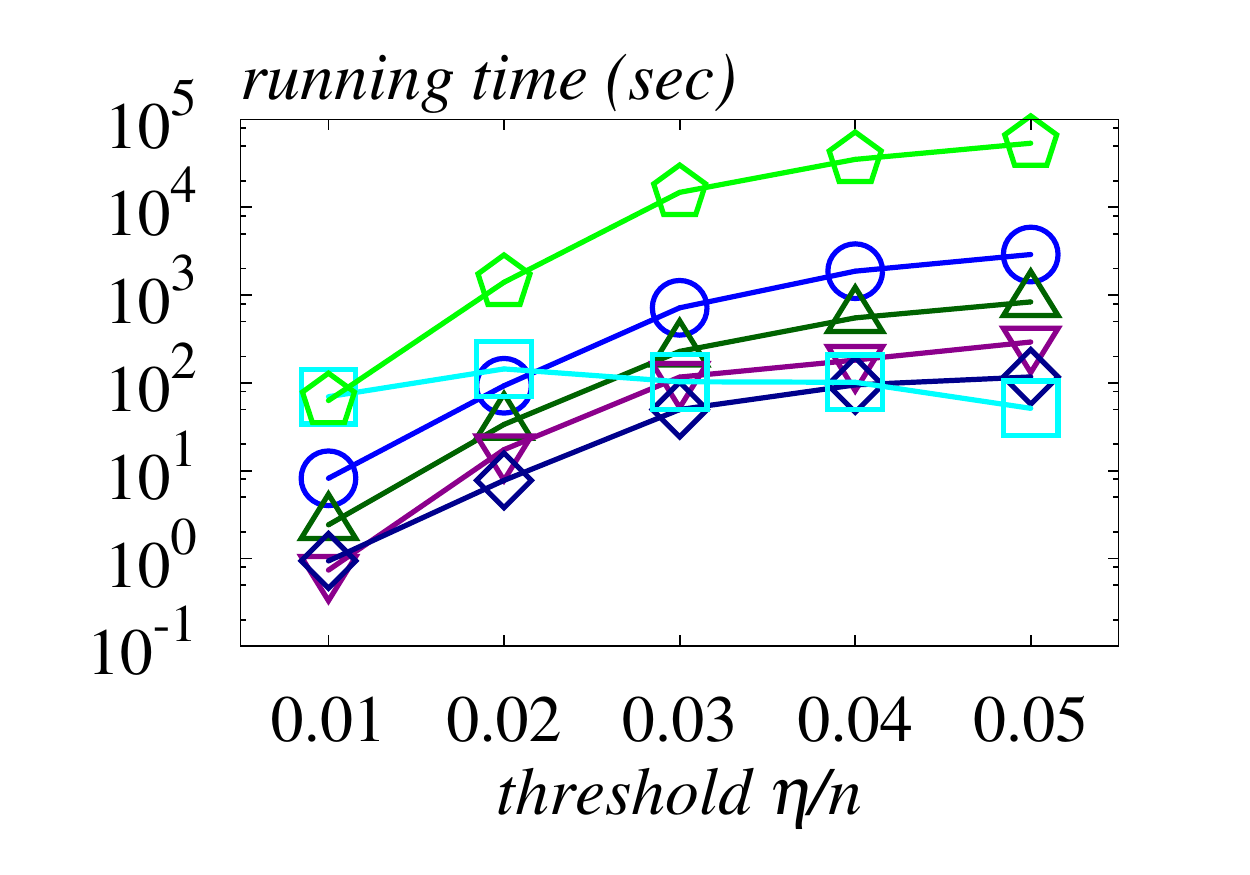}\label{subfig:LiveJournal-runtime-lt}}
	\caption{Running time vs. threshold under the LT model.}\label{fig:runtime-lt}
\end{figure*}

\spara{Seed Size vs. Threshold} \figurename~\ref{fig:seed-lt} reports the number of nodes selected by different algorithms under the LT model. In general, the results show similar trends to those observed in \figurename~\ref{fig:seed-ic}. Similarly, \AdaptSM selects a close number of nodes as \ASM does on the four datasets, with negligible difference. \TEUC requires around $40\%$ more nodes than the five adaptive algorithms do. Details are displayed in Table~\ref{tbl:impro-ratio}. In addition, we also observe that \ASME selects more nodes than \TEUC for several settings (e.g., $\eta/n=0.01$ on the Epionions and Youtube datasets). Through a careful analysis, we find that (i) all the algorithms select less nodes under the LT model than those under the IC model, and (ii) \ASME selects $8$ seed nodes in a batch with influence spread much higher than the requirements. These observations clearly tell us that there is a tradeoff in the setting of batch size. Increasing the batch size will speed up the algorithms but may result in more nodes selected. 

\spara{Running Time vs. Threshold} \figurename~\ref{fig:runtime-lt} shows the results of running time for different thresholds under the LT model. The conclusions we summarize for \figurename~\ref{fig:runtime-ic} are generally applicable to \figurename~\ref{fig:runtime-lt} as well. The major differences lie in two aspects: (i) the running time under the LT model is shorter than that under the IC model under the same setting as it takes less time to generate a random \RR-set under the LT model than that under the IC model (as mentioned and analyzed in previous work~\cite{Arora_benchmark_2017,Tang_OPIM_2018}\eat{see Appendix A in \cite{Tang_OPIM_2018} for details}), which is consistent with the results in \figurename~\ref{fig:seed-lt}, (ii) \ASMF outperforms \TEUC on {Epinions} and \ASME outperforms \TEUC on both {Epinions} and {Youtube} for all cases under the LT model. This fact indicates  (i) the batched version of \ASM is more scalable than \TEUC does, and (ii) when the batch size $b$ is well-calibrated, \ASM can beat \TEUC in both efficiency and effectiveness. 

\subsection{Discussions on Spread Distribution}\label{sec:spread-distribution}

As discussed previously, non-adaptive algorithms may find solutions with influence spread far away from the requirement (\ie~either under-qualified or over-qualified). \figurename~\ref{fig:spread-dis} reports the spread distribution of $20$ realizations achieved by the \ASM and \TEUC algorithms on the {NetHEPT} dataset under the IC and LT models, respectively. The solid (red) line in the figure represents the spread threshold ($153$) required. As shown, \TEUC fails to reach the threshold for $5$ and $6$ realizations under the IC and LT models, respectively, with corresponding percentages of $25\%$ and $30\%$. In addition, for $5$ and $6$ realizations under the IC and LT models, respectively, the seed nodes selected by \TEUC produce influence spread much higher (over $50\%$) than the requirement. In contrast, \ASM meets the spread requirement for all the realizations under both the IC and LT models. Moreover, the spread produced by \ASM is generally kept close to the requirement. The spread exceeds the requirement by more than $50\%$ for only $2$ realizations under the LT model. These two over-qualified exceptions are due to that the last seed node selected achieves much higher spread than the gap to reach $\eta$, which is rare to happen in practice. These observations indicate that non-adaptive algorithms are unreliable for seed minimization. 
\begin{figure}[!t]
	\vspace{-0.1in}
	\centering
	\vspace{-0.1in}
	\subfloat[IC model]{\includegraphics[width=0.49\linewidth]{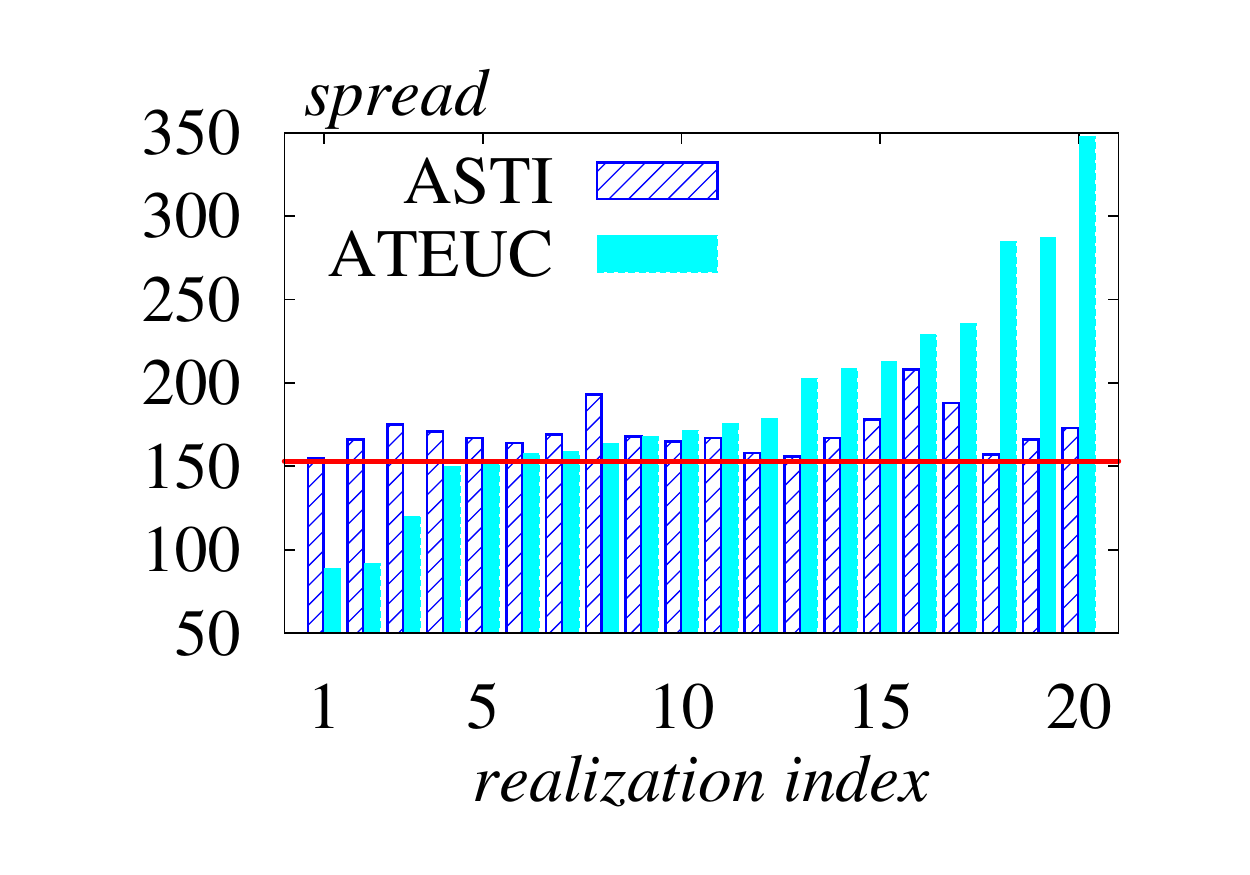}\label{subfig:hep-ic-bar}}\hfill
	\subfloat[LT model]{\includegraphics[width=0.49\linewidth]{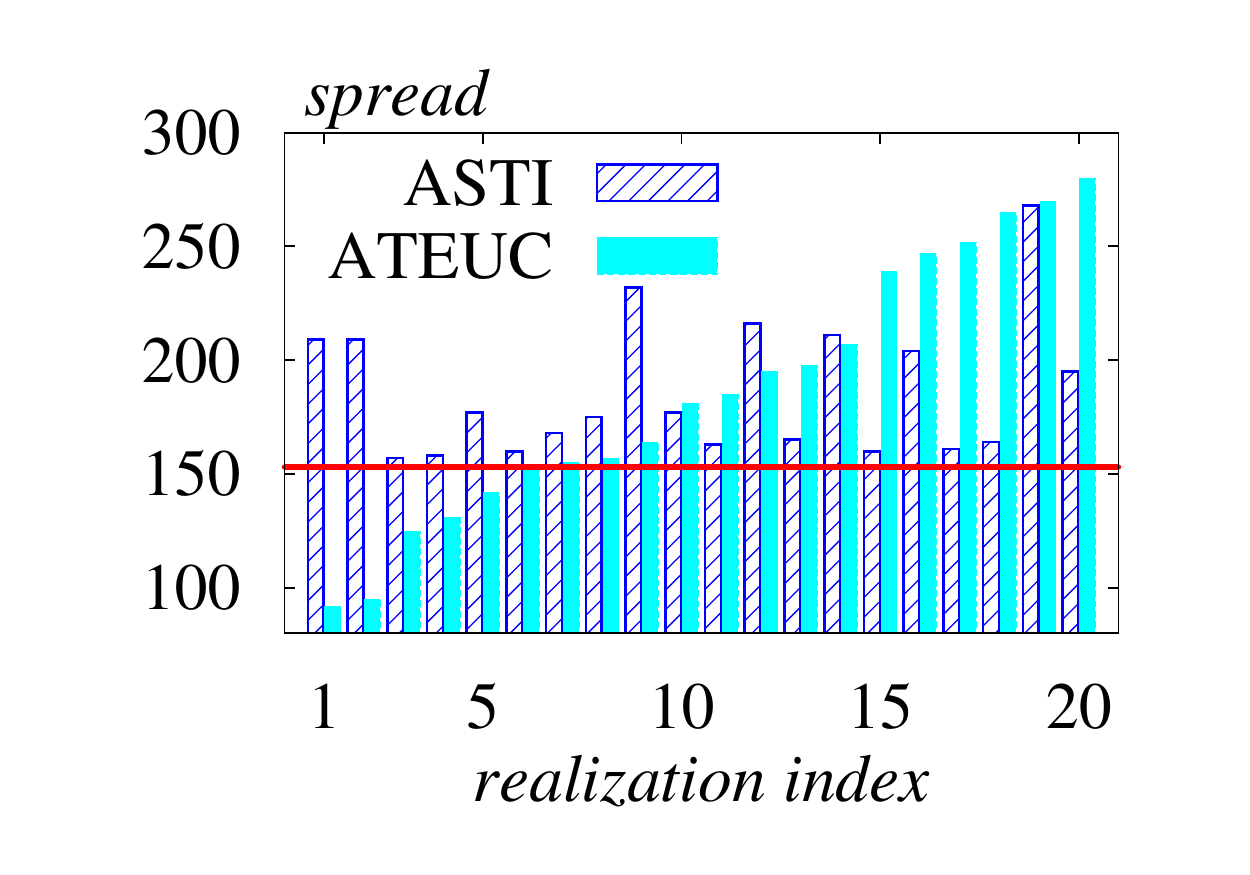}\label{subfig:hep-lt-bar}}
	\caption{Spread for 20 realizations on {NetHEPT}.}\label{fig:spread-dis}
	\vspace{2mm}
\end{figure}

\section{Conclusion}\label{sec:conclusion}
This paper studies the problem of adaptive seed minimization, and proposes algorithms that provide both strong theoretical guarantees and superior empirical effectiveness. Our approach is based on a novel \ASM framework instantiated by a truncated influence maximization algorithm \OPIMF, which has a provable approximation guarantee. The core of our \OPIMF algorithm is an elegant sampling method that generates random multi-root reverse reachable (\RR) sets for estimating the truncated influence spread. We also extend \OPIMF into its batched version \OPIMFB to further improve the efficiency of seed selection. With extensive experiments on real data, we show that our solutions considerably outperform the state of the art for seed minimization under both the IC and LT diffusion models.

\eat{We note that our adaptive algorithms take several hours to finish on the LiveJournal dataset with millions of nodes when the threshold is large (e.g., $\eta/n=0.05$). Reusing the \RR-sets generated in previous rounds of seed selection can further speed up the algorithms but will introduce bias on spread estimation. In the future, we will study how to fix or characterize the bias.}

\begin{acks}
	This research is supported by \grantsponsor{NRF-RSS2016-004}{Singapore National Research Foundation}{} under grant~\grantnum{NRF-RSS2016-004}{NRF-RSS2016-004}, by \grantsponsor{MOE2015-T2-2-069}{Singapore Ministry of Education Academic Research Fund Tier 2}{} under grant~\grantnum{MOE2015-T2-2-069}{MOE2015-T2-2-069}, by \grantsponsor{NUSSUG}{National University of Singapore}{} under an~\grantnum{NUSSUG}{SUG}, by \grantsponsor{MOE2017-T1-002-024}{Singapore Ministry of Education Academic Research Fund Tier 1}{} under grant~\grantnum{MOE2017-T1-002-024}{MOE2017-T1-002-024}, and by a \grantnum{NSERC}{Discovery grant} and a \grantnum{NSERC}{Discovery Accelerator Supplement grant} from the \grantsponsor{NSERC}{Natural Sciences and Engineering Research Council of Canada (NSERC)}{}.
\end{acks}

\newpage
\bibliographystyle{ACM-Reference-Format}
\bibliography{reference}
\appendix
\section{Concentration Bounds}\label{appenix:inequality}
We show some useful martingale concentration bounds, \ie~the Chernoff-like bounds \cite{Tang_IMM_2015} and their variants \cite{Tang_OPIM_2018}.
\begin{lemma}[\cite{Tang_IMM_2015}]\label{lemma:concentration-additive}
	Let $X_1-\E[X_1],\dots,X_{T}-\E[X_{T}]$ be a martingale difference sequence such that $X_i\in[0,1]$ for each $i$. Let $\bar{X}=\frac{1}{{T}}\sum_{i=1}^{T} X_i$. If $\E[X_i]$ is identical for every $i$, i.e., $\E[X_i]=\E[\bar{X}]$, then for any $\lambda\geq 0$, we have
	\begin{align}
	&\Pr[\bar{X}>\E[\bar{X}]+\lambda]\leq \exp\Big(-\frac{\lambda^2{T}}{2\E[\bar{X}]+2\lambda/3}\Big),\label{eqn:uppertail}\\
	&\Pr[\bar{X}<\E[\bar{X}]-\lambda]\leq \exp\Big(-\frac{\lambda^2{T}}{2\E[\bar{X}]}\Big).\label{eqn:lowertail}
	\end{align}
\end{lemma}
\begin{lemma}[\cite{Tang_OPIM_2018}]\label{lemma:concentration-ept}
	Let $X_1-\E[X_1],\dots,X_{T}-\E[X_{T}]$ be a martingale difference sequence such that $X_i\in[0,1]$ for each $i$. Let $\bar{X}=\frac{1}{{T}}\sum_{i=1}^{T} X_i$. If $\E[X_i]$ is identical for every $i$, i.e., $\E[X_i]=\E[\bar{X}]$, then for any $\lambda\geq 0$, we have
	\begin{align}
	&\Pr\Big[\E[\bar{X}]\cdot T<\Big(\sqrt{\bar{X} T+\tfrac{2\lambda}{9}}-\sqrt{\tfrac{\lambda}{2}} \Big)^2- \tfrac{\lambda}{18}\Big]\leq \e^{-\lambda},\label{eqn:ept-lower}\\
	&\Pr\Big[\E[\bar{X}]\cdot T>\Big(\sqrt{\bar{X} T+\tfrac{\lambda}{2}}+\sqrt{\tfrac{\lambda}{2}} \Big)^2\Big]\leq \e^{-\lambda}.\label{eqn:ept-upper}
	\end{align}
\end{lemma}

\section{Proofs}\label{appendix:proof}
We first introduce the following lemma that is used to prove Theorem~\ref{thm:approx}.
\begin{lemma}[\cite{Golovin_adaptive_2017}]
	If function $\Gamma$ satisfies all the following conditions:
	\begin{itemize}
		\item there exists $Q$ such that $\Gamma_\phi(V)=Q$ for all $\phi$;
		\item $\Gamma$ is integer-valued;
		\item $\Gamma$ is self-certifying;
		\item $\Gamma$ is strong adaptive monotone;
		\item $\Gamma$ is strong adaptive submodular;
	\end{itemize}
	then an $\alpha$-approximate greedy policy $\pi$ achieves an approximation ratio of $\frac{(\ln \eta+1)^2}{\alpha}$.
\end{lemma}
\begin{proof}[Proof of Theorem~\ref{thm:approx}]
	Obviously, $\Gamma_\phi(V)=\eta$ for all $\phi$ and $\Gamma$ is an integer-valued function. Now, we need to prove that for any $v\in V_i$, $\phi,\phi'\in\Omega_i$, and $j\leq i$
	\begin{align}
	&\Gamma_{\phi}(S_{i-1})=\Gamma_{\phi}(V)~\text{if and only if}~\Gamma_{\phi'}(S_{i-1})=\Gamma_{\phi'}(V),\label{eqn:self-certifying}\\
	&\Gamma_\phi(v\mid S_{i-1})\geq 0,\label{eqn:strong-monotone}\\
	&\Delta (v\mid S_{j-1})\geq \Delta (v\mid S_{i-1}),\label{eqn:adapt-sub}\\
	&\Delta (v\mid S_{j-1}; S_{i-1})\geq \Delta (v\mid S_{i-1}),\label{eqn:point-sub}
	\end{align}
	where $\Delta (v\mid S_{j-1};S_{i-1}):=\EWi{i}[\Gamma_\Phi(v\mid S_{j-1})]$. Equation~\eqref{eqn:self-certifying} represents self-certifying, Equation~\eqref{eqn:strong-monotone} describes strong monotonicity, Equations~\eqref{eqn:adapt-sub} and \eqref{eqn:point-sub} capture strong adaptive submodularity.
	
	Equation~\eqref{eqn:self-certifying} obviously holds, \ie~if $\Gamma_{\phi}(S_{i-1})=\Gamma_{\phi}(V)=\eta$, we must have $\Gamma_{\phi'}(S_{i-1})=\eta=\Gamma_{\phi'}(V)$, and vice versa.
	
	Equation~\eqref{eqn:strong-monotone} holds naturally as ``selecting more nodes never hurts'' the function $\Gamma$. 
	
	Next, we prove Equation~\eqref{eqn:adapt-sub}. Let $\phi_i$ be a realization of $G_i$ with probability $p(\phi_i)$ according to the influence propagation. Let $\Omega_j(\phi_i)$ be the subset realizations of $\Omega_j$ that are consistent with $\phi_i$. That is, for every $\phi\in\Omega_j(\phi_i)$ and every edge $e\in E_i$, the statuses of $e$ are the same in $\phi$ and $\phi_i$ such that both are either live or blocked. Then, for any $\phi_i$, 
	\begin{equation*}
	\sum_{\phi\in\Omega_j(\phi_i)}p(\phi)=p(\phi_i).
	\end{equation*}
	In addition, for any $\phi\in \Omega_i$, let $V_\phi(v\mid S_{i-1})$ be the set of nodes activated by $v$ in $G_i$. Thus, $\abs{V_\phi(v\mid S_{i-1})}$ is the spread of $v$ in $G_i$ under realization $\phi$. As a consequence, the marginal truncated spread of $v$ in $G_i$ under $\phi$ is
	\begin{equation*}
	\Gamma_\phi(v\mid S_{i-1})=\min\{\abs{V_\phi(v\mid S_{i-1})},\eta_i\}.
	\end{equation*}
	Similarly, for any $\phi\in\Omega_j$, we have
	\begin{equation*}
	\Gamma_{\phi}(v\mid S_{j-1})=\min\{\abs{V_{\phi}(v\mid S_{j-1})},\eta_j\}\geq\Gamma_\phi(v\mid S_{i-1}),
	\end{equation*}
	where the inequality is due to $G_i\subseteq G_j$ and $\eta_i\leq\eta_j$. Therefore,
	\begin{align*}
	\Delta (v\mid S_{j-1})
	&=\sum_{\phi\in\Omega_j}\Gamma_\phi(v\mid S_{j-1})\cdot p(\phi)\\
	&\geq\sum_{\phi\in\Omega_j}\Gamma_\phi(v\mid S_{i-1})\cdot p(\phi)\\
	&=\sum_{\phi_i}\sum_{\phi\in\Omega_j(\phi_i)}\Gamma_\phi(v\mid S_{i-1})\cdot p(\phi)\\
	&=\sum_{\phi_i}\Gamma_\phi(v\mid S_{i-1})\cdot\sum_{\phi\in\Omega_j(\phi_i)}p(\phi)\\
	&=\sum_{\phi_i}\Gamma_\phi(v\mid S_{i-1})\cdot p(\phi_i)\\
	&=\Delta (v\mid S_{i-1}). 
	\end{align*}
	
	Finally, we prove Equation~\eqref{eqn:point-sub}. For any $\phi\in \Omega_i$, we have
	\begin{equation*}
	\Gamma_\phi(v\mid S_{j-1})\geq \Gamma_\phi(v\mid S_{i-1}).
	\end{equation*}
	Taking the expectation over $\Phi\sim \Omega_i$ completes the proof.
\end{proof}

\begin{proof}[Proof of Theorem~\ref{thm:mrr-relative-random}]
	We prove the elementary version of Equation~\eqref{eqn:mrr-relative-random}, \ie~for any given realization $\phi$,
	\begin{equation*}
		(\ratio){\Gamma}_\phi(S)\leq \E[\tilde{\Gamma}_\phi(S)] \leq {\Gamma}_\phi(S),
	\end{equation*}
	where the expectation is only taken over the randomness of root size $K$.
	
	Let $x=I_\phi(S)$ denote the number of nodes influenced by $S$ under $\phi$. Let $p(x)$ be the probability that none of the $k$ nodes sampled can be influenced by $S$, which is given by
	\begin{equation*}
	p(x):=\Pr[\tilde{\Gamma}_\phi(S)=0]={\tbinom{n-x}{k}}/{\tbinom{n}{k}}=\prod_{i=0}^{k-1}\frac{n-x-i}{n-i}.
	\end{equation*}
	Then, by the definition of $\tilde{\Gamma}_\phi(S)$, with probability $p(x)$, $\tilde{\Gamma}_\phi(S)=0$; and with probability $1-p(x)$, $\tilde{\Gamma}_\phi(S)=\eta$. As a consequence, we have
	\begin{equation}
	\E[\tilde{\Gamma}_\phi(S)] = \eta\big(1- \E[p(x)]\big),
	\end{equation}
	where the expectation on the right hand side is taken with respect to the randomness of $k$. Let $f(x)$ be the ratio of $\E[\tilde{\Gamma}_\phi(S)]$ to ${\Gamma}_\phi(S)$, which is given by
	\begin{equation*}
	f(x):=\frac{\E[\tilde{\Gamma}_\phi(S)]}{{\Gamma}_\phi(S)}=\frac{\eta\big(1- \E[p(x)]\big)}{\min\{x,\eta\}}.
	\end{equation*}
	Now, we need to prove that $\ratio\leq f(x)\leq 1$. We consider the following two scenarios: (i) $x\geq \eta$, and (ii) $x < \eta$.
	
	\underline{(i) $x\geq \eta$:} In this case, $f(x)=1- \E[p(x)]\leq 1$. Meanwhile,
	\allowdisplaybreaks[4]
	\begin{align*}
	\E[p(x)]
	&=r\prod_{i=0}^{k^\bot}\frac{n-x-i}{n-i}+(1-r)\prod_{i=0}^{k^\bot-1}\frac{n-x-i}{n-i}\\
	&=\Big(\frac{r(n-x-k^\bot)}{n-k^\bot}+(1-r)\Big)\prod_{i=0}^{k^\bot-1}\frac{n-x-i}{n-i}\\
	&=\Big(1-\frac{rx}{n-k^\bot}\Big)\prod_{i=0}^{k^\bot-1}\frac{n-x-i}{n-i}.
	\end{align*}
	As $1-y\leq \e^{-y}$ for any $0\leq y\leq 1$, in the above equation,
	\begin{equation*}
	    r.h.s.\leq \e^{-{rx}/{(n-k^\bot)}}\prod_{i=0}^{k^\bot-1} \e^{-x/(n-i)}\leq \e^{-({rx}/{n}+k^\bot x/n)}=\e^{-x/\eta}.
	\end{equation*}
	As $x\geq \eta$ by assumption, this implies that $f(x)\geq \ratio$.
	
	\underline{(ii) $x < \eta$:} In this case, $f(x)=\eta\big(1- \E[p(x)]\big)/x$. Take the derivative, 
	\begin{equation*}
	f'(x)
	=\frac{\eta}{x^2}\big(\E[p(x)]-1-x\E[p'(x)]\big).
	\end{equation*}
	Let $g(x)=p(x)-1-xp'(x)$. Take the derivative, when $x> 0$,
	\begin{equation*}
	g'(x)=p'(x)-p'(x)-xp''(x)=-xp''(x).
	\end{equation*}
	According to the definition of $p(x)$, we can get that
	\begin{equation*}
	    p''(x)=\sum_{i=0}^{k-1}\sum_{j=0,j\neq i}^{k-1}\frac{p(x)}{(n-x-i)(n-x-j)}\geq 0.
	\end{equation*}
	Thus, $g(x)$ decreases with $x$, which indicates that $g(x)\leq g(0)=0$. This implies that $f'(x)\leq 0$. As a consequence,
	\begin{align*}
	&f(x)\leq f(1)=\eta(1-\E[p(1)])=\eta(1-\E[{(n-k)}/{n}])=1,\\
	&f(x)\geq f(\eta)=1-\E[p(\eta)]\geq \ratio,
	\end{align*}
	where the last step above follows from the analysis for the case $x \ge \eta$, by considering the special case $x = \eta$.

	Hence, the theorem is proved.
\end{proof}

\begin{proof}[Proof of Corollary~\ref{corollary:mrr-relative-random-general}]
	Equation~\eqref{eqn:mrr-relative-random-general} follows directly from  Theorem~\ref{thm:mrr-relative-random}. By Equation~\eqref{eqn:mrr-relative-random-general},
	\begin{align*}
		&\E[{\Gamma}(S\mid S_{i-1})]\geq \E[\tilde{\Gamma}(S\mid S_{i-1})],\\
		&\frac{1}{\E[{\Gamma}(S^\prime\mid S_{i-1})]}\geq \frac{\ratio}{\E[\tilde{\Gamma}(S'\mid S_{i-1})]}.
	\end{align*}
	Hence, Equation~\eqref{eqn:mrr-relative-ratio} holds.
\end{proof}

\begin{proof}[Proof of Lemma~\ref{lem:NP-appro}]
	We consider the special case of the adaptive seed minimization problem in which the probability $p(e)=1$ for each edge $e \in E$. In this case, for any node $v \in V$, the set of nodes influenced by $v$ is the set of nodes that can be reached by $v$ in $G$, denoting as the \textit{cover set} $S_v$. Thus, for each node $v \in V$, its cover set $S_v$ is deterministic. As a consequence, the adaptive seed minimization problem reduces to a \textit{set cover} problem, \ie~aiming to find as few nodes as possible to cover at least $\eta$ nodes. 
	Feige~\cite{Feige_setcover_1998} has shown that no polynomial time algorithm can approximate the optimal solution of set cover within a ratio of $(1-\varepsilon)\ln \eta$ for any $\varepsilon >0$ unless $\NP\subseteq \DTIME(n^{O(\log \log n)})$. Hence, lemma~\ref{lem:NP-appro} holds on noting that ASM generalizes set cover. 
\end{proof}

\begin{proof}[Proof of Lemma~\ref{lem:thrim-alpha}]
	Let $\mathcal{E}$ be the following event:
	\begin{equation*}
	\mathcal{E}(v^\ast)\colon \E[\tilde{\Gamma}(v^\ast\mid S_{i-1})]\geq(1-\hat{\varepsilon})\E[\tilde{\Gamma}(v^\circ\mid S_{i-1})].
	\end{equation*}
	Note that $v^\ast$ is the seed node returned by the policy which is a random variable. Let $U_t$ be the set of possible seed nodes selected (but not necessarily returned) by \OPIMF in the $t$-th iteration in which each node $u\in U_t$ has a probability $\Pr[u]$ such that $\sum_{u\in U_t}\Pr[u]=1$, where $1\leq t\leq T$. Let $U^\ast_t$ denote the set of random seed nodes returned at the $t$-th iteration of \OPIMF, where $U^\ast_t\subseteq U_t$. Therefore, the event $\mathcal{E}(v^\ast)$ does not happen only if there exists a node $v^\ast_t\in U^\ast_t$ at iteration $t\in[1, T]$ satisfying that $\mathcal{E}(v^\ast_t)$ does not happen.
	
	If \OPIMF stops at the iteration $t=T$, according to the setting of $\theta_{\max}$ and by \cite{Tang_IMM_2015}, we have
	\begin{equation}\label{eqn:prob-last}
	\Pr[(t=T)\wedge \neg \mathcal{E}(v^\ast_t)]\leq \delta/3.
	\end{equation}
	If \OPIMF stops at the iteration $t<T$, for any node $v\in V_i$, we define two events $\mathcal{E}_1(v)$ and $\mathcal{E}_2(v)$ as
	\begin{align*}
	&\mathcal{E}_1(v)\colon\E[\Lambda_\R(v)]\geq\big(\sqrt{\Lambda_\R(v)+{2a_1}/{9}}-\sqrt{{a_1}/{2}} \big)^2- {a_1}/{18},\\
	&\mathcal{E}_2(v)\colon\E[\Lambda_\R(v)]\leq\big(\sqrt{\Lambda_\R(v)+{a_2}/{2}}+\sqrt{{a_2}/{2}} \big)^2.
	\end{align*}
	where $\E[\Lambda_\R(v)]={\abs{\R}}\cdot\E[\tilde{\Gamma}(v\mid S_{i-1})]/{\eta_i}$ is the expected coverage of $v$ in $\R$.
	Then, if $v$ is independent of $\R$, by Lemma~\ref{lemma:concentration-ept}, we have
	\begin{align}
	&\Pr\big[\neg \mathcal{E}_1(v)\big]\leq\frac{\delta}{3Tn_i},\label{eqn:ept-tis-lower}\\
	&\Pr\big[\neg \mathcal{E}_2(v)\big]\leq\frac{\delta}{3T}.\label{eqn:ept-tis-upper}
	\end{align}
	By a union bound that ensures all the $n_i$ nodes satisfying Equation~\eqref{eqn:ept-tis-lower}, we have 
	\begin{align*}
	\Pr\big[\neg \mathcal{E}_1(v^\ast_t)\big]
	&=\sum_{u\in U^\ast_t}\Pr[\neg \mathcal{E}_1(u)]\cdot\Pr[u]\\
	&\leq\sum_{u\in U_t}\Pr[\neg \mathcal{E}_1(u)]\cdot\Pr[u]\\
	&\leq \sum_{u\in U_t}{\delta}/{(3T)}\cdot \Pr[u]\\
	&={\delta}/{(3T)}
	\end{align*}
	Meanwhile, $v^\circ$ is independent of $\R$ naturally. Thus, together with the fact that $\Lambda_\R(v^\circ)\leq \Lambda_\R(v^\ast)$, by Equation~\eqref{eqn:ept-tis-upper}
	\begin{equation*}
	\Pr\big[\E[\Lambda_\R(v^\circ)]>\Lambda^u(v^\circ)\big]\leq\Pr\big[\neg \mathcal{E}_2(v^\circ)\big]\leq{\delta}/{(3T)}.
	\end{equation*}
	As a consequence, when \OPIMF stops at ${\Lambda^l(v^\ast_t)}/{\Lambda^u(v^\circ)}\geq 1-\hat{\varepsilon}$, if the event $\mathcal{E}(v^\ast_t)$ does not happen, then at least one of the events $\mathcal{E}_1(v^\ast_t)$ and $\mathcal{E}_2(v^\circ)$ does not happen. Thus, the event $\mathcal{E}(v^\ast_t)$ does not happen for all $t<T$ with probability at most:
	\begin{equation}\label{eqn:prob-before}
	\Pr\bigg[\!\bigvee_{t=1}^{T-1}\neg\mathcal{E}(v^\ast_t)\bigg]\leq (T-1)(\frac{\delta}{3T}+\frac{\delta}{3T})\leq \frac{2\delta}{3}.
	\end{equation}
	Combining Equations~\eqref{eqn:prob-last} and \eqref{eqn:prob-before} shows that the event $\mathcal{E}(v^\ast)$ holds with probability at least $1-\delta$. Thus, together with the Equation~\eqref{eqn:mrr-relative-ratio} in Corollary~\ref{corollary:mrr-relative-random-general}, we have
	\begin{align*}
	\E\Big[\frac{1}{\Delta(v^\ast\mid S_{i-1})}\Big]
	&\leq \Big(\frac{1-\delta}{(\ratio)(1-\hat{\varepsilon})}+\delta\cdot\eta_i\Big)\cdot\frac{1}{\Delta(v^\circ\mid S_{i-1})}\\
	&\leq \frac{1}{(\ratio)(1-\varepsilon)}\cdot\frac{1}{\Delta (v^\circ\mid S_{i-1})}.
	\end{align*}
	Hence, the lemma is proved.
\end{proof}

\begin{proof}[Proof of Lemma~\ref{lemma:time-ept}]
	For any node $v$, $v$ is not visited by a random \RR-set $R$ if and only if $v\notin R$. The probability for not visiting $v$ under a realization $\phi$ is $p(x_{v})$, where $x_{v}=I_\phi(v\mid S_{i-1})$ is the number of nodes that can be activated by $v$ in $G_i$ under $\phi$. On the other hand, if a node is visited, all of its incoming edges will be examined. Let $\indeg_v$ denote the number of $v$'s incoming edges. Then, the expected time complexity for generating a random \RR-set is
	\begin{equation}\label{eqn:def-time}
	\sum_{v\in V}\indeg_v\cdot\E[1-p(x_{v})],
	\end{equation}
	where the expectation is over the randomness of both $k$ and $\Phi$. In addition, we already know that
	\begin{equation}\label{eqn:relation-optt}
	\E[\eta_i(1-p(x_{v}))]=\E[\tilde{\Gamma}(v\mid S_{i-1})]\leq \OPTT_i,
	\end{equation}
	Combining \eqref{eqn:def-time} and \eqref{eqn:relation-optt} gives 
	\begin{equation*}
	\sum_{v\in V}{\indeg_v}\cdot\E[1-p(x_{v})]\leq \sum_{v\in V}\indeg_v\frac{\OPTT_i}{\eta_i}=\frac{\OPTT_i}{\eta_i}m_i.
	\end{equation*}
	Hence, the lemma is proved.
\end{proof}

\report{
\begin{proof}[Proof of Lemma~\ref{lemma:sample-ept}]
	Let $\varepsilon_1=\hat{\varepsilon}/{2}$ and ${\varepsilon}_2$ be the root of
	\begin{equation*}
	{\varepsilon}_2=\sqrt{\frac{a(2+2{\varepsilon}_2/3)}{\E[\Lambda_{\R}(v^\ast)]}},
	\end{equation*}
	where $a=c\ln({4n_iT}/{\delta})$ for any $c\geq 1$ and $\delta=1/n_i$. Let
	\begin{equation*}
	\theta^\ast:=\frac{12\eta_i\ln({4n_iT}/{\delta})}{(1-\hat{\varepsilon})\hat{\varepsilon}^{2}\OPTT_i}.
	\end{equation*}
	As $\hat{\varepsilon}=O(\varepsilon)$, one can verify that $\theta^\ast=O\big(\frac{\eta_i\ln n_i}{\varepsilon^2\OPTT_i}\big)$.\footnote{Without loss of generality, we assume $\hat{\varepsilon}\leq 0.5$. If $\hat{\varepsilon}>0.5$, \OPIMF achieves a higher approximation of $0.5$ with $O\big(\frac{\eta_i\ln n_i}{\OPTT_i}\big)$ \RR-sets.} Define the events $\mathcal{E}_1,\mathcal{E}_2,\mathcal{E}_3,\mathcal{E}_4$ as follows:
	\begin{align*}
	&\mathcal{E}_1=\big\{\Lambda_{\R}(v^\diamond)\geq(1-\varepsilon_1)\E[\Lambda_{\R}(v^\diamond)]\big\},\\
	&\mathcal{E}_2=\big\{\Lambda_{\R}(v^\ast)\leq\E[\Lambda_{\R}(v^\ast)]+\varepsilon_1\E[\Lambda_{\R}(v^\diamond)]\big\},\\
	&\mathcal{E}_3=\big\{\Lambda_{\R}(v^\ast)\geq(1-{\varepsilon}_2)\cdot\E[\Lambda_{\R}(v^\ast)]\big\},\\
	&\mathcal{E}_4=\big\{\Lambda_{\R}(v^\ast)\leq(1+{\varepsilon}_2)\cdot\E[\Lambda_{\R}(v^\ast)]\big\}.
	\end{align*}
	
	Then, when a number of $\abs{\R}=c\theta^\ast$ \RR-sets are generated, by Lemma~\ref{lemma:concentration-additive}, it is easy to verify that any event $\mathcal{E}_j$ ($1\leq j\leq 4$) does not happen with probability at most
	\begin{equation*}
	\Pr[\neg\mathcal{E}_j]\leq\big({\delta}/{4}\big)^c.
	\end{equation*}
	By the union bound, the probability that all the events $\mathcal{E}_1,\mathcal{E}_2,\mathcal{E}_3,\mathcal{E}_4$ happen is at least $1-\delta^c$. 
	
	If the events $\mathcal{E}_1,\mathcal{E}_2$ happen,
	\begin{align*}
	\E[\Lambda_{\R}({S}^\ast)]
	&\geq\Lambda_{\R}(v^\ast)-{\varepsilon}_1\cdot\E[\Lambda_{\R}(v^\diamond)]\\
	&\geq\Lambda_{\R}(v^\diamond)-{\varepsilon}_1\cdot\E[\Lambda_{\R}(v^\diamond)]\\
	&\geq(1-\varepsilon_1)\cdot\E[\Lambda_{\R}(v^\diamond)]-{\varepsilon}_1\cdot\E[\Lambda_{\R}(v^\diamond)]\\
	&=(1-\hat{\varepsilon})\cdot\E[\Lambda_{\R}(v^\diamond)],\\
	&=(1-\hat{\varepsilon})\cdot\frac{\OPTT_i}{\eta_i}\abs{\R}.
	\end{align*}
	Thus, we have 
	\begin{equation}\label{eqn:eps2}
		{\varepsilon}_2=\sqrt{\frac{a(2+2{\varepsilon}_2/3)}{\E[\Lambda_{\R}(v^\ast)]}}\leq \sqrt{\frac{(2+2{\varepsilon}_2/3)\hat{\varepsilon}^{2}}{12}}<\hat{\varepsilon}/2.
	\end{equation}
	
	In addition, let 
	\begin{equation*}
		\Lambda_l:=(1+{\varepsilon}_2)\cdot\E[\Lambda_{\R}(v^\ast)].
	\end{equation*}
	According to the definition of $\varepsilon_2$, we have
	\begin{equation*}
		\E[\Lambda_{\R}(v^\ast)]=\big(\sqrt{\Lambda_l+{2a}/{9}}-\sqrt{{a}/{2}} \big)^2- {a}/{18}.
	\end{equation*}
	Since $a_1\leq a$, if event $\mathcal{E}_4$ happens (\ie~$\Lambda_\R(v^\ast)\leq \Lambda_l$), then
	\begin{align*}
		&\Lambda^l(v^\ast)-\E[\Lambda_{\R}(v^\ast)]\\
		&\geq \big(\sqrt{\Lambda_\R(v^\ast)+{2a}/{9}}-\sqrt{{a}/{2}} \big)^2- {a}/{18}- \E[\Lambda_{\R}(v^\ast)]\\
		&\geq \Lambda_\R(v^\ast)-\Lambda_l.
	\end{align*}
	As a consequence, if event $\mathcal{E}_3$ also happens, we have
	\begin{equation}\label{eq:cond2}
	\Lambda^l(v^\ast)\geq\Lambda_{\R}(v^\ast)-{\varepsilon}_2\E[\Lambda_{\R}(v^\ast)]\geq\frac{1-2\varepsilon_2}{1-\varepsilon_2}\Lambda_{\R}(v^\ast).
	\end{equation}
	Similarly, let 
	\begin{equation*}
	\Lambda_u:=(1-{\varepsilon}_2)\cdot\E[\Lambda_{\R}(v^\ast)].
	\end{equation*}
	According to the definition of $\varepsilon_2$, we have
	\begin{equation*}
	\E[\Lambda_{\R}(v^\ast)]\geq\big(\sqrt{\Lambda_u+{a}/{2}}+\sqrt{{a}/{2}}\big)^2.
	\end{equation*}
	Since $a_2\leq a$, if event $\mathcal{E}_3$ happens (\ie~$\Lambda_\R(v^\ast)\geq \Lambda_u$), then
	\begin{align*}
		\frac{\Lambda^u(v^\circ)}{\Lambda_u}
		&\leq\frac{\big(\sqrt{\Lambda_\R(v^\ast)+{a}/{2}}+\sqrt{{a}/{2}}\big)^2}{(1-{\varepsilon}_2)\cdot\E[\Lambda_{\R}(v^\ast)]}\\
		&\leq\frac{\Lambda_\R(v^\ast)}{(1-{\varepsilon}_2)\cdot\Lambda_u}.
	\end{align*}
	As a consequence, we have
	\begin{equation}\label{eq:cond3}
	\Lambda^u(v^\circ)\leq\frac{\Lambda_{\R}(v^\ast)}{1-\varepsilon_2}.
	\end{equation}
	Putting it all together of \eqref{eqn:eps2}, \eqref{eq:cond2} and \eqref{eq:cond3}, we have
	\begin{align*}
	\frac{\Lambda^l(v^\ast)}{{\Lambda}^u(v^\circ)}
	&\geq\frac{(1-2\varepsilon_2)\Lambda_{\R}(v^\ast)}{1-\varepsilon_2}\cdot \frac{1-\varepsilon_2}{\Lambda_{\R}(v^\ast)}\geq 1-\hat{\varepsilon}.\label{eq:lowerbound2}
	\end{align*}
	Therefore, when a number of $c\theta^\ast$ \RR-sets are generated, \OPIMF does not stop only if at least one of the events in $\mathcal{E}_1,\mathcal{E}_2,\mathcal{E}_3,\mathcal{E}_4$ does not happen, with probability at most $\delta^c$.
	
	Let $t$ be the first iteration that the number of \RR-sets generated by \OPIMF reaches $\theta^\ast$ such that $2^{t-2}\cdot\theta_\circ<\theta^\ast$ and $2^{t-1}\cdot\theta_\circ\geq\theta^\ast$. From this iteration onward, the expected number of \RR-sets further generated is at most
	\begin{align*}
	\sum_{z\geq t}\theta_\circ\cdot 2^{z-1}\cdot \delta^{2^{z-t}}
	&=2^{t-1}\cdot\theta_\circ\sum_{z=0}2^{z}\cdot \delta^{2^z}\\
	&\leq 2\theta^\ast\sum_{z=0}2^{-2^z+z}\\
	&\leq 2\theta^\ast\sum_{z=0}2^{-z}\\
	&\leq 4\theta^\ast.
	\end{align*}
	The first inequality is due to $2^{t-1}\cdot\theta_\circ<2\theta^\ast$ and $\delta\leq 1/2$, and the second inequality is due to $-2^z+z\leq -z$. If the algorithm stops before the $t$-th iteration, there are at most $\theta^\ast$ random samples generated. Therefore, the expected number of random samples generated is less than $5\theta^\ast$, which is $O\big(\frac{\eta_i\ln n_i}{\varepsilon^2\OPTT_i}\big)$.
	
	Hence, the lemma is proved.
\end{proof}
}

\begin{proof}[Proof of Lemma~\ref{lemma:main-lem}]
	The time complexity of \OPIMF is determined by that for generating \RR-sets. By Wald's equation \cite{Wald_equation_1947}, the expected total time used for generating \RR-sets equals  the expected number of \RR-sets generated, times the expected time used for generating one \RR-set. Thus, according to Lemmas~\ref{lemma:time-ept} and~\ref{lemma:sample-ept}, the expected time complexity of \OPIMF is $O(\frac{m_i+n_i}{\varepsilon^2}\ln{n_i})$.
\end{proof}

\begin{figure*}[!t]
	\centering
	\includegraphics[height=9.5pt]{lgd}\vspace{-0.18in}\\
	\subfloat[NetHEPT]{\includegraphics[width=0.23\linewidth]{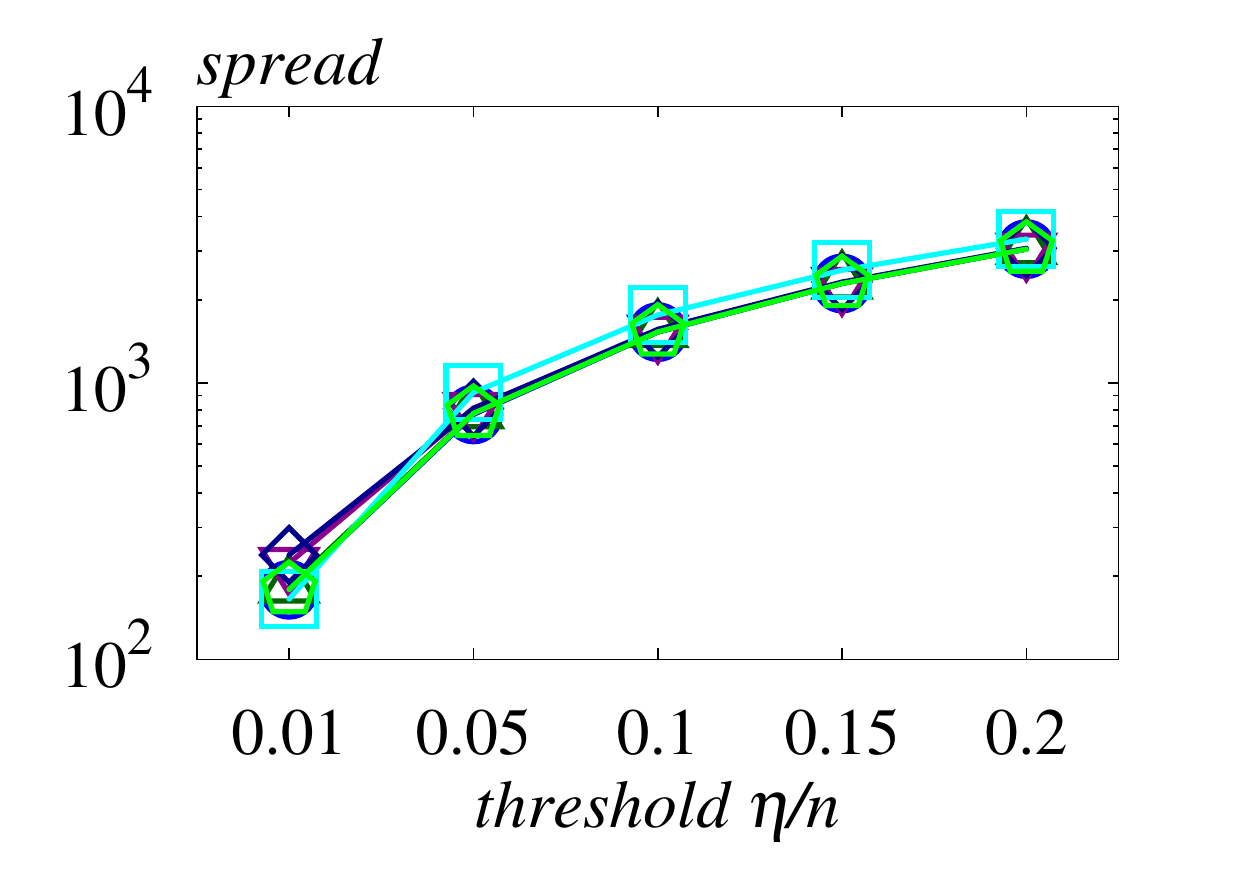}\label{subfig:NetHEPT-spread-ic}}\hfill
	\subfloat[Epinions]{\includegraphics[width=0.23\linewidth]{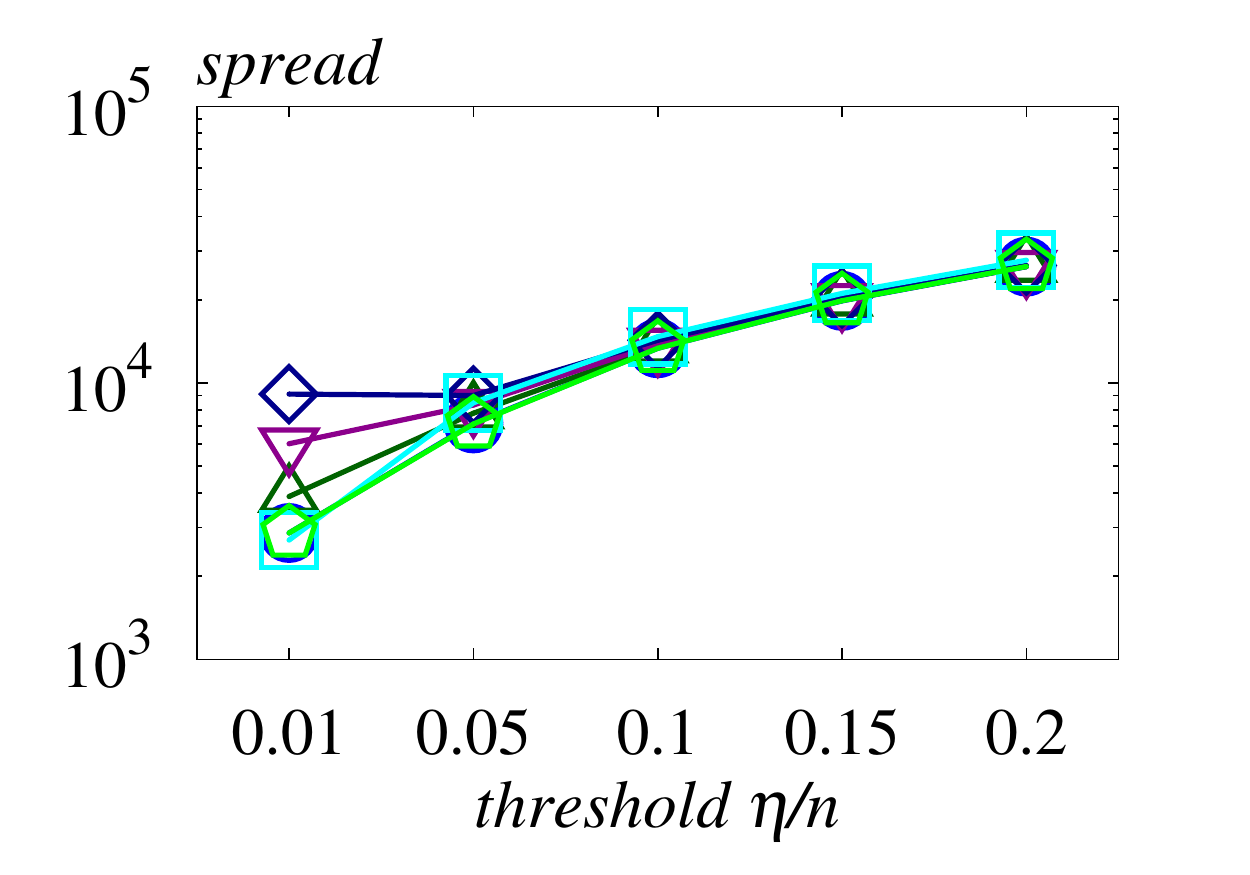}\label{subfig:Epinions-spread-ic}}\hfill
	\subfloat[Youtube]{\includegraphics[width=0.23\linewidth]{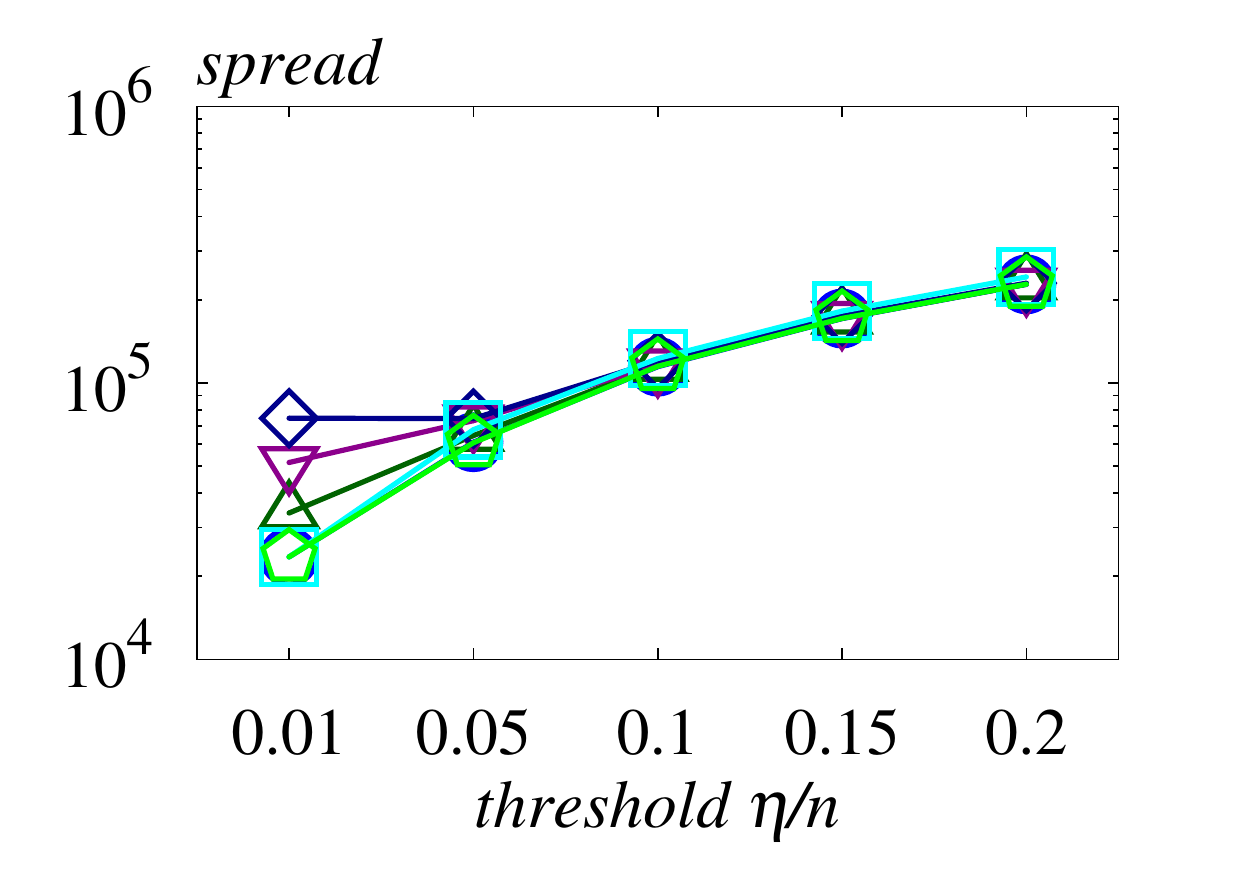}\label{subfig:Youtube-spread-ic}}\hfill
	\subfloat[LiveJournal]{\includegraphics[width=0.23\linewidth]{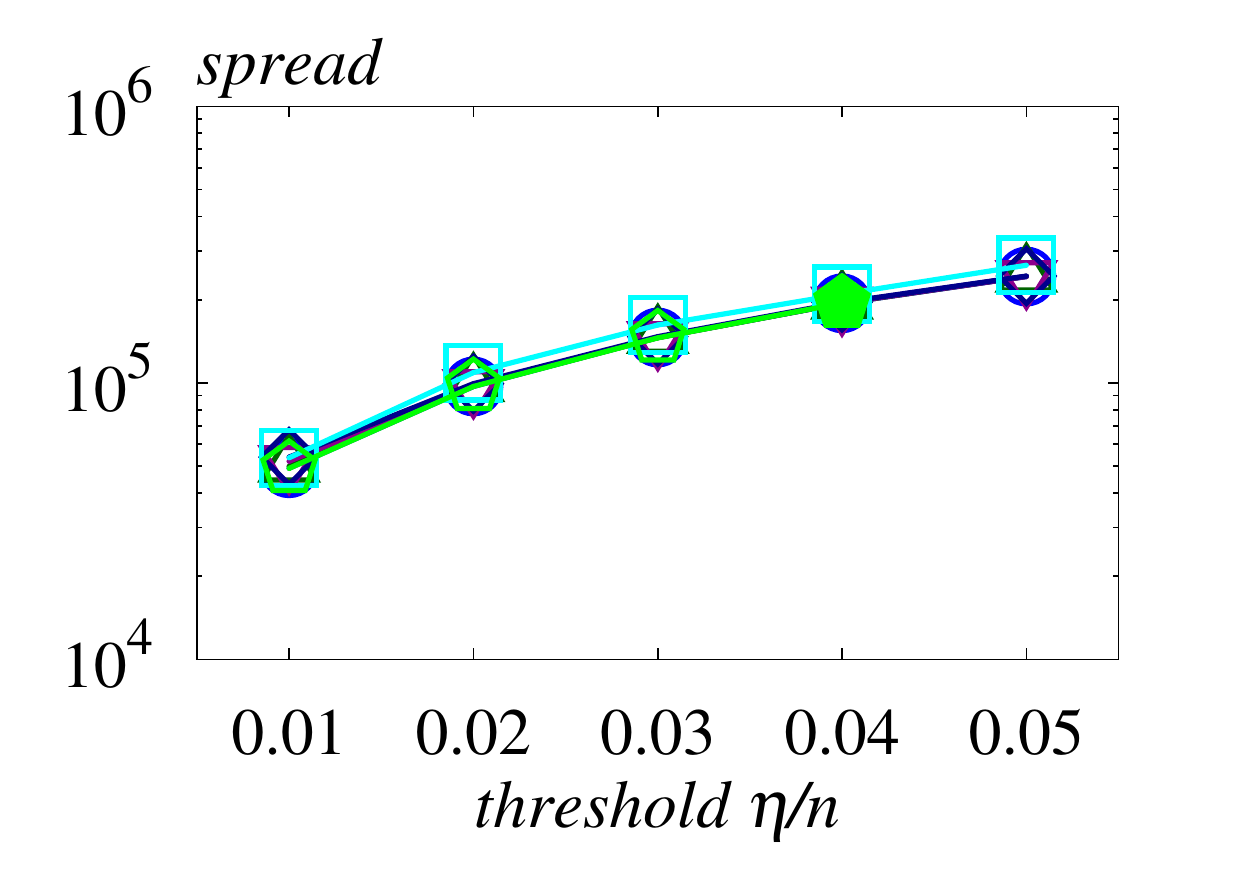}\label{subfig:LiveJournal-spread-ic}}
	\caption{Spread vs. threshold under the IC model.}\label{fig:spread-ic}
	\vspace{-0.9mm}
\end{figure*}

\report{
	\begin{figure*}[!t]
		\centering
		\subfloat[NetHEPT]{\includegraphics[width=0.23\linewidth]{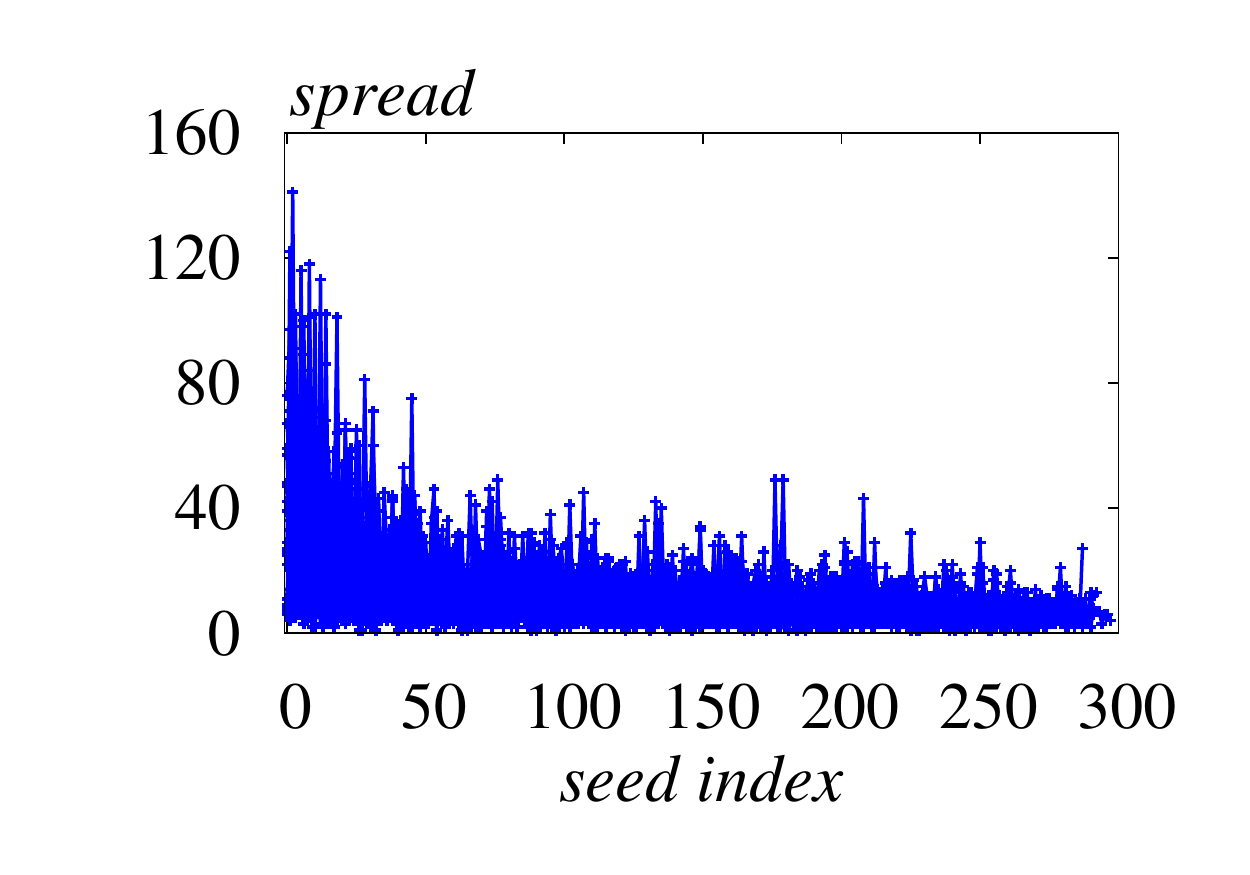}\label{subfig:NetHEPT-ic}}\hfill
		\subfloat[Epinions]{\includegraphics[width=0.23\linewidth]{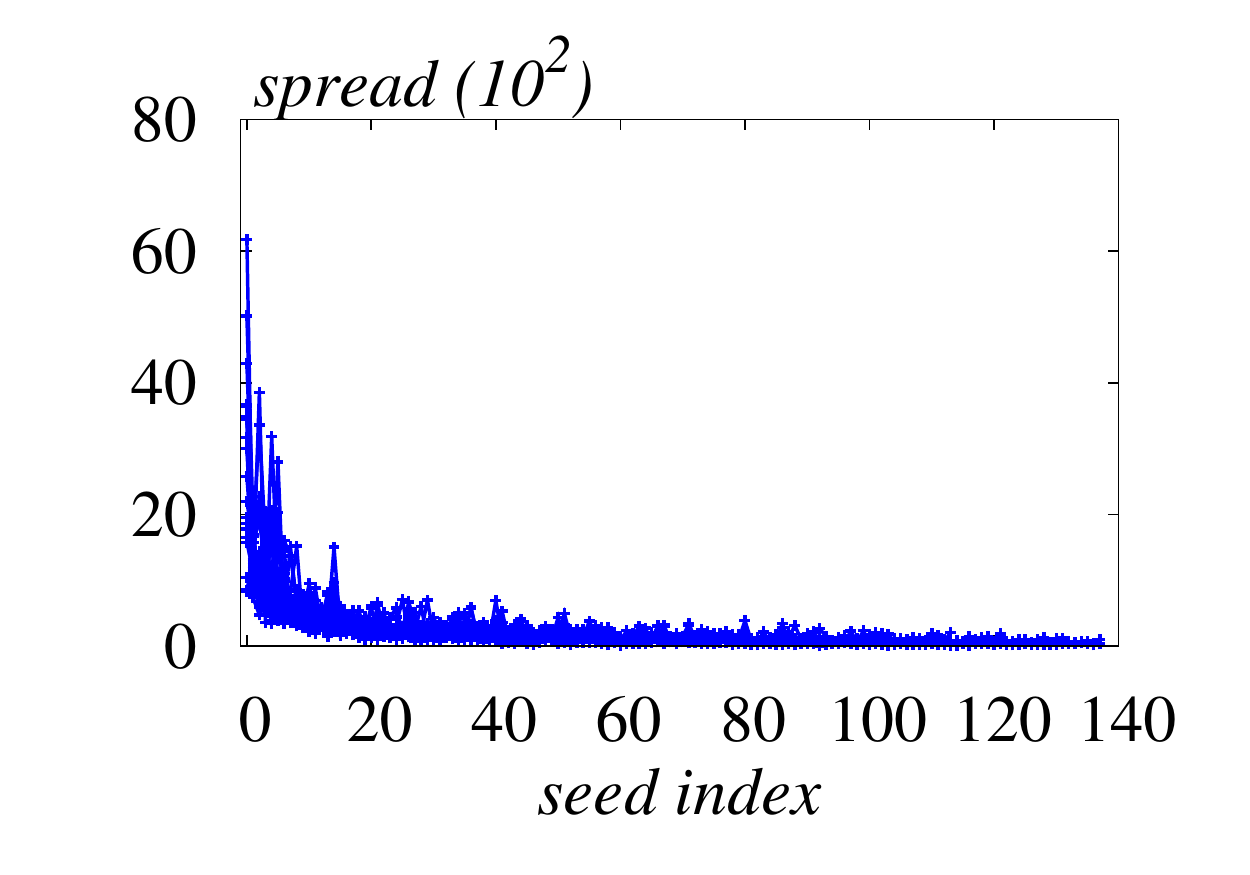}\label{subfig:Epinions-ic}}\hfill
		\subfloat[Youtube]{\includegraphics[width=0.23\linewidth]{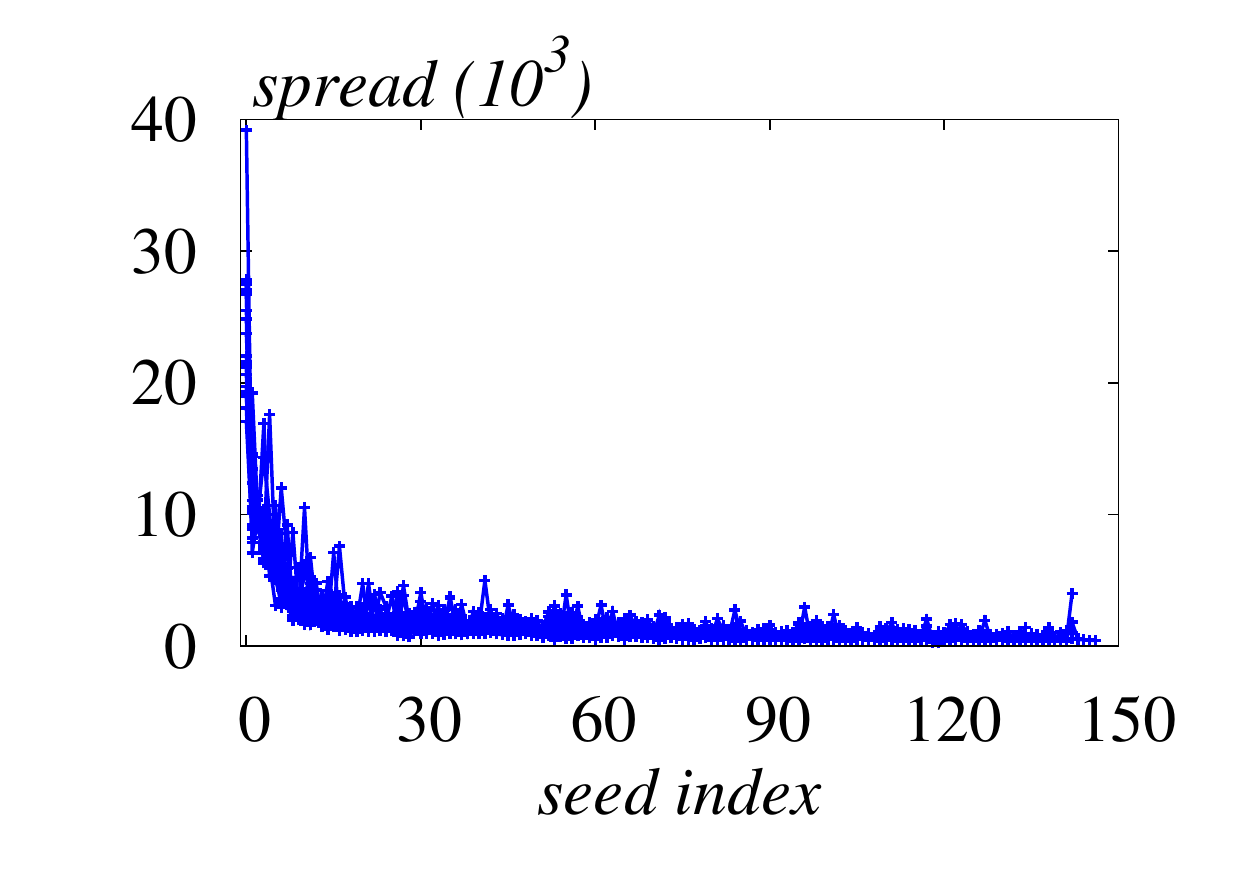}\label{subfig:Youtube-ic}}\hfill
		\subfloat[LiveJournal]{\includegraphics[width=0.23\linewidth]{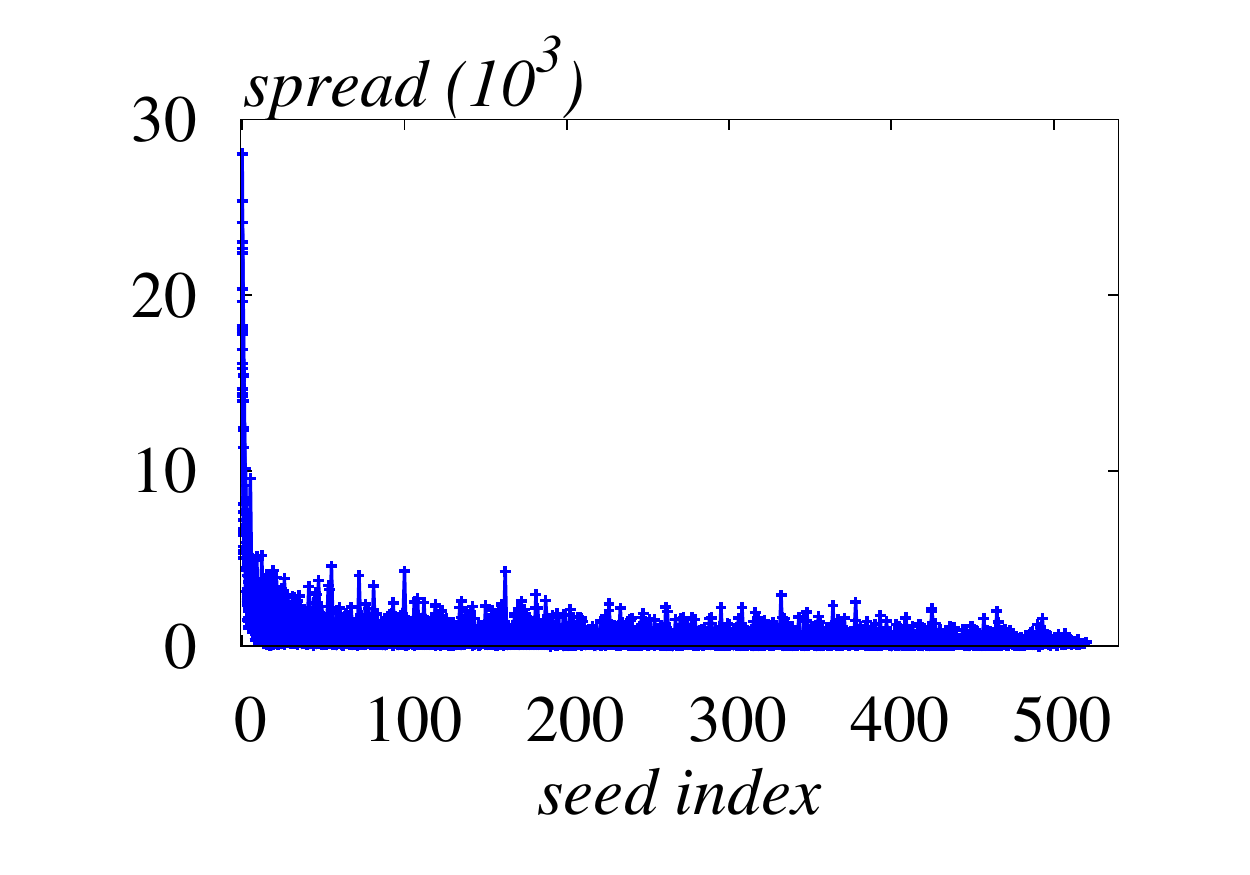}\label{subfig:LiveJournal-ic}}
		\caption{Marginal truncated spread under the IC model.}\label{fig:marginal-ic}
	\end{figure*}
}

\report{
\begin{proof}[Proof of Lemma~\ref{lem:trimb-alpha}]
	Let $S^\ast$ be the seed set returned by the batched policy with $|S^\ast|=b$ and $S^\circ$ be the corresponding optimal seed set in the $i$-th round. Let $\mathcal{E}_b$ be the following event:
	\begin{equation*}
	\mathcal{E}_b(S^\ast)\colon \E[\tilde{\Gamma}(S^\ast \mid S_{i-1})]\geq\rho_b(1-\hat{\varepsilon})\E[\tilde{\Gamma}(S^\circ\mid S_{i-1})].
	\end{equation*}
	Let $S^\ast_t$ be the generalized definition of $v^\ast_t$ in Section~\ref{sec:trim-analysis}. If $S^\ast$ is returned at $T$-th iteration, based on the setting of $T$ and by \cite{Tang_IMM_2015}, we still have 
	\begin{equation}\label{eqn:prob-last-b}
	\Pr[(t=T)\wedge \neg \mathcal{E}_b(S^\ast_t)]\leq \delta/3.
	\end{equation}
	If \OPIMFB stops at the iteration $t<T$, for any node $S\subseteq V_i$ obtained by greedy method with $|S|=b$, we define two events $\mathcal{E}_{b,1}(S)$ and $\mathcal{E}_{b2}(v)$ as
	\begin{align*}
	&\mathcal{E}_{b,1}(S):\E[\Lambda_\R(S)]\geq\big(\sqrt{\Lambda_\R(S)+{2a_1}/{9}}-\sqrt{{a_1}/{2}} \big)^2- {a_1}/{18},\\
	&\mathcal{E}_{b,2}(S):\E[\Lambda_\R(S)]\leq\big(\sqrt{\Lambda_\R(S)/\rho_b+{a_2}/{2}}+\sqrt{{a_2}/{2}} \big)^2.
	\end{align*}
	where $\E[\Lambda_\R(S)]={\abs{\R}}\cdot\E[\tilde{\Gamma}(S\mid S_{i-1})]/{\eta_i}$ is the expected coverage of $S$ in $\R$.
	
	Based on Lemma~\ref{lemma:concentration-ept}, we could have 
	\begin{align}
	\Pr\big[\neg \mathcal{E}_{b,1}(S)\big]\leq\frac{\delta}{3T\binom{n_i}{b}}.
	\end{align}
	Similarly, by union bound for all $\binom{n_i}{b}$ candidates of size-$b$ node set, we could immediately have
	\begin{align}
	\Pr\big[\neg \mathcal{E}_{b,1}(S^\ast_t)\big]\leq {\delta}/{(3T)}.
	\end{align}	
	Let $S^\circ_{\R}$ be the size-$b$ seed set that could cover largest number of \RR-sets in $\R$. Since $S$ is derived by Greedy method from $\R$, by the property of greedy method, we have $\Lambda_\R(S)\ge\rho_b\Lambda_\R(S^\circ_{\R})\ge\rho_b\Lambda_\R(S^\circ).$ Then $\Lambda_\R(S)/\rho_b$ can be taken as the upper bound of $\Lambda_\R(S^\circ)$. Similarly, by Lemma~\ref{lemma:concentration-ept}, we have following equation
	\begin{align}
	\Pr\big[\neg \mathcal{E}_{b,2}(S^\circ)\big]\leq{\delta}/{(3T)}.
	\end{align}
	By following the analysis in Section~\ref{sec:trim-analysis}, we acquire the fact that event $\mathcal{E}_b$ holds with at least $1-\delta$ probability where $\delta=1/n_i$. By Corollary~\ref{corollary:mrr-relative-random-general}, the expected approximation ratio of \OPIMFB is at least \[(1-\delta)\cdot(1-\hat{\varepsilon})\cdot\rho_b\cdot(\ratio)=\rho_b(\ratio)(1-\varepsilon).\] Hence, the lemma is proved.
\end{proof}
}

\section{Discussions on Influence Spread}\label{appendix:influence-spread}
\figurename~\ref{fig:spread-ic} reports the spread of the tested algorithms under the IC model (results under the LT model are similar). For the most parts, all the algorithms achieve a comparable spread on the four datasets. The major differences lie in $\eta/n=0.01$ on {Epinions} and {Youtube}. As observed, \ASME (resp. \TEUC) achieves the largest (resp. smallest) spread among all algorithms. This is because the batch size $b$ is relatively large with regard to the small threshold, owing  to which the spread of the 8-size seed set selected by \ASME significantly overshoots $0.01n$ on  Epinions and Youtube. Another interesting observation is that the spread achieved by \TEUC is slightly larger than each of the other five adaptive algorithms as the threshold  becomes larger (not quite noticeable in the figure). This is because \TEUC selects considerably more seeds than the adaptive algorithms do, resulting in a larger spread at the cost of an excessive number of seeds. This is also supported by the results in Table~\ref{tbl:impro-ratio}.

\report{
\section{Discussions on Marginal Truncated Spread}\label{appendix:marginal-truncated-spread}
To explore the property of the marginal truncated spread, we record the marginal spread of each seed node selected by adaptive algorithms under the $20$ realizations sampled. \figurenames~\ref{fig:marginal-ic} shows the result of each realization with $\eta/n=0.2$ on corresponding datasets (or $\eta/n=0.05$ on the LiveJournal dataset) under the IC model. (The result under the LT model is similar.) In general, the marginal spread diminishes along the index of the seed node, which is consistent with the property of submodularity as expected. Note that the spread fluctuation is due to the randomness of the tested realizations, \ie~in some particular realizations, some seed node selected later may influence more nodes than some seed node selected earlier.
}

\balance
\end{sloppy}

\end{document}